\documentclass[11pt]{article}
\usepackage[dvipsnames]{xcolor}   
\newif\ifcomments
\commentstrue

\usepackage[margin=1in]{geometry}
\usepackage[T1]{fontenc}
\usepackage{lmodern,microtype}
\usepackage{authblk}
\usepackage{amsmath,amssymb,amsthm,mathtools}
\usepackage[normalem]{ulem}
\usepackage{etoolbox}

\makeatletter
\patchcmd{\l@section}{\addvspace{1.0em \@plus\p@}}{\addvspace{0.3em \@plus\p@}}{}{%
  \PackageError{paper}{Could not adjust contents spacing}{Check the definition of l@section.}}
\makeatother

\usepackage{booktabs}
\usepackage{tikz,pgfplots}
\pgfplotsset{compat=1.18}
\usetikzlibrary{positioning,arrows.meta}
\usepackage[hidelinks]{hyperref}
\usepackage{placeins}

\usepackage[nameinlink,noabbrev]{cleveref}

\newtheorem{theorem}{Theorem}[section]
\newtheorem{corollary}[theorem]{Corollary}
\newtheorem{lemma}[theorem]{Lemma}
\newtheorem{proposition}[theorem]{Proposition}
\theoremstyle{definition}
\newtheorem{definition}[theorem]{Definition}

\theoremstyle{remark}
\newtheorem{remark}[theorem]{Remark}

\AddToHook{env/corollary/begin}{\crefalias{theorem}{corollary}}
\AddToHook{env/lemma/begin}{\crefalias{theorem}{lemma}}
\AddToHook{env/proposition/begin}{\crefalias{theorem}{proposition}}
\AddToHook{env/definition/begin}{\crefalias{theorem}{definition}}
\AddToHook{env/remark/begin}{\crefalias{theorem}{remark}}
\AddToHook{env/claim/begin}{\crefalias{theorem}{claim}}

\newcommand{\E}{\mathbb E}
\newcommand{\TV}{\operatorname{TV}}
\newcommand{\ket}[1]{\lvert #1\rangle}
\newcommand{\bra}[1]{\langle #1\rvert}
\newcommand{\C}{\mathbb C}
\newcommand{\Hmin}{H_{\min}}
\newcommand{\R}{\mathbb R}
\newcommand{\Z}{\mathbb Z}
\DeclareMathOperator{\Var}{Var}

\crefname{section}{Section}{Sections}
\crefname{theorem}{Theorem}{Theorems}
\crefname{corollary}{Corollary}{Corollaries}
\crefname{problem}{Problem}{Problems}
\crefname{figure}{Figure}{Figures}
\crefname{table}{Table}{Tables}
\crefname{lemma}{Lemma}{Lemmas}
\crefname{proposition}{Proposition}{Propositions}
\crefname{definition}{Definition}{Definitions}
\crefname{remark}{Remark}{Remarks}
\crefname{claim}{Claim}{Claims}
\crefname{equation}{}{}

\title{Tight Query Lower Bounds for Quantum Sampling, with an Application to Certified Randomness}
\author{Keshav Bhateja\thanks{kbhateja@vt.edu} }
\author{Mehdi Esmaili\thanks{mesmaili@vt.edu} }
\author{Atul Mantri\thanks{atulmantri@vt.edu}}
\affil{Department of Computer Science, Virginia Tech, Blacksburg, VA 24061, USA}
\date{}
\begin{document}
\maketitle
\begin{abstract}
Random circuit sampling experiments, the leading demonstrations of quantum computational advantage, are tested with the linear cross-entropy benchmark (XEB), which scores the returned outcomes by their probabilities under the ideal output distribution. We study what a high XEB score certifies about an untrusted quantum device. In the quantum query model, we show that the score of an honest quantum sampler, the ideal score, plays the role of Tsirelson's bound for this test, both as a limit for efficient devices and as a certificate of randomness. First, exceeding the ideal score by a constant requires $\Omega(N^{1/3})$ queries for the Haar and Fourier oracle ensembles with $N$ outcomes, and this bound is tight; this proves an earlier conjecture for Haar-random states and extends it to Fourier sampling. Second, for Fourier sampling and for Haar-random states accessed through the canonical state-preparation oracle, the $n$-bit output of a device that makes polynomially many queries and scores within $o(1)$ of the ideal score has $n-O(\log n)$ bits of smooth min-entropy, even against an adversary who is entangled with the device and later learns the oracle; this is nearly optimal. As a corollary, sampling the collision distribution, which is proportional to the square of the ideal distribution, has query complexity $\Theta(N^{1/3})$, and a coherent version of the same algorithm approximately prepares quantum Hadamard products of states with uniform amplitude magnitudes. Both lower bounds rest on progress measures that are zero for the honest sampler, bound the excess over the ideal score, and increase only slightly with each query. For Fourier sampling the measure is the norm of the coherent amplitudes for removing pairs in a purification of the random function, combining the polynomial method with the compressed oracle technique; for Haar-random states it is the square root of the average level in an expansion in Dirichlet orthogonal polynomials, controlled by an exact identity for the effect of one observation. We expect these measures to be useful for query lower bounds on other sampling problems.
\end{abstract}

\setcounter{tocdepth}{2}
\tableofcontents
\newpage

\section{Introduction}\label{sec:intro}

Sampling experiments~\cite{AruteEtAl,ZhongEtAl,WuEtAl,MadsenEtAl,MorvanEtAl} are the leading demonstrations of quantum computational advantage; surveys are given in~\cite{HarrowMontanaro,HangleiterEisert,HangleiterStatus}. In the random circuit sampling experiments~\cite{AruteEtAl,WuEtAl,MorvanEtAl}, a device receives a random circuit $C$ on $n$ qubits and is asked to return strings $z\in\{0,1\}^n$ distributed according to the output distribution $p(z)=|\bra zC\ket{0^n}|^2$ of the circuit. The photonic experiments~\cite{ZhongEtAl,MadsenEtAl} instead perform Gaussian boson sampling. Since it is infeasible to check directly whether a few samples follow a distribution over $N=2^n$ outcomes, the random circuit samples are instead scored by the \emph{linear cross-entropy benchmark} (XEB)~\cite[Supplementary Information, eq.~(27)]{AruteEtAl}, a linear version of cross-entropy benchmarking~\cite{BoixoEtAl}: a classical computer computes the ideal probabilities of the returned strings $z_1,\dots,z_k$ and reports the average $(N/k)\sum_ip(z_i)$. We use the unshifted convention $B=1+F_{\mathrm{XEB}}$, where $F_{\mathrm{XEB}}$ is the experimental XEB fidelity. Uniformly random strings have expected score $1$, because $\sum_zp(z)=1$, while an honest sampler has expected score $N\sum_zp(z)^2$, the \emph{ideal score}, which is about $2$ for random circuits. The experiment of~\cite{AruteEtAl} reported a score of about $1.002$, and there is complexity-theoretic evidence that efficient classical algorithms cannot score noticeably above $1$~\cite{AaronsonGunn}.

A high score, however, does not imply that the device sampled from $p$. A device that knew $p$ and always returned the most likely string would score far more than an honest sampler (for a Haar-random state, about $\ln N\approx42$ when $n=60$, compared with about $2$), even though its output would be completely predictable. This gap between scoring well and sampling correctly leads to the question that motivates this work:
\begin{center}
\emph{What does a high cross-entropy score certify about an untrusted quantum device?}
\end{center}
We study this question for the application that has motivated much of the theory, namely certified randomness.

Generating random bits with quantum mechanics is straightforward, but convincing a skeptical user who does not trust the hardware is much harder. In Bell experiments, two devices receive independent uniform public inputs, chosen independently of their shared state with Eve, and cannot communicate during the round. If they win the CHSH game~\cite{CHSH} with probability above the classical limit $3/4$, one party's output bit contains certified entropy~\cite{Colbeck,PironioEtAl}, even conditioned on the inputs and an adversary who is entangled with the devices~\cite{VaziraniVidick,MillerShi}. At Tsirelson's bound $\cos^2(\pi/8)\approx0.854$, which is the largest winning probability of any quantum strategy~\mbox{\cite{Tsirelson}}, either party’s output bit has one bit of min-entropy, even conditioned on the public inputs and Eve’s quantum register. The drawback of this approach is that the two devices must not communicate, a condition that a user who simply receives the bits cannot verify.

An alternative, proposed in~\cite{AaronsonHung}, is to certify randomness with a single device by reusing the sampling experiment. In this protocol a verifier sends challenge circuits, the device must answer quickly with samples (too quickly for a classical simulation to keep up), and the verifier later checks their XEB score. The adversary, whom we call Eve, may have built the device, may share entanglement with it, and may learn the circuits afterward. Even an ideal sample is not perfectly unpredictable to Eve: since she knows the circuit, she can guess its most likely output, which an honest sampler returns with probability about $(\ln N)/N$. Moreover, a device can fix its first $O(\log n)$ output bits by resampling until its output starts with that many zeros, which changes its score only slightly (\cref{prop:randomness-amplification} analyzes a variant that uses amplitude amplification). Consequently, the most that one can hope to certify from a single sample is $n-O(\log n)$ bits.

The guarantees proved so far remain well short of this limit. In the random oracle model, Aaronson and Hung~\cite[Theorem~7.15 and Corollary~7.16]{AaronsonHung} show that a device making polynomially many queries with expected linear XEB score $1+\gamma$, for a constant $\gamma>0$, has a single-round output state close to one with about $\gamma n$ bits of conditional von Neumann entropy, even against an entangled Eve. The same work asks whether, under a plausible hardness assumption, a single sample can be certified to carry the maximum of $n-O(\log n)$ bits. For a list of strings whose empirical Fourier score is high, $n/40$ bits of min-entropy for the whole list are certified in~\cite{BassirianEtAl} against an adversary without quantum side information.

Before one can ask how much randomness a near-ideal score certifies, one must first know whether an efficient device can exceed the ideal score at all. The reason is that the score is linear in the device's output distribution. If some efficient device scored $2+\varepsilon$, then a new device could run it with probability $1/(1+\varepsilon)$ and otherwise output a fixed string known to Eve. The fixed string has expected score $1$, because its expected ideal probability is $1/N$. This mixture has expected score exactly $2$, essentially the ideal score, yet Eve guesses its output with probability at least $\varepsilon/(1+\varepsilon)$, which equals $1/6$ when $\varepsilon=0.2$. Hence, a near-ideal score can certify nearly all of the randomness only if no efficient device exceeds the ideal score by a constant.

The question of how high an efficient quantum device can score was raised in~\cite{AaronsonTalk} and studied in~\cite{Kretschmer} under the name \emph{quantum supremacy Tsirelson inequality}. In the analogy with the CHSH game, the conjectured classical limit $1$ plays the role of $3/4$, and the question is whether the ideal score plays the role of $\cos^2(\pi/8)$. The first answers were obtained in~\cite{Kretschmer}, in the query model. For a Haar-random state, it was shown there that an expected score of $2+\varepsilon$ requires $\Omega(\varepsilon^{5/4}N^{1/4}/\log N)$ queries, so that efficient algorithms cannot exceed the ideal score by $1/\mathrm{poly}(n)$. The same work gave an algorithm that scores $2+\Omega(1)$ with $O(N^{1/3})$ queries and conjectured that $N^{1/3}$ queries are necessary. For Fourier sampling, where the honest sampler scores about $3$, it was shown that no algorithm making a single query does better. The corresponding question for many queries is Problem~12 in~\cite{Aaronson}, which asks how many queries are needed to produce a list of Fourier samples whose empirical score is about $4$.

\paragraph{The query model and the expected score.}
We work in the query model developed in~\cite{Kretschmer}. In this model the circuit is replaced by an oracle drawn once from a known ensemble. The algorithm may query the oracle in superposition, adaptively, and with controlled and inverse queries, and all other computation is free. Theorems in this model are unconditional. Honest sampling requires only one query, whereas exceeding its expected score may require many more. Our bounds quantify this separation. For an oracle with ideal distribution \(p\), the \emph{expected score} of an algorithm \(\mathcal A\) that outputs a string \(Z\) is
\begin{equation}\label{eq:intro-score}
 B(\mathcal A)=N\,\E\bigl[p(Z)\bigr],
\end{equation}
where the expectation is over the oracle and the randomness of the algorithm. We consider the following two ensembles.

\emph{Fourier sampling.} Here the oracle is a uniformly random function $f:\{0,1\}^n\to\{\pm1\}$, accessed through the value oracle $\ket{x,b}\mapsto\ket{x,b\oplus(1-f(x))/2}$. A single query suffices to sample from a distribution determined by $f$: apply Hadamard gates to $\ket{0^n}$, query the oracle with the target qubit in the state $\ket-$, which turns the bit flip into the sign $f(x)$ and so multiplies each $\ket x$ by $f(x)$, apply Hadamard gates again, and measure~\cite{AaronsonChen}. The outcome is $z$ with probability
\begin{equation}\label{eq:intro-fourier}
 p_f(z)=\widehat f(z)^2,\qquad \widehat f(z)=\frac1N\sum_{x\in\{0,1\}^n}(-1)^{z\cdot x}f(x),
\end{equation}
and these probabilities sum to $1$ by Parseval's identity. Each $\sqrt N\widehat f(z)$ is a normalized sum of $N$ independent random signs, and hence nearly a standard Gaussian, so the honest sampler scores its fourth moment, $N\,\E\sum_zp_f(z)^2=3-2/N$.

\emph{Haar-random states.} Here the oracle hides a uniformly random unit vector $\ket\phi=\sum_z\alpha_z\ket z$ in $\C^N$, and $p_\phi(z)=|\alpha_z|^2$. The \emph{canonical state-preparation oracle} $O_\phi$~\cite{AmbainisRosmanisUnruh} swaps $\ket\phi$ with a fixed marker state $\ket\bot$ that is orthogonal to all basis states $\ket z$, and acts as the identity on the orthogonal complement of $\ket\phi$ and $\ket\bot$. With \emph{Haar-random unitary access}, the algorithm may instead apply, to any state, a Haar-random unitary $U$ with $U\ket{0^n}=\ket\phi$, its inverse, and controlled versions of both, as it could if it held the circuit itself. This form of access is at least as strong, since $O_\phi$ can be built from controlled applications of $U$ and $U^\dagger$ (\cref{sec:prelim-notation}), so lower bounds against it are more meaningful. The vector $(p_\phi(z))_z$ is uniformly distributed on the probability simplex, so each $Np_\phi(z)$ is nearly an exponential random variable, and the honest sampler scores its second moment, $N\,\E\sum_zp_\phi(z)^2=2N/(N+1)$.

\subsection{Our results}\label{sec:intro-results}

Our results answer both questions in the query model, namely how high an efficient device can score and how much randomness a near-ideal score certifies, and the answers follow the CHSH analogy closely. For algorithms that make $o(N^{1/3})$ queries, the ideal score behaves like a Tsirelson bound: no such algorithm exceeds it by more than $o(1)$, while $O(N^{1/3})$ queries suffice to exceed it by a constant. Moreover, for Fourier sampling and for Haar-random states accessed through $O_\phi$, a device that makes polynomially many queries and comes within $o(1)$ of the ideal score has $n-O(\log n)$ bits of smooth min-entropy with smoothing $o(1)$, in analogy with the perfectly random bit at Tsirelson's bound. \Cref{fig:landscape,tab:results} summarize the first result.

\FloatBarrier
\begin{figure}[htbp]
\centering
\begin{tikzpicture}
\begin{axis}[
  width=0.8\textwidth, height=4.8cm,
  xmin=-0.035, xmax=0.5, ymin=0.75, ymax=3.75,
  xtick={0,0.25,0.3333,0.5},
  xticklabels={$1$,$N^{1/4}$,$N^{1/3}$,$N^{1/2}$},
  ytick={1,2,3},
  xlabel={number of queries $q$ (logarithmic scale)},
  ylabel={expected score},
  axis lines=left, clip=false,
  tick label style={font=\small}, label style={font=\small},
]
\fill[black!7] (axis cs:0.25,0.75) rectangle (axis cs:0.3333,3.75);
\node[font=\footnotesize,align=center,text=black!70] at (axis cs:0.2917,1.2) {open before\\[-1pt]this work};
\addplot[domain=0:0.3365, samples=300, very thick] {2 + 0.6*(2^(40*(3*x-1)/2) + 2^(40*(3*x-1)))};
\node[font=\footnotesize,anchor=south west] at (axis cs:0.03,2.08) {upper bound of \cref{thm:intro-xeb}};
\addplot[dashed, thick, black!55] coordinates {(0,2) (0.3333,2)};
\addplot[dashed, thick, black!55] coordinates {(0.3333,3) (0.5,3)};
\addplot[only marks, mark=*, mark size=2.2pt] coordinates {(0,2) (0.3333,3)};
\addplot[only marks, mark=o, mark size=2.2pt] coordinates {(-0.035,1)};
\node[font=\footnotesize,anchor=north west] at (axis cs:0.0,1.95) {honest sampler};
\node[font=\footnotesize,anchor=west] at (axis cs:0.345,3.18) {squaring algorithm};
\node[font=\footnotesize,anchor=west] at (axis cs:-0.025,1) {uniform output, no queries};
\end{axis}
\end{tikzpicture}
\caption[Expected XEB scores and query complexity]{Expected scores for a Haar-random state as a function of the number of queries, drawn schematically on a logarithmic scale. The solid curve illustrates the scaling form of the upper bound of \cref{thm:intro-xeb}, and the dashed segments mark the scores of the honest sampler, $2N/(N+1)\approx2$ with one query, and of the squaring algorithm of part~3 of \cref{thm:intro-xeb}, $3-o(1)$ with $O(N^{1/3})$ queries; these are lower bounds on the best score, which is not determined beyond $N^{1/3}$ queries. No algorithm exceeds $2+o(1)$ with $o(N^{1/3})$ queries; at $N^{1/4}$ queries, where the previous lower bound stopped, the upper bound is within $O(N^{-1/8})$ of the ideal score, and the shaded range between $N^{1/4}$ and $N^{1/3}$ was left open by the previous bounds. For Fourier sampling the bounds have the same form, with the scores $2$ and $3$ replaced by $3$ and $5$. The uniform output, which makes no query, is drawn just left of $q=1$, since a logarithmic axis has no zero. The axis extends to $N^{1/2}$ queries, the number of measured copies at which a repeated outcome typically first appears.}
\label{fig:landscape}
\end{figure}
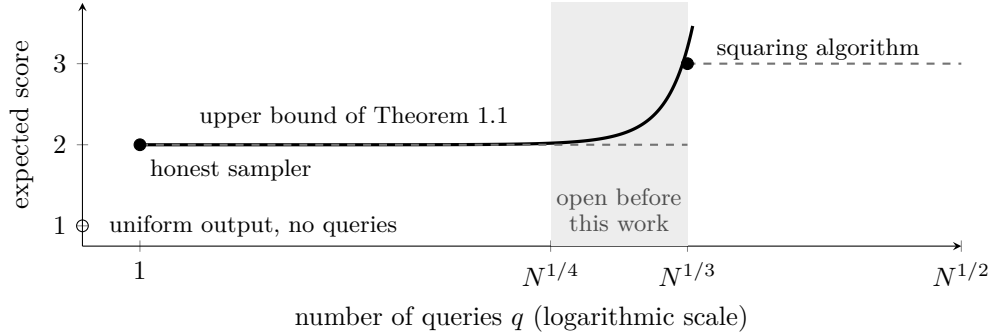

\begin{table}[t]
\centering
\renewcommand{\arraystretch}{1.15}
\begin{tabular}{@{}lcccc@{}}
\toprule
& \multicolumn{2}{c}{Previous work~\cite{Kretschmer}} & \multicolumn{2}{c}{This work}\\
\cmidrule(lr){2-3}\cmidrule(l){4-5}
Task & lower bound & upper bound & lower bound & upper bound\\
\midrule
Haar-random state, score $2+\varepsilon$ & $\widetilde\Omega(N^{1/4})$ & $O(N^{1/3})$ & $\Omega(N^{1/3})$ & $O(N^{1/3})$\\
Fourier sampling, score $3+\varepsilon$ & $>1$ & --- & $\Omega(N^{1/3})$ & $O(N^{1/3})$\\
Collision distribution, error $\nu$ & --- & --- & $\Omega(N^{1/3})$ & $O(N^{1/3})$\\
\bottomrule
\end{tabular}
\caption{The number of queries needed to reach an expected score of $2+\varepsilon$ for a Haar-random state or $3+\varepsilon$ for Fourier sampling, where $\varepsilon>0$ is a fixed constant, and to sample the collision distribution within a fixed total variation error $\nu$ (\cref{cor:collision}), for distributions over $N=2^n$ outcomes. The upper bounds of this work come from the squaring algorithm and need $\varepsilon<1$ and $\varepsilon<2$, respectively, while the previous Haar upper bound holds for the constant $\varepsilon$ attained by the algorithm of~\cite{Kretschmer}. The lower bounds hold for all $0<\varepsilon\le1$ and $0<\varepsilon\le4$, respectively, and the admissible range of $\nu$ is stated in \cref{cor:collision}. The previous Fourier bound is that one query does not suffice~\cite{Kretschmer}. For Haar-random states, every bound holds both for the canonical oracle $O_\phi$ and for Haar-random unitary access, the lower bound in the latter case once $N$ is large enough in terms of $\varepsilon$ (\cref{thm:haar-unitary}).}
\label{tab:results}
\end{table}

\begin{theorem}[Query complexity of exceeding the ideal score; informal]\label{thm:intro-xeb}
Let $\mathcal A$ be a quantum algorithm that makes at most $q$ queries.
\begin{enumerate}
\item For Fourier sampling, $B(\mathcal A)\le3+O(q^{3/2}/\sqrt N+q^3/N)$. In particular, an expected score of $3+\varepsilon$, for $0<\varepsilon\le4$, requires $q\ge(3\varepsilon/16)^{2/3}N^{1/3}$.
\item For a Haar-random state, $B(\mathcal A)\le2N/(N+1)+O(q^{3/2}/\sqrt N+q^3/N)$ with the canonical oracle $O_\phi$, and under Haar-random unitary access the same bound holds with an additional term $O(q\log N/\sqrt N)$. Hence an expected score of $2+\varepsilon$ with $0<\varepsilon\le1$ requires $\Omega(\varepsilon^{2/3}N^{1/3})$ queries with the canonical oracle, and also with Haar-random unitary access once $N$ is large enough in terms of $\varepsilon$.
\item A variant of the $O(N^{1/3})$-query algorithm of~\cite{Kretschmer}, which we call the \emph{squaring algorithm}, makes $O(N^{1/3})$ queries and has expected score at least $5-o(1)$ for Fourier sampling and at least $3-o(1)$ for either form of Haar access.
\end{enumerate}
Consequently, for every fixed $b\in(3,5)$ for Fourier sampling, and every fixed $b\in(2,3)$ for either form of Haar access, the query complexity of achieving $B(\mathcal A)\ge b$ is $\Theta(N^{1/3})$, where the constant in the lower bound depends on $b$.
\end{theorem}

Thus the honest sampler, which makes a single query, remains optimal up to $o(1)$ until an algorithm makes about $N^{1/3}$ queries. For example, at $N^{1/4}$ queries, a number for which the previous lower bound still allows a constant advantage, the score stays within $O(N^{-1/8})$ of the ideal score, and an expected Fourier score of $4$ requires more than $0.32\,N^{1/3}$ queries, which is more than $3\times10^5$ for $n=60$. For Haar-random states, the theorem confirms the conjecture of~\cite{Kretschmer}, as well as the expectation stated there that the $O(N^{1/3})$-query algorithm of that work can be made to score $3-o(1)$; we prove this for the variant described below. For Fourier sampling, it settles the expected-score version of Problem~12 mentioned in~\cite{Aaronson}. 
Outputting a uniformly random entry of a list uses no further queries, and the expected score of that entry is the expected empirical score of the list. If the empirical score is at least $3.9$ with probability $0.9$, then, since rewards are nonnegative, this expected score is at least $0.9\cdot3.9>3.5$, so producing such a list takes more than $0.2\,N^{1/3}$ queries. In the other direction, the empirical score of independent runs of the squaring algorithm exceeds any fixed score below $5$, or below $3$ for Haar-random states, with high probability (\cref{cor:empirical}).

The squaring algorithm explains where the exponent $1/3$ comes from. It draws $m$ ideal samples and stores the outcomes. The stored set has total probability of order $m/N$, and amplitude amplification~\cite{BrassardHoyerMoscaTapp} raises to nearly one the probability that one more ideal sample lies in it. This costs $m+O(\sqrt{N/m})$ queries, which is $O(N^{1/3})$ for $m=N^{1/3}$. Likely outcomes are stored more often, and the amplified sample, which is distributed according to $p$ conditioned on lying in the stored set, favors them again. As a result, the output $z$ has probability nearly proportional to $p(z)^2$, which explains the name. A sampler whose output distribution is proportional to $p(z)^2$ scores $N\sum_zp(z)^3/\sum_zp(z)^2$, a ratio of consecutive moments, which is $15/3=5$ for the nearly Gaussian Fourier coefficients and $6/2=3$ for the nearly exponential Haar probabilities. To reach these scores, our variant runs a fixed number of amplification steps, tuned to the typical probability of the stored set, as suggested in~\cite{Kretschmer}. The algorithm has the same structure as the quantum algorithm for finding collisions~\cite{BrassardHoyerTapp}, whose query complexity is also $N^{1/3}$~\cite{AaronsonShi}, but we do not know of a reduction between the two problems. Moreover, the gain in score does not come from a repeated outcome among the stored samples, which for $m=N^{1/3}$ are distinct with high probability, but from the coincidence of the amplified sample with one of them.

\paragraph{The lower-bound framework.}
For Haar-random states, we first explain why earlier methods stopped at $N^{1/4}$. The argument of~\cite{Kretschmer} first uses the simulation of~\cite{AmbainisRosmanisUnruh} to replace queries by resource states of the form $a\ket\phi+b\ket\bot$. Haar phase averaging then reduces the score analysis to ordinary copies of $\ket\phi$, with a copy count quadratic in $q$, up to accuracy-dependent and logarithmic factors. A constant advantage over the ideal score requires collisions among the measured copies, which occur at about $\sqrt N$ copies. This yields the earlier $N^{1/4}$ query lower bound, up to logarithmic factors. The reduction thus incurs a quadratic cost before any collision is counted. Avoiding this loss requires an argument about the queries themselves, whose effect on the state is coherent and leaves no classical record of what the algorithm has learned. Instead of simulating queries by copies, we therefore ask what the state of an algorithm must look like if it is to score above the honest sampler, and we track a \emph{progress measure} that is zero for the honest sampler and grows, after $t$ queries, by only $O(\sqrt{(t+1)/N})$ at the next query.

For Fourier sampling, we purify the random choice of function, as in the compressed oracle technique~\cite{Zhandry}. In this representation, the basis states of the purifying register are indexed by subsets of $\{0,1\}^n$, which we call recorded sets. Each query changes a recorded set by at most one point. The number of recorded points is not itself a useful measure, since an algorithm can record many points without scoring above the honest sampler; what matters is how the points are arranged. In the honest sampler, the part of the state with output $z$ records exactly one point in the Hadamard basis state $\ket{\chi_z}=H^{\otimes n}\ket z=N^{-1/2}\sum_x(-1)^{z\cdot x}\ket x$. To score more, an output branch must have a nonzero amplitude for removing a pair with the coefficients of its Hadamard basis state. Our measure is the norm of the coherent amplitudes for removing such pairs. Comparing pair removal before and after a query leaves one other recorded point to remove, with coefficient $O(1/\sqrt N)$ (\cref{sec:overview-fourier}).

For Haar-random states, the natural analogue would count repeated amplitudes such as $\alpha_z^2$ or $|\alpha_z|^2$. However, every term of the normalization $\sum_z|\alpha_z|^2=1$ repeats an amplitude, although the sum is a constant. We instead expand the state in orthogonal polynomials on the probability simplex, an expansion in which the normalization has degree zero. Each term has a nonnegative integer \emph{level} that counts repeated amplitudes after this expansion. The progress measure $\lambda$ is the square root of the average level, so it need not be an integer. To show that a query increases $\lambda$ only slightly, we prove an exact identity for the orthogonal polynomials of the Dirichlet distribution, which describes how their degree changes when one more outcome is observed (\cref{sec:overview-haar}).

\paragraph{Certified randomness at the ideal score.}
We now turn to the randomness that a near-ideal score certifies. Since the answer depends on what Eve may know and do, we first describe the experiment precisely. It models one sample from a round of the protocol of~\cite{AaronsonHung}, with the oracle in place of the challenge circuit.
\begin{enumerate}\setlength{\itemsep}{0pt}
\item The device and Eve share an arbitrary entangled state, of arbitrary finite dimension, independent of the freshly sampled oracle.
\item The oracle is drawn: a random function $f$, or a Haar-random state $\ket\phi$ accessed through $O_\phi$. From this point on, the device and Eve do not communicate.
\item The device makes $q$ queries, acting arbitrarily on its own registers and on its half of the shared state, and outputs a string $Z$.
\item Eve receives the complete classical description $R$ of the oracle (the truth table of $f$, or the vector $\phi$) and measures her register $E$, in a way that may depend on $R$, in order to guess $Z$.
\end{enumerate}
No assumption is made about how the device works. The restriction on communication cannot be removed, since a device that sends $Z$ to Eve certifies nothing, and step~4 gives Eve all the information that the device could have obtained from its queries. Without smoothing, the conditional min-entropy $\Hmin(Z\mid RE)$ equals $-\log_2$ of the best probability with which Eve can guess $Z$~\cite{KoenigRennerSchaffner}. Its smooth version $\Hmin^\eta(Z\mid RE)$ is at least $h$ if some normalized state within purified distance $\eta$ of the actual joint state of $R$, $Z$, and $E$ has conditional min-entropy at least $h$~\cite{TomamichelColbeckRenner}. The purified distance is a fidelity-based distance that lies between the trace distance and the square root of twice the trace distance.\footnote{For Haar-random states, $R$ is a continuous classical register; both quantities extend to this case in the standard way, by conditioning on each value of $R$, and the bound then holds for every finite discretization of $\phi$.} Smooth min-entropy is the quantity needed for randomness extraction, since a seeded extractor that is secure against quantum side information turns such a $Z$ into nearly $h$ bits that are close to uniform even given Eve's register~\cite{TomamichelEtAlHashing}.

\begin{theorem}[Certified randomness at the ideal score; informal]\label{thm:intro-randomness}
Let $\delta\ge0$. Consider a device in this experiment that makes $q\ge1$ queries, and suppose that its expected score is at least $3-2/N-\delta$ for Fourier sampling, or at least $2N/(N+1)-\delta$ for a Haar-random state accessed through $O_\phi$; we call $\delta$ the allowed \emph{deficit}. Then there is an $\eta=O(\sqrt\delta+q^{3/4}N^{-1/4})$ such that, if $\eta\le1/2$,
\begin{align*}
 \Hmin^\eta(Z\mid RE)&\ge n-4\log_2q-\log_2n-O(1)&&\text{for Fourier sampling,}\\
 \Hmin^\eta(Z\mid RE)&\ge n-2\log_2q-\log_2n-O(1)&&\text{for Haar-random states.}
\end{align*}
The constants do not depend on the dimension of the state shared with Eve. In particular, if $q\le n^{O(1)}$ and $\delta=o(1)$, the output has $n-O(\log n)$ bits of smooth min-entropy with smoothing $o(1)$, and the smoothing is negligible if $\delta$ is negligible.
\end{theorem}

In a state within purified distance $\eta$ of the actual one, Eve guesses $Z$ with probability at most $2^{-n+O(\log n)}$, even though she knows the entire oracle and may hold an arbitrarily large quantum memory entangled with the device. More generally, a bound $\Hmin^\eta(Z\mid RE)\ge h$ implies that her guessing probability in the actual state is at most $2^{-h}+\eta$, since trace distance is at most purified distance. For $n=60$, $q=100$, and $\delta=0$, the explicit bounds of \cref{thm:fourier-randomness,thm:haar-randomness} give at least $24$ bits for Fourier sampling with $\eta\approx0.00112$ and at least $39$ bits for Haar-random states with $\eta\approx0.0201$, while an ideal Haar sample carries about $54.6$ bits without smoothing.

The proof shows that, up to a total error $\eta$, the part of the device's state with output $z$ is $K_z\psi_0$, where $\psi_0$ is the initial shared state and $K_z$ is an operator that acts on the device alone and has norm at most $\mathrm{poly}(q)\sqrt{p(z)}$ for every fixed oracle (\cref{sec:overview-fourier-randomness,sec:overview-haar-randomness}). No new progress measure is needed for this step: the structure comes from the terms that the proof of the score bound discards, and the remaining argument uses only the fact that a device with few queries changes slowly as the oracle changes. Since such operators do not act on Eve's register, in this nearby state she guesses $Z$ at most $\mathrm{poly}(q)$ times as often as she could guess an ideal sample. The theorem is not restricted to deficits that tend to zero: for a constant deficit below a threshold it certifies the same number of bits, with a constant smoothing parameter of order $\sqrt\delta$.

Three terms in \cref{thm:intro-randomness} cannot be removed. First, consider a device that outputs a fixed string known to Eve with probability $\mu$ and otherwise runs the honest sampler (\cref{prop:randomness-smoothing}). For Fourier sampling its deficit is $\delta=2(1-1/N)\mu$, but Eve guesses its output with probability at least $\mu$, and every state in which her guessing probability is small is at purified distance at least about $\sqrt\mu$ from the actual one. Hence the term $\sqrt\delta$ in $\eta$ has the right order.

Second, some bound on the number of queries is necessary, as the squaring algorithm of part~3 of \cref{thm:intro-xeb} shows. Running it with probability about $1/2$, and otherwise outputting a fixed string, uses $O(N^{1/3})$ queries and attains the ideal score, yet Eve guesses the output with probability about $1/2$. Correspondingly, the term $q^{3/4}N^{-1/4}$ in $\eta$ becomes a constant when $q$ reaches about $N^{1/3}$. In \cref{fig:randomness-chsh}(b), the first of these devices lies on the upper edge of the darker region, and the second is the cross on the ideal-score line; part~(a) of the figure shows the corresponding picture for the CHSH game, in which the PR box plays the role of the squaring algorithm.

\FloatBarrier
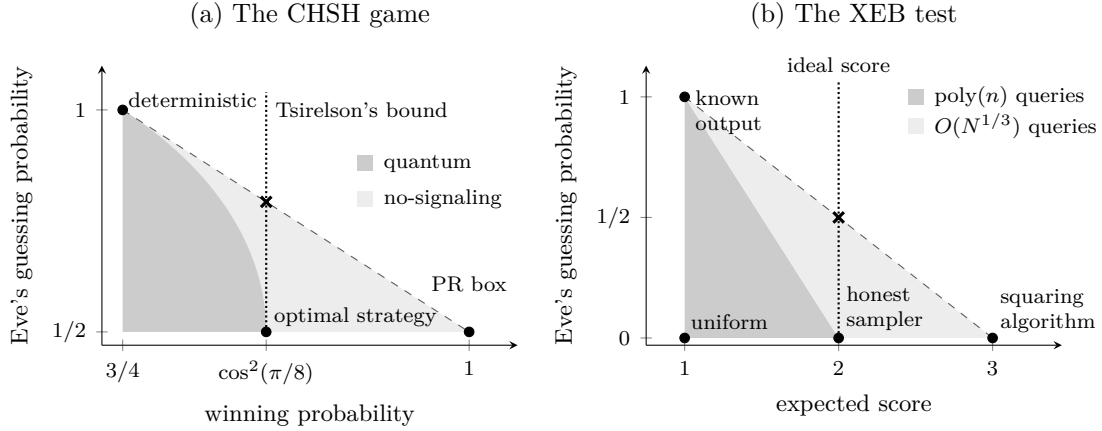
\begin{figure}[htbp]
\centering
\begin{tikzpicture}
\begin{axis}[name=chsh,
  scale only axis, width=5.5cm, height=3.7cm,
  xmin=0.735, xmax=1.035, ymin=0.47, ymax=1.1,
  xtick={0.75,0.853553,1}, xticklabels={$3/4$,$\cos^2(\pi/8)$,$1$},
  ytick={0.5,1}, yticklabels={$1/2$,$1$},
  xlabel={winning probability}, ylabel={Eve's guessing probability},
  title={(a) The CHSH game}, title style={font=\small, yshift=4pt},
  axis lines=left, clip=false,
  tick label style={font=\scriptsize}, label style={font=\footnotesize},
]
\fill[black!7] (axis cs:0.75,0.5) -- (axis cs:0.75,1) -- (axis cs:1,0.5) -- cycle;
\addplot[fill=black!20, draw=none, domain=0.75:0.853553, samples=120] {0.5+0.5*sqrt(max(0,2-4*(2*x-1)^2))} -- (axis cs:0.853553,0.5) -- (axis cs:0.75,0.5) -- cycle;
\draw[dashed, black!65] (axis cs:0.75,1) -- (axis cs:1,0.5);
\draw[densely dotted, thick] (axis cs:0.853553,0.5) -- (axis cs:0.853553,1.04);
\node[font=\scriptsize, anchor=west, inner sep=1pt] at (axis cs:0.858,1.0) {Tsirelson's bound};
\addplot[only marks, mark=*, mark size=1.8pt] coordinates {(0.75,1) (0.853553,0.5) (1,0.5)};
\addplot[only marks, mark=x, mark size=2.8pt, very thick] coordinates {(0.853553,0.792893)};
\node[font=\scriptsize, anchor=south west, inner sep=1pt] at (axis cs:0.752,1.0) {deterministic};
\node[font=\scriptsize, anchor=south west, inner sep=1pt] at (axis cs:0.857,0.505) {optimal strategy};
\node[font=\scriptsize, anchor=south, inner sep=1pt] at (axis cs:1.0,0.585) {PR box};
\node[fill=black!20, minimum width=6pt, minimum height=6pt, inner sep=0pt] (a1) at (axis cs:0.925,0.88) {};
\node[font=\scriptsize, anchor=west] at (a1.east) {quantum};
\node[fill=black!7, minimum width=6pt, minimum height=6pt, inner sep=0pt] (a2) at (axis cs:0.925,0.8) {};
\node[font=\scriptsize, anchor=west] at (a2.east) {no-signaling};
\end{axis}
\begin{axis}[at={(chsh.south east)}, anchor=south west, xshift=1.7cm,
  scale only axis, width=5.5cm, height=3.7cm,
  xmin=0.75, xmax=3.45, ymin=-0.03, ymax=1.13,
  xtick={1,2,3}, ytick={0,0.5,1}, yticklabels={$0$,$1/2$,$1$},
  xlabel={expected score}, ylabel={Eve's guessing probability},
  title={(b) The XEB test}, title style={font=\small, yshift=4pt},
  axis lines=left, clip=false,
  tick label style={font=\scriptsize}, label style={font=\footnotesize},
]
\fill[black!20] (axis cs:1,0) -- (axis cs:2,0) -- (axis cs:1,1) -- cycle;
\fill[black!7] (axis cs:2,0) -- (axis cs:3,0) -- (axis cs:1,1) -- cycle;
\draw[dashed, black!65] (axis cs:1,1) -- (axis cs:3,0);
\draw[densely dotted, thick] (axis cs:2,0) -- (axis cs:2,1.06) node[above, font=\scriptsize] {ideal score};
\addplot[only marks, mark=*, mark size=1.8pt] coordinates {(1,1) (1,0) (2,0) (3,0)};
\addplot[only marks, mark=x, mark size=2.8pt, very thick] coordinates {(2,0.5)};
\node[font=\scriptsize, anchor=north west, align=left, inner sep=1pt] at (axis cs:1.05,1.04) {known\\output};
\node[font=\scriptsize, anchor=south west, inner sep=1pt] at (axis cs:1.02,0.025) {uniform};
\node[font=\scriptsize, anchor=south west, align=left, inner sep=1pt] at (axis cs:2.03,0.02) {honest\\sampler};
\node[font=\scriptsize, anchor=south west, align=left, inner sep=1pt] at (axis cs:3.03,0.01) {squaring\\algorithm};
\node[fill=black!20, minimum width=6pt, minimum height=6pt, inner sep=0pt] (b1) at (axis cs:2.5,1.0) {};
\node[font=\scriptsize, anchor=west] at (b1.east) {$\mathrm{poly}(n)$ queries};
\node[fill=black!7, minimum width=6pt, minimum height=6pt, inner sep=0pt] (b2) at (axis cs:2.5,0.88) {};
\node[font=\scriptsize, anchor=west] at (b2.east) {$O(N^{1/3})$ queries};
\end{axis}
\end{tikzpicture}

\caption[Certified randomness in CHSH and XEB]{Two certification pictures, drawn schematically. (a)~In the CHSH game, quantum strategies win with probability at most $\cos^2(\pi/8)$~\cite{Tsirelson}, and for each winning probability Eve's probability of guessing one party's output bit, conditioned on the public inputs and her quantum register, is at most the curve shown, which reaches $1/2$ at Tsirelson's bound. This follows from the local-output bound of~\cite[eq.~(6)]{PironioEtAl} by conditioning on Eve's measurement outcome and using concavity. Mixing a deterministic strategy with the PR box~\cite{PopescuRohrlich}, a hypothetical device that respects no-signaling and wins with certainty, gives the dashed segment, which meets Tsirelson's bound at a guessing probability of about $0.79$. (b)~The analogous picture for the XEB test with a Haar-random state accessed through $O_\phi$; Eve's guessing probability for the honest sampler, about $(\ln N)/N$, is drawn as $0$. Mixtures of the uniform output, the honest sampler, and an output known to Eve make at most one query and fill the darker triangle; away from the ideal score, the exact region for polynomially many queries is not determined. With polynomially many queries no device scores more than $o(1)$ above the ideal score (\cref{thm:intro-xeb}), and within $o(1)$ of the ideal score the output has $n-O(\log n)$ bits of smooth min-entropy with smoothing $\eta=o(1)$ (\cref{thm:intro-randomness}). Her actual guessing probability is then at most $2^{-n+O(\log n)}+\eta=o(1)$, so the region narrows to $0$ at the ideal score as $n$ grows.
With $O(N^{1/3})$ queries the squaring algorithm plays the role of the PR box: mixing it with an output known to Eve attains the ideal score while Eve guesses with probability about $1/2$, which is the cross. For Fourier sampling the scores $2$ and $3$ become $3$ and $5$.}
\label{fig:randomness-chsh}
\end{figure}

Third, neither loss, $2\log_2q$ for Haar-random states or $4\log_2q$ for Fourier sampling, can be reduced below $2\log_2q-O(1)$. Consider the device that makes $q$ queries and uses amplitude amplification to raise the probability that one ideal sample lies in a fixed set of about $N/q^2$ strings (\cref{prop:randomness-amplification}). It loses only $O(q^{-2}+q^4N^{-1/2})$ in score, but its output lies in this known set with probability at least $3/4$, so its smooth min-entropy is at most $n-2\log_2q+O(1)$ for every smoothing parameter $\eta\le1/2$. Thus the coefficient $2$ of $\log_2q$ in the Haar bound cannot be lowered, the Haar bound is within $\log_2n+O(1)$ bits of this upper bound, and the correct coefficient of $\log_2q$ for Fourier sampling lies between $2$ and $4$.

At near-ideal scores, \cref{thm:intro-randomness} answers the oracle version of the question raised in~\cite{AaronsonHung}, unconditionally and with $n-O(\log n)$ bits, the most that can be certified. It goes beyond earlier results in three respects. First, it certifies smooth min-entropy for a single sample, which an extractor can use directly. Against an entangled Eve, the earlier single-round guarantee is about $\gamma n$ bits of von Neumann entropy in a nearby state~\cite[Theorem~7.15 and Corollary~7.16]{AaronsonHung}; the multi-round guarantee uses entropy accumulation~\cite{DupuisFawziRenner}. For lists of Fourier samples without quantum side information, the earlier guarantee is $n/40$ bits for the whole list~\cite{BassirianEtAl}. Second, in our setting Eve learns the oracle itself. In the earlier model, the device queries $C'=VCW$, where $V$ applies independent random phases and $W$ is Haar-random on the subspace orthogonal to $\ket{0^n}$ and fixes $\ket{0^n}$. These transformations preserve the output distribution of $C$. Eve learns $C$, while $V$ and $W$ remain secret; after this averaging, conditioning on $C$ is equivalent to conditioning on its output distribution~\cite[Theorem~7.1]{AaronsonHung}. We have no corresponding result for an Eve who learns a full Haar-random unitary. Third, it covers Fourier sampling against an entangled Eve, which was previously analyzed only for devices that behave as the honest sampler~\cite{AaronsonHung}. In the other direction, the earlier results apply to any constant advantage over the score $1$ of a uniformly random output and to full unitary access, and they include multi-round protocols with finite tests and, under computational assumptions, randomness expansion. The two sets of results are therefore complementary: the earlier ones certify a constant fraction of the randomness from any violation of the classical bound, and ours certify all but $O(\log n)$ bits of it within $o(1)$ of the ideal score.

\paragraph{Limitations of the randomness result.}
\Cref{thm:intro-randomness} concerns a single oracle experiment at nearly the ideal score, and several steps separate it from a deployed protocol. Current devices operate far from this regime: the experiment of~\cite{AruteEtAl} scored about $1.002$, against an ideal score of about $2$, whereas \cref{thm:intro-randomness} needs a deficit small enough that $\eta\le1/2$. For such devices, the guarantees of~\cite{AaronsonHung} are the relevant ones. All of our results are in the oracle model, which is the price of being unconditional. The theorem also assumes a bound on the expected score of the run whose output is certified, whereas a protocol observes the empirical score of finitely many samples, conditions on passing a test, and runs many rounds. In addition, the challenge comes from a truly random oracle rather than from a short seed, so we claim no randomness expansion. Finally, for Haar-random states we need the canonical oracle $O_\phi$ rather than unitary access (\cref{sec:overview-haar-randomness}).

\paragraph{The collision distribution, heavy outputs, and Hadamard products.}
The output distribution of the squaring algorithm is close to the \emph{collision distribution} $r_p(z)=p(z)^2/\sum_yp(y)^2$, that is, the distribution of the common value of two independent samples from $p$, conditioned on their being equal. The squaring algorithm samples it to within average total variation distance $O(N^{-1/3})$ (\cref{thm:output-law}), where the total variation distance of two distributions $a$ and $b$ is $\TV(a,b)=(1/2)\sum_z|a(z)-b(z)|$.

A small error in total variation causes only a small loss in score, although the reward $Np(z)$ of a single outcome can be as large as $N$. Indeed, the score that a sampler with output distribution $a$ loses relative to $r_p$ comes from the outcomes to which $a$ assigns less probability than $r_p$. This missing probability has total mass $\TV(a,r_p)$ and is at most $r_p(z)$ at each outcome $z$, so by the Cauchy--Schwarz inequality a sampler with $\E\,\TV(a,r_p)\le\nu$ loses only $O(\sqrt\nu)$ on average over the oracle. More precisely, it scores at least $5-\sqrt{35\nu}-o(1)$ for Fourier sampling and $3-\sqrt{12\nu}-o(1)$ for Haar-random states, where $5$ and $3$ are the scores of $r_p$, and \cref{thm:intro-xeb} then gives the following.
\begin{corollary}[Sampling the collision distribution]\label{cor:collision}
For every fixed $0<\nu<4/35$ for Fourier sampling, and every fixed $0<\nu<1/12$ for either form of Haar access, every algorithm with no supplied copies whose output distribution $a$ satisfies $\E\,\TV(a,r_p)\le\nu$ makes $\Omega(N^{1/3})$ queries, while the squaring algorithm makes $O(N^{1/3})$ queries and has $\E\,\TV(a,r_p)=O(N^{-1/3})$.
\end{corollary}

The collision distribution also gives a tight bound for heavy output generation~\cite{AaronsonChen}. For Haar-random states, honest sampling produces outcomes above the median with limiting probability $(1+\ln2)/2$, while the squaring algorithm attains $(2+2\ln2+(\ln2)^2)/4$. Every fixed target strictly between these values has query complexity $\Theta(N^{1/3})$ under either form of Haar access without supplied copies (\cref{sec:heavy-outputs,thm:heavy-outputs}). This holds for both the empirical median and the fixed median of the marginal probability law, with probability averaged over the oracle and the algorithm's randomness. Honest sampling already exceeds the original $2/3$ heavy-output threshold; the result concerns a fixed improvement over its heavy probability.

If the $m$ samples are kept coherent rather than measured, the same idea can be used to prepare the quantum Hadamard product of two unknown states $\sum_z\alpha_z\ket z$ and $\sum_z\beta_z\ket z$, that is, the normalized state with amplitudes proportional to $\alpha_z\beta_z$~\cite{HolmesEtAl,HuangHadamard}. We analyze it when all amplitudes have magnitude $N^{-1/2}$. In that case the algorithm keeps its $m$ stored samples of the first state in superposition, amplifies the component in which a fresh sample of the second state agrees with one of them, and then erases the agreeing stored register by a reversible computation instead of measuring it. The result is within trace distance $O(N^{-1/6})$ of the Hadamard product after $O(N^{1/3})$ queries (\cref{thm:hadamard}), whereas amplifying the agreement of one copy of each state uses about $\sqrt N$ queries. For Haar-random inputs we have only a lower bound: since measuring the Hadamard product of $\ket\phi$ with itself samples $r_p$, preparing it to within average trace distance below $1/12$ requires $\Omega(N^{1/3})$ queries (\cref{prop:self-product}).

\subsection{Related work}\label{sec:intro-related}

The complexity-theoretic case for quantum sampling advantage was made in~\cite{AaronsonArkhipov,BremnerJozsaShepherd,BremnerMontanaroShepherd,BoulandEtAl,Movassagh}, and foundations for heavy-output tasks related to XEB, including Fourier sampling as a black-box problem, were developed in~\cite{AaronsonChen}. Fourier sampling with a function of the form $(-1)^{g(x)}$, where $g$ is a polynomial of degree $3$ over $\mathbb F_2$, is the sampling task of a class of instantaneous quantum polynomial-time (IQP) circuits~\cite{BremnerJozsaShepherd,BremnerMontanaroShepherd}; our results concern a uniformly random function and do not address this structured case. Classical algorithms are known to reach nontrivial XEB scores on shallow or noisy circuits~\cite{BarakChouGao,GaoEtAl}, noisy random circuit sampling admits a polynomial-time classical algorithm in the asymptotic regime~\cite{AharonovEtAl}, and tensor-network methods have produced samples for the circuits of~\cite{AruteEtAl} with a comparable score~\cite{PanChenZhang}.

Besides protocols based on Bell tests and on sampling, certified randomness can also be obtained from cryptographic single-device protocols~\cite{BCMVV,MahadevVaziraniVidick}, which are efficiently verifiable but require the device to maintain a superposition while it interacts with the verifier. There are also protocols in the random oracle model based on the proof of quantumness of~\cite{YamakawaZhandry}, such as~\cite{KhuranaRobertsTal,ColadangeloEtAl}, whose min-entropy is conditioned on the oracle but not on quantum side information. Independent and concurrent work~\cite{JainMutrejaRoberts} certifies randomness at the optimal rate in the quantum random oracle model. On the experimental side, a sampling-based protocol was run on a $56$-qubit trapped-ion processor, certifying $71{,}313$ bits of smooth min-entropy against a restricted class of adversaries and under the additional assumptions of that experiment~\cite{LiuTrappedIon}, and a subsequent experiment demonstrated randomness amplification~\cite{LiuRandomness}.

On the technical side, our Fourier analysis combines the polynomial method~\cite{BealsEtAl} with the purified random function of the compressed oracle technique~\cite{Zhandry}. The generic bound on the size of the recorded set does not by itself control the score; we instead track coherent pair removal in the Fourier basis, as explained in \cref{sec:overview-fourier}. Path recording~\cite{MaHuang} plays a similar role for Haar-random unitaries, and a polynomial method for unitary oracles was developed in~\cite{SheYuen}.

\paragraph{Organization of the paper.}
\Cref{sec:overview} explains the main arguments, starting with the simplest case, and \cref{sec:prelim} fixes the model and the notation. \Cref{sec:fourier} proves the Fourier lower bound and the exact one-query identity, and \cref{sec:haar} proves the Haar lower bound, including the extension to Haar-random unitary access. \Cref{sec:squaring} analyzes the squaring algorithm for both ensembles; it proves part~3 of \cref{thm:intro-xeb} and \cref{cor:collision}, and it shows that the empirical score of independent runs rarely falls below any fixed value smaller than its limit
(\cref{cor:empirical}). \Cref{sec:randomness} proves \cref{thm:intro-randomness}, gives a robust form of the exact privacy statement for one phase query (\cref{prop:fr-one-phase}), and constructs the two devices that show which terms of the theorem are necessary (\cref{prop:randomness-smoothing,prop:randomness-amplification}). \Cref{sec:hadamard} treats Hadamard products. Appendix~\ref{app:deferred-proofs} collects the supporting Haar proofs, the moment calculations, and the proofs for supplied copies.

\paragraph*{Acknowledgments and AI disclosure.}
\addcontentsline{toc}{subsection}{Acknowledgments and AI disclosure}

This work was supported by the faculty startup grant of AM at Virginia Tech.

The authors conceived the high-level idea of the framework and developed it to obtain the result for Haar-random states, using Claude Opus 5 to discuss ideas, to create early drafts, and to write earlier versions of some sections. When prompted with these ideas, ChatGPT Astra 6.1 produced the main results for Fourier sampling and suggested certified randomness as an extension of the approach. All sections and results were subsequently revised with the help of both Claude Opus 5.5 and ChatGPT Sol 6.1, and the authors take full responsibility for the correctness of the results.

\section{Technical overview}\label{sec:overview}

In this section we explain the main arguments behind our results and defer complete proofs to later sections. All of our lower bounds follow the same plan. We define a \emph{progress measure} $\lambda$ that is zero for the honest sampler, show that the score can exceed the ideal score by at most $O(\lambda+\lambda^2)$, and show that after $t$ queries the next query increases $\lambda$ by only $O(\sqrt{(t+1)/N})$, so that $\lambda=O(q^{3/2}/\sqrt N)$ after $q$ queries. The randomness theorems are then obtained from the terms that the proof of the score bound discards. Both ensembles follow this plan and differ only in how each step is carried out (\cref{fig:structure}). We begin with a single query to a random Boolean function, for which both the score and the privacy of an optimal device are given by exact identities, and then treat many queries, first for Fourier sampling and then for Haar-random states. The upper bound, the squaring algorithm, was described in \cref{sec:intro-results}; its analysis in \cref{sec:squaring} uses only amplitude amplification and moment estimates for the two ensembles. The body is organized by result rather than by ensemble: \cref{sec:fourier,sec:haar} prove the two lower bounds, \cref{sec:squaring} analyzes the algorithm, and \cref{sec:randomness} proves both randomness theorems.

\subsection{Fourier sampling with one query}
\label{sec:overview-one-query}

Consider first an algorithm that makes a single \emph{phase query} $\ket x\mapsto f(x)\ket x$ to an $n$-qubit address register, without control. Its final state is then linear in $f$, of the form $\Psi(f)=\sum_xf(x)\,v_x$ for fixed vectors $v_x$, and the algorithm outputs $z$ when a final measurement $\{\Pi_z\}$ gives the outcome $z$. Since every $f$ is a linear combination $\sum_w\widehat f(w)\chi_w$ of the characters $\chi_w(x)=(-1)^{w\cdot x}$, such an algorithm is determined by its behavior on single characters. On the input $\chi_w$, the honest sampler outputs $w$ with certainty, which is the Bernstein--Vazirani algorithm~\cite{BernsteinVazirani}. For an arbitrary algorithm of this kind, let $\sigma$ denote its probability of identifying a uniformly random $w$ from one query to $\chi_w$. Its score on a random $f$ is then given exactly by
\begin{equation}
 B(\mathcal A)=1-\frac2N+2\sigma .
 \label{eq:ov-bv}
\end{equation}
For example, an algorithm that ignores the oracle has $\sigma=1/N$ and scores $1$, while the honest sampler has $\sigma=1$. The proof of \eqref{eq:ov-bv} is elementary. The score is $\sum_z\E_f\bigl[Np_f(z)\,\|\Pi_z\Psi(f)\|^2\bigr]$, and both factors are quadratic in $f$: the \emph{reward} and the probability of the output $z$ are
\begin{align*}
 Np_f(z)&=N\widehat f(z)^2=\frac1N\sum_{x,y}\chi_z(x)\chi_z(y)f(x)f(y),\\
 \|\Pi_z\Psi(f)\|^2&=\sum_{x',y'}f(x')f(y')\langle v_{x'},\Pi_zv_{y'}\rangle .
\end{align*}
Each summand is therefore an average of a product of four values of $f$, which vanishes unless the four points match in pairs. Matching $x=y$ and $x'=y'$ contributes $\sum_z\E_f\|\Pi_z\Psi(f)\|^2=1$; each of the two ways of matching across the factors, such as $x=x'$ and $y=y'$, contributes $N^{-1}\sum_z\|\Pi_z\Psi(\chi_z)\|^2=\sigma$; and the case in which all four points coincide, which the three matchings count three times instead of once, contributes $-2/N$.

Since $\sigma\le1$, an algorithm that makes one phase query scores at most $3-2/N$, with equality exactly when it identifies every character. This is the phase-query case of the one-query theorem of~\cite{Kretschmer}. The identity does not hold for a general value query; for example, for $n=1$ an algorithm that reads $f(1)$ by a value query with the target qubit in $\ket0$ identifies every character but scores only $1$. The identity also shows what an optimal device looks like. Suppose that the device shares an arbitrary entangled state $\ket{\psi_0}$ with Eve, as in \cref{thm:intro-randomness}. If $B(\mathcal A)=3-2/N$, then $\sigma=1$, and by linearity the part of the state with output $z$ is $\Pi_z\Psi(f)=\widehat f(z)\,\Psi(\chi_z)$. On the fixed input $\chi_z$ the device applies a fixed unitary to its own registers, which leaves the reduced state $\rho_E$ of Eve unchanged. Hence, if $\rho^f_z$ denotes Eve's subnormalized reduced state on the \emph{output-$z$ branch} $\Pi_z\Psi(f)$, whose trace is the probability of the output $z$, then
\begin{equation}
 \rho^f_z=p_f(z)\,\rho_E
 \qquad\text{for every $f$ and $z$.}
 \label{eq:ov-exact-privacy}
\end{equation}
Thus the output is an exact sample from $p_f$ and is independent of Eve's register, even given $f$. Our main theorems are robust versions of \eqref{eq:ov-bv} and \eqref{eq:ov-exact-privacy} for algorithms that make many queries; for a single phase query, a robust form of \eqref{eq:ov-exact-privacy} itself is given in \cref{prop:fr-one-phase}.

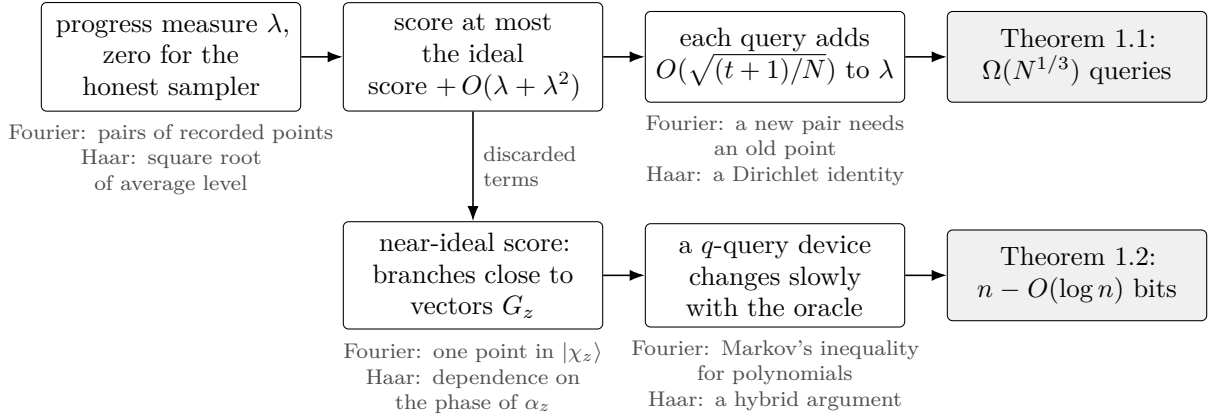
\begin{figure}[t]
\centering
\begin{tikzpicture}[
  box/.style={draw, rounded corners=2pt, align=center, font=\small, inner sep=4pt, text width=3.15cm, minimum height=1.2cm},
  res/.style={box, fill=black!6},
  note/.style={font=\scriptsize, align=center, text width=3.8cm, text=black!70},
  arr/.style={-{Latex[length=2mm]}, semithick}]
\node[box] (pm) {progress measure $\lambda$,\\ zero for the honest sampler};
\node[box, right=0.55cm of pm] (sc) {score at most the ideal\\ score $+\,O(\lambda+\lambda^2)$};
\node[box, right=0.55cm of sc] (gr) {each query adds\\ $O(\sqrt{(t+1)/N})$ to $\lambda$};
\node[res, right=0.55cm of gr] (t1) {\cref{thm:intro-xeb}:\\ $\Omega(N^{1/3})$ queries};
\draw[arr] (pm) -- (sc); \draw[arr] (sc) -- (gr); \draw[arr] (gr) -- (t1);
\node[note, text width=4.3cm, below=1pt of pm] {Fourier: pairs of recorded points\\ Haar: square root of average level};
\node[note, below=1pt of gr] {Fourier: a new pair needs\\ an old point\\ Haar: a Dirichlet identity};
\node[box, below=1.45cm of sc] (st) {near-ideal score:\\ branches close to\\ vectors $G_z$};
\node[box, right=0.55cm of st] (sl) {a $q$-query device\\ changes slowly\\ with the oracle};
\node[res, right=0.55cm of sl] (t2) {\cref{thm:intro-randomness}:\\ $n-O(\log n)$ bits};
\draw[arr] (sc) -- node[right, font=\scriptsize, text=black!70, align=left] {discarded\\ terms} (st);
\draw[arr] (st) -- (sl); \draw[arr] (sl) -- (t2);
\node[note, below=1pt of st] {Fourier: one point in $\ket{\chi_z}$\\ Haar: dependence on\\ the phase of $\alpha_z$};
\node[note, below=1pt of sl] {Fourier: Markov's inequality\\ for polynomials\\ Haar: a hybrid argument};
\end{tikzpicture}
\caption{The structure of the arguments of \cref{sec:overview}. Both lower bounds follow the top row, and both randomness theorems follow the bottom row, which starts from the terms that the proof of the score bound discards. The notes below the boxes indicate how each step is carried out for Fourier sampling and for Haar-random states.}
\label{fig:structure}
\end{figure}

\subsection{The Fourier lower bound}
\label{sec:overview-fourier}

With many queries there is no classical record of what the algorithm has learned, since it queries in superposition and a later query can undo an earlier one. We therefore purify the random function, as in the compressed oracle technique~\cite{Zhandry}. This makes the pairing argument of \cref{sec:overview-one-query} systematic: a product of values of $f$ has nonzero mean only if every point occurs in it an even number of times, and the purifying register records, in superposition, the set of points that occur an odd number of times. This purifying register, from which each query reads the value $f(x)$, starts in the uniform superposition over all functions $f$. Its Fourier basis consists of the vectors $\ket S=2^{-N/2}\sum_f\bigl(\prod_{x\in S}f(x)\bigr)\ket f$, indexed by the sets $S\subseteq\{0,1\}^n$, and multiplying by $f(x)$ maps $\ket S$ to $\ket{S\triangle\{x\}}$. A query to the address $x$ (on the part of the state where the target qubit is $\ket-$) therefore adds $x$ to the \emph{recorded set} $S$ if it is absent and removes it if it is present. Consequently, after $t$ queries the joint state has the form $\sum_{|S|\le t}\ket S\otimes\ket{\Psi_S}$, where the vectors $\ket{\Psi_S}$ are the Fourier coefficients of the state of the algorithm, viewed as a polynomial in $f$ as in the polynomial method~\cite{BealsEtAl}. We write $\Psi$ for the final joint state and $\Psi_z=\Pi_z\Psi$ for its output-$z$ branch.

\paragraph{The score in terms of recorded points.}
Let $a_x$ be the operator that removes the point $x$ from the recorded set, and maps to zero the states in which $x$ is not recorded, so that multiplication by $f(x)$ acts on the register as $a_x+a_x^\dagger$. The reward of the output $z$ is $N\widehat f(z)^2$, and the normalized Fourier coefficient $\sqrt N\widehat f(z)=N^{-1/2}\sum_x\chi_z(x)f(x)$ acts as $b_z+b_z^\dagger$, where the operator $b_z=N^{-1/2}\sum_x\chi_z(x)\,a_x$ removes one recorded point with the amplitudes of the Hadamard basis state $\ket{\chi_z}=H^{\otimes n}\ket z=N^{-1/2}\sum_x\chi_z(x)\ket x$. Since the squared norm of a joint state is the average over $f$ of the squared norm of the algorithm's state, the squared norm of $\sqrt N\widehat f(z)\Psi_z$ is $\E_f\,N\widehat f(z)^2\|\Psi_z(f)\|^2$, the contribution of the output $z$ to the score. Hence $B(\mathcal A)=\sum_z\|(b_z+b_z^\dagger)\Psi_z\|^2$, and expanding the square gives
\begin{equation}
 B(\mathcal A)=1-\frac2N\,\E|S|
 +2\sum_z\|b_z\Psi_z\|^2
 +2\operatorname{Re}\sum_z\langle\Psi_z,b_z^2\Psi_z\rangle ,
 \label{eq:ov-score-expansion}
\end{equation}
where $\E|S|$ is the expected number of recorded points in the final state; the small correction $-(2/N)\E|S|$ appears because a point cannot be recorded twice. For one phase query, exactly one point is recorded, so the interference term vanishes, the quantity $\sum_z\|b_z\Psi_z\|^2$ equals $\sigma$, and \eqref{eq:ov-score-expansion} reduces to \eqref{eq:ov-bv}. The honest sampler makes one query and then applies Hadamard gates, so each of its output-$z$ branches holds exactly one recorded point, in the state $\ket{\chi_z}$, which gives the score $1+2-2/N$.

\paragraph{Pairs of recorded points as the progress measure.}
The recorded set has size at most the number of queries, but a query can either add or remove a point, and the size alone does not control the score. We instead track the amplitudes with which $b_z^2$ removes two recorded points. The score bound below will show that exceeding the ideal score requires $b_z^2\Psi_z\ne0$ for some branch. The relevant exact inequality is
\begin{equation}
 \|b_zg\|^2\le\|g\|^2+\|b_z^2g\|^2
 \qquad\text{for every vector $g$,}
 \label{eq:ov-positivity}
\end{equation}
since the difference between the two sides equals $\|(b_z^\dagger b_z-I)g\|^2$ plus a nonnegative correction. In particular, $\|b_zg\|^2>\|g\|^2$ implies $b_z^2g\ne0$. Since $\sum_z\|b_z\Psi_z\|^2\le\E|S|$ by the orthogonality of the characters, the negative term $-(2/N)\E|S|$ in \eqref{eq:ov-score-expansion} is at most $-(2/N)\sum_z\|b_z\Psi_z\|^2$, which turns the factor $2$ in front of that sum into $2(1-1/N)$. Applying \eqref{eq:ov-positivity} to each branch and bounding the interference term by the Cauchy--Schwarz inequality, we obtain
\begin{equation}
 B(\mathcal A)\le3-\frac2N+2\lambda+2\Bigl(1-\frac1N\Bigr)\lambda^2,
 \qquad
 \lambda=\Bigl(\sum_z\|b_z^2\Psi\|^2\Bigr)^{1/2} .
 \label{eq:ov-payoff}
\end{equation}
Our progress measure $\lambda$ is the norm of the family $(b_z^2\Psi)_z$, whose members coherently remove pairs of recorded points with amplitudes determined by $\chi_z$. It is evaluated on the whole final state rather than branch by branch, because the projections $\Pi_z$ belong to the final measurement, whereas $\lambda$ has to be tracked after every query. Because $\lambda$ vanishes as long as at most one point is recorded, every algorithm that makes one query, of any kind, scores at most $3-2/N$; this recovers the one-query bound mentioned in \cref{sec:overview-one-query} in full generality.

\paragraph{Growth of the progress measure.}
It remains to show that pairs form slowly. Unitaries applied between queries do not act on the recording register, so only queries change $\lambda$. The difference between applying $b_z^2$ before and after a query that toggles $x$ removes one other recorded point, with coefficient $2\chi_z(x)/\sqrt N$ and a sign determined by whether $x$ is recorded. The factor $2$ counts the two orders of removal. By Parseval's identity, the squared norms after removing this other point sum to at most $t$ over all Hadamard basis states. The address components of a query are orthogonal, so the total squared norm of this difference is at most $4t/N$, summed over $z$ and the addresses. The query itself preserves the norm of the existing family $(b_z^2\Psi)_z$, so the triangle inequality gives
\begin{equation}
 \lambda_{t+1}\le\lambda_t+2\sqrt{t/N},
 \qquad\text{hence}\qquad
 \lambda\le\frac2{\sqrt N}\sum_{t=0}^{q-1}\sqrt t\le\frac43\cdot\frac{q^{3/2}}{\sqrt N},
 \label{eq:ov-growth}
\end{equation}
where $\lambda_t$ denotes the progress after $t$ queries. Together with \eqref{eq:ov-payoff}, this proves part~1 of \cref{thm:intro-xeb}: a score of $3+\varepsilon$ with $0<\varepsilon\le4$ forces $\lambda\ge\varepsilon/4$, and hence $q\ge(3\varepsilon/16)^{2/3}N^{1/3}$. The squaring algorithm shows that the order of the growth bound is tight up to a constant factor: its score of nearly $5$ forces $\lambda\ge(\sqrt5-1)/2-o(1)$ by \eqref{eq:ov-payoff}, while it uses $O(N^{1/3})$ queries.

\subsection{Near-optimal randomness from Fourier sampling}
\label{sec:overview-fourier-randomness}

We now explain why a near-ideal score makes the output unpredictable. In the warm-up, the ideal score forced each branch to be $\widehat f(z)$ times the state of the device on the fixed input $\chi_z$, and on a fixed input the device leaves Eve's state unchanged. With many queries this holds only approximately, and we use two tools to make the argument robust. The first is the term $\|(b_z^\dagger b_z-I)g\|^2$ in the gap between the two sides of \eqref{eq:ov-positivity}, which the proof of the score bound discarded. The second is Markov's inequality for polynomials~\cite{BorweinErdelyi}, which states that a polynomial of degree $d$ that is bounded by $1$ on $[-1,1]$ has derivative at most $d^2$ there. We need only the case of the endpoint $1$, and its proof extends without change to operator-valued polynomials (\cref{lem:fr-markov}). Let the device and Eve share a pure state $\ket{\psi_0}$ on registers $A$ and $E$, let $\rho_E$ be the reduced state of Eve, and let the device apply a unitary $V_f$ to $A$, which is the form of every device after the purification of \cref{lem:normal-form}. Suppose that $B(\mathcal A)\ge3-2/N-\delta$.

\paragraph{The structure of the branches at a near-ideal score.}
The quantity $\|(b_z^\dagger b_z-I)\Psi_z\|^2$ measures how far a branch is from satisfying $b_z^\dagger b_z\Psi_z=\Psi_z$, the equation that holds for a branch with exactly one recorded point in the state $\ket{\chi_z}$, as for the honest sampler. If we keep this quantity in the proof of \eqref{eq:ov-payoff}, a score deficit $\delta$ and a small progress measure force
\[
 \sum_z\|(b_z^\dagger b_z-I)\Psi_z\|^2\le\frac{\delta}{2-2/N}+O(\lambda+\lambda^2).
\]
Consequently, the branches are close to the vectors $G_z=(b_z+b_z^\dagger)\,b_z\Psi_z$, obtained from $\Psi_z$ by removing one recorded point with the amplitudes of $\ket{\chi_z}$ and multiplying by $\sqrt N\widehat f(z)$; indeed, $\Psi_z-G_z=(I-b_z^\dagger b_z)\Psi_z-b_z^2\Psi_z$, and both terms are small. In total,
\begin{equation}
 \Bigl(\sum_z\|\Psi_z-G_z\|^2\Bigr)^{1/2}\le\eta,
 \qquad
 \eta=\sqrt\delta+O\bigl(q^{3/4}N^{-1/4}\bigr)\quad\text{for }q\le N^{1/3}.
 \label{eq:ov-eta}
\end{equation}

\paragraph{The rate of change of a $q$-query device.}
For a fixed $f$, removing a recorded point with the amplitudes of $\ket{\chi_z}$ corresponds to differentiating the unique multilinear extension of the state in the normalized direction $\chi_z/\sqrt N$. All derivatives below are taken in the unique multilinear extension of $f\mapsto V_f$. Since $\ket{\psi_0}$ does not depend on $f$, the vector $G_z$ is therefore the initial state acted on by the operator
\[
 K_z=\sqrt N\widehat f(z)\,\Pi_z\,\partial_{\chi_z}V_f,\qquad \partial_{\chi_z}V_f=N^{-1/2}\sum_x\chi_z(x)\,\frac{\partial V_f}{\partial f(x)},
\]
on the device alone, where $\partial_{\chi_z}V_f$ denotes the derivative at $f$ in the direction $\chi_z/\sqrt N$. Because Eve knows $f$, we need a bound on $K_z$ for every $f$, not only on average.

The unitary $V_f$ is a polynomial of degree at most $q$ in the values $f(x)$, and its multilinear extension to the cube $[-1,1]^N$ takes values that are averages of unitaries, so it has norm at most $1$ on the cube. The normalized direction $\chi_z/\sqrt N$ has an entry of magnitude $N^{-1/2}$ at each address. Let $\mathcal I_+$ and $\mathcal I_-$ be the sets of addresses at which $\chi_z$ agrees or disagrees with $f$. Multiplying the values of $f$ on $\mathcal I_+$, or on $\mathcal I_-$, by a factor $t\in[-1,1]$ keeps the point in the cube, and $\partial_{\chi_z}V_f$ is $N^{-1/2}$ times the difference of the derivatives at $t=1$ of the two resulting polynomials of degree at most $q$. Markov's inequality for polynomials bounds each of these derivatives by $q^2$, so $\|\partial_{\chi_z}V_f\|\le2q^2/\sqrt N$, and hence $\|K_z\|^2\le4q^4p_f(z)$. Apart from the score bound and the growth estimate, this is the only step
that uses the assumption that the device makes few queries.

\paragraph{Bounding the knowledge of Eve.}
An operator that acts on the device alone cannot make any measurement outcome of Eve more likely than $\|K_z\|^2$ times its probability on her initial state. Hence, for any initial entanglement and any measurement that Eve performs,
\begin{equation}
 \Pr\nolimits_{G_z}[m]\le4q^4p_f(z)\,\Pr\nolimits_{\rho_E}[m]
 \qquad\text{for every outcome $m$ and all $f$ and $z$,}
 \label{eq:ov-domination}
\end{equation}
where $\Pr_{G_z}[m]$ denotes the unnormalized weight of outcome $m$ on the vector $G_z$, and $\Pr_{\rho_E}[m]$ is its probability on Eve's initial state; this is the robust form of \eqref{eq:ov-exact-privacy}.

The actual branches need not satisfy such a bound, because a device may output a fixed string with small probability, and so we pass to a nearby state. We normalize the vector $\sum_z\ket z\otimes G_z$, in which the first register holds $z$, measure the purifying register of $f$ in the basis $\{\ket f\}$, which gives Eve the function $f$, and discard the device. The resulting state of $f$, $z$, and Eve's register is within purified distance $\eta$ of the actual one. Indeed, the unit vector $\sum_z\ket z\otimes\Psi_z$ is within distance $\eta$ of $\sum_z\ket z\otimes G_z$, so the sine of its angle with the normalized vector is at most $\eta$ (\cref{prop:randomness-assembly}), and measuring or discarding registers does not increase the purified distance~\cite{TomamichelColbeckRenner}. In this state, if $\eta\le1/2$, so that the normalizing factor $\sum_z\|G_z\|^2$ is at least $1/4$, then \eqref{eq:ov-domination} implies that Eve guesses $Z$ with probability at most
\[
 O(q^4)\,\E_f\max_zp_f(z)=O(q^4n/N),
\]
which is $O(q^4)$ times her probability of guessing an ideal sample. Taking logarithms gives the Fourier case of \cref{thm:intro-randomness}.

\subsection{The lower bound for Haar-random states}\label{sec:overview-haar}

For a Haar-random state $\ket\phi=\sum_x\alpha_x\ket x$, we write $p_x=|\alpha_x|^2$. We first consider algorithms with no supplied copies of $\ket\phi$. After $t$ queries to $O_\phi$, the state is a polynomial of total degree at most $2t$ in the amplitudes and their conjugates. More generally, with $s$ supplied copies, its degree is at most $s+2t$ (\cref{lem:haar-degree}). The plan of \cref{sec:overview-fourier} carries over; the main difference lies in the choice of the progress measure.

\paragraph{The normalization obstacle.}
For Fourier sampling, progress was measured by pairs of recorded points. The natural analogue for Haar-random states is a repeated amplitude, as in $\alpha_z^2$ or $\alpha_z\overline{\alpha_z}=p_z$, which is the coherent counterpart of seeing the outcome $z$ twice. However, every term of the normalization $\sum_xp_x=1$ repeats an amplitude, although the sum itself is a constant, and the oracle produces such sums at every query. For example, two queries applied to $\ket\bot$ return $\ket\bot$, yet the second query, expanded in the amplitudes, gives $O_\phi\ket\phi=(\sum_xp_x)\ket\bot+(1-\sum_xp_x)\ket\phi$. We therefore measure the dependence on the probabilities $p_x$ by the degree in an expansion in orthogonal polynomials. Like the Fourier degree of a Boolean function, this degree depends only on the function and not on the way in which it is written, so the normalization has degree zero.

\paragraph{The level and the progress measure.}
Two properties of the Haar measure make this approach work. Write $\alpha_x=\sqrt{p_x}e^{\mathrm i\theta_x}$. The phases $\theta_x$ are independent and uniform, so two terms are orthogonal unless, at every $x$, the number of factors $\alpha_x$ minus the number of factors $\overline{\alpha_x}$ is the same in both terms. Moreover, independently of the phases, the vector $p$ is uniform on the probability simplex, that is, it has the Dirichlet distribution with all parameters equal to $1$. In a term with $k$ unpaired factors $\alpha_x$ (or $\overline{\alpha_x}$), these factors contribute $p_x^k$ to the squared magnitude of the term, which reweights the distribution of $p$ exactly as if $x$ had been observed $k$ times, and the result is again a Dirichlet distribution.

We expand the rest of the term in orthogonal polynomials for that distribution, and we define the \emph{level} of the term as its degree in this expansion plus, for each outcome, the number of unpaired factors at that outcome beyond the first. The first unpaired factor at an outcome is not counted, because every branch $\alpha_z\ket z$ of the honest sampler has one. Thus $\alpha_x\overline{\alpha_y}$ with $x\ne y$, and the normalization, have level zero, while $\alpha_x^2$ and $p_x-1/N$ have level one. Our progress measure is the square root of the average level,
\[
 \lambda=\Bigl(\sum_kk\,\|P_k\Psi\|^2\Bigr)^{1/2},
\]
where $P_k$ denotes the projection onto the part of the state of level $k$. For $s$ copies of $\ket\phi$, the quantity $\lambda^2$ is exactly the expected number of measured copies that repeat an earlier outcome, namely $s(s-1)/(N+s-1)$ (\cref{lem:haar-copies}). The argument therefore also applies to algorithms that receive $s$ copies in addition to making $q$ queries, and shows that exceeding the ideal score by a constant requires $s+q\sqrt{s+q+1}=\Omega(\sqrt N)$ (\cref{cor:haar-birthday}). With $s=0$, the query contribution is of order $q^{3/2}$, which explains the $N^{1/3}$ lower bound.

\paragraph{The score in terms of the level.}
A state of level zero scores at most the ideal score $2N/(N+1)$. Its terms contain each outcome at most once. Since the reward $Np_z$ does not depend on the phases, different terms do not interfere, and each term is correlated with $Np_z$ only through a single factor $\alpha_z$ or $\overline{\alpha_z}$. This factor shifts the distribution of $p_z$ as one observed sample would, which raises the mean of $Np_z$ to at most $2N/(N+1)$. A term of positive level can score more, but only by an amount proportional to its level, because a polynomial of low degree cannot concentrate its weight where $p_z$ is large. Finally, the interference between the part of level zero and the rest involves only the fluctuation of $Np_z$ around its mean, which has variance at most $2$, so by the Cauchy--Schwarz inequality it contributes $O(\lambda)$. Altogether,
\begin{equation}
 B(\mathcal A)\le\frac{2N}{N+1}+O\bigl(\lambda+\lambda^2\bigr).
 \label{eq:haar-extraction}
\end{equation}

\paragraph{The effect of one query on the level.}
Write $O_\phi=I-2\ket{\eta_\phi}\bra{\eta_\phi}$ with $\ket{\eta_\phi}=(\ket\bot-\ket\phi)/\sqrt2$. Returning to the no-copy case, consider a term of total degree $d$ in the amplitudes and their conjugates, so that $d\le2t$ after $t$ queries. Each new factor raises the level by at most one, and it can do so only when it falls on an outcome that already appears in the term; a factor at a fresh outcome does not raise the level. When the new factor does fall on an occupied outcome, either it repeats an unpaired factor of the same kind or it pairs with a conjugate factor, and the two cases need different arguments. In the first case, only an $O((d+1)/N)$ fraction of the weight moves up a level. Indeed, a term of degree $d$ has unpaired factors on at most $d$ outcomes, and a polynomial of degree at most $d$ in $p$ cannot concentrate on a few outcomes, so these outcomes carry only an $O(d/N)$ fraction of its weight (\cref{lem:haar-subset}).

If instead the new factor pairs up with a conjugate factor, it creates $p_x$, and in the step that applies $\bra\phi$ the resulting terms are summed over $x$, which can bring back the normalization. These terms must therefore be bounded together, and for this purpose we prove an identity that concerns only the Dirichlet distribution and its orthogonal polynomials.

Observing one more sample $x$ from $p$ turns a Dirichlet prior into a Dirichlet posterior, and this is how a new factor at $x$ reweights a term. Let $|a|\ge N$ be the sum of the parameters of the prior. We show that, on average over $x$ with probability $a_x/|a|$, an orthogonal polynomial of exact degree $r\ge1$ for the prior has a fraction $r/(|a|+2r-1)\le r/N$ of its prior squared norm in degree $r-1$ for the posterior, and none in lower degrees. The identity follows from integration by parts on the simplex. For background on orthogonal polynomials for Dirichlet distributions, see~\cite{DunklXu,GriffithsSpano}.

By duality, it shows that when $\bra\phi$ is applied to $\sum_x\alpha_xh_x\ket x$, where each $h_x$ has degree at most $r$, the part of the result $\sum_xp_xh_x$ of degree $r+1$ carries at most a fraction $(r+1)/N$ of the weight of the input. Since $\lambda$ is the norm of the state in which the part of level $k$ is multiplied by $\sqrt k$, and $\sqrt{k+1}-\sqrt k\le1$, each of the two steps changes $\lambda$ by at most the norm of the part of the state that it moves up a level plus the norm of the part that it moves down a level, and the second part is bounded by applying the same estimates to the adjoint step (\cref{lem:haar-down}). Together, the two cases show that a query, controlled or not, raises $\lambda$ by $O(\sqrt{(d+1)/N})=O(\sqrt{(t+1)/N})$ after $t$ queries. Hence $\lambda=O(q^{3/2}/\sqrt N)$ after $q$ queries, and \eqref{eq:haar-extraction} gives part~2 of \cref{thm:intro-xeb}.

\paragraph{Extension to Haar-random unitary access.}
For each fixed $\phi$, $q$ queries to a random unitary $U$ with $U\ket{0^n}=\ket\phi$ can be simulated by $2q$ queries to $O_\phi$ with error $O(q/\sqrt N)$~\cite[Theorem~19]{Kretschmer}. This changes the score by at most the error times the largest reward $N\max_zp_z$, whose average is $\ln N+O(1)$. The squaring algorithm only prepares $\ket\phi$ and reflects about it, so it runs unchanged with $U$.

\subsection{Near-optimal randomness for Haar-random states}
\label{sec:overview-haar-randomness}

The three steps of \cref{sec:overview-fourier-randomness} carry over to this setting, with the phase $\theta_z$ of $\alpha_z$ in place of the Fourier coefficient, and the main change is in the second step. First, the part of the output-$z$ branch that does not depend on $\theta_z$ consists of the terms with no unpaired factor at $z$. At level zero these terms score at most $1$, well below the ideal score, because without a factor at $z$ the reweighting does not favor $z$. The part of positive level has weight at most $\lambda^2$, so a near-ideal score forces the branches to equal their $\theta_z$-dependent parts $G_z$ up to a total error $\eta=O(\sqrt\delta+q^{3/4}N^{-1/4})$.

Second, unlike a Fourier coefficient, the phase $\theta_z$ can be changed without leaving the ensemble. Rotating it moves $\ket\phi$ by at most $2\sqrt{p_z}$, so, when the oracle is replaced one query at a time, the unitary of the device changes by $O(q\sqrt{p_z})$. Hence $G_z$, which is the branch minus the average of its copies with the phase $\theta_z$ rotated, is the initial state acted on by an operator on the device of norm $O(q\sqrt{p_z})$; this hybrid argument takes the place of Markov's inequality for polynomials.

Third, as in \eqref{eq:ov-domination}, the unnormalized weight of every outcome of a measurement of Eve on $G_z$ is at most $O(q^2)\,p_z$ times its probability on her initial state. Since $\E\max_zp_z=(\ln N+O(1))/N$, we obtain the Haar case of \cref{thm:intro-randomness}. The factor $q^2$, instead of the factor $q^4$ that Markov's inequality gives, is the reason why the bound loses $2\log_2q$ bits here but $4\log_2q$ bits for Fourier sampling.

The argument does not extend to an Eve who learns a full Haar-random unitary $U$. Our score bound for unitary access averages over the columns of $U$ other than $\phi$, which such an Eve knows, and a device can correlate its output with these columns. For instance, a device that applies $U$ to its half of a maximally entangled state shared with Eve produces a uniform output that Eve, knowing $U$, recovers exactly. That device scores only $1$, but a proof for this setting would have to rule out such correlations at the ideal score.

\section{Preliminaries}\label{sec:prelim}

\subsection{Notation and query algorithms}\label{sec:prelim-notation}

Throughout, $n\ge1$ and $N=2^n$. We identify $\{0,1\}^n$ with $[N]=\{1,\dots,N\}$ when convenient, and we write $z\cdot x$ for the inner product of $z,x\in\{0,1\}^n$ modulo $2$. Logarithms are to base $2$ unless written $\ln$. Inner products $\langle u,v\rangle$ are conjugate-linear in the first argument, $\|v\|$ is the norm of a vector, and $\|T\|$ is the operator norm of an operator $T$. For a positive semidefinite operator $M$, we write $\operatorname{Tr}M$ for its trace. All Hilbert spaces are finite-dimensional, except for the spaces of square-integrable functions on the simplex that
are defined in \cref{sec:prelim-dirichlet} and used in \cref{sec:haar}.

The oracle of an experiment is denoted $\omega$ when we do not need to specify the ensemble; it is a function $f$ in the Fourier ensemble and a state $\phi$ in the Haar ensemble. We write $p_\omega$ for the ideal distribution on $\{0,1\}^n$ that the oracle determines, and $O_\omega$ for the unitary through which it is accessed, namely $O_f$ of \eqref{eq:prelim-value-oracle} or $O_\phi$ of \eqref{eq:prelim-canonical-oracle}.

\paragraph{The Fourier ensemble.}
The oracle is a uniformly random function $f:\{0,1\}^n\to\{-1,1\}$, so that the $N$ values $f(x)$ are independent uniformly random signs. For $z\in\{0,1\}^n$, let $\chi_z(x)=(-1)^{z\cdot x}$ denote the characters, which satisfy $\chi_z(x)\chi_z(y)=\chi_z(x\oplus y)$ and the orthogonality relations
\begin{equation}\label{eq:prelim-walsh}
 \sum_{z}\chi_z(x)\chi_z(y)=N\,\mathbf 1[x=y],\qquad \sum_x\chi_z(x)\chi_w(x)=N\,\mathbf 1[z=w].
\end{equation}
The Fourier coefficients and the ideal distribution are
\begin{equation}\label{eq:prelim-fourier}
 \widehat f(z)=\frac1N\sum_x\chi_z(x)f(x),\qquad p_f(z)=\widehat f(z)^2,
\end{equation}
and $\sum_zp_f(z)=1$ for every $f$ by Parseval's identity, which follows from \eqref{eq:prelim-walsh}.

The algorithm accesses $f$ through the \emph{value oracle}
\begin{equation}\label{eq:prelim-value-oracle}
 O_f\ket{x,b}=\ket{x,\,b\oplus(1-f(x))/2},\qquad x\in\{0,1\}^n,\ b\in\{0,1\},
\end{equation}
which acts on an address register of $n$ qubits and a target qubit. In the basis $\ket+=(\ket0+\ket1)/\sqrt2$, $\ket-=(\ket0-\ket1)/\sqrt2$ of the target qubit, the oracle acts as
\begin{equation}\label{eq:prelim-phase-kickback}
 O_f\ket x\ket+=\ket x\ket+,\qquad O_f\ket x\ket-=f(x)\ket x\ket-,
\end{equation}
because flipping the target qubit leaves $\ket+$ unchanged and multiplies $\ket-$ by $-1$. We call a query in which the target qubit is in the state $\ket-$ a \emph{phase query}. The operator $O_f$ is a permutation of the computational basis with $O_f^2=I$, so it is unitary, Hermitian, and its own inverse.

\paragraph{The Haar ensemble.}
The oracle hides a uniformly random unit vector $\ket\phi=\sum_{x\in[N]}\alpha_x\ket x$ in $\C^N$, and the ideal distribution is $p_\phi(x)=|\alpha_x|^2$. The ensemble is a distribution on unit vectors rather than on rays: the global phase of $\ket\phi$ matters, because the oracle below swaps $\ket\phi$ with a fixed vector, so that $\ket\phi$ and $e^{\mathrm i\theta}\ket\phi$ give different oracles.

The query register is $\mathsf Q=\C\ket\bot\oplus\C^N$, where $\ket\bot$ is a fixed unit vector orthogonal to $\C^N$. The \emph{canonical state-preparation oracle} is
\begin{equation}\label{eq:prelim-canonical-oracle}
 O_\phi=I-2\ket{\eta_\phi}\bra{\eta_\phi},\qquad \ket{\eta_\phi}=\frac{\ket\bot-\ket\phi}{\sqrt2},
\end{equation}
which is the reflection about the hyperplane orthogonal to $\ket{\eta_\phi}$. Expanding the projection gives
\begin{equation}\label{eq:prelim-oracle-expansion}
 O_\phi=I-\ket\bot\bra\bot+\ket\bot\bra\phi+\ket\phi\bra\bot-\ket\phi\bra\phi ,
\end{equation}
so $O_\phi$ swaps $\ket\bot$ and $\ket\phi$ and acts as the identity on the vectors orthogonal to both $\ket\bot$ and $\ket\phi$. Like $O_f$, it is unitary, Hermitian, and its own inverse.

The expansion \eqref{eq:prelim-oracle-expansion} also shows how the hidden amplitudes enter: in the basis $\{\ket\bot,\ket1,\dots,\ket N\}$ of $\mathsf Q$, every matrix entry of $O_\phi$ is a polynomial of degree at most two in the amplitudes $\alpha_x$ and their conjugates $\overline{\alpha_x}$, each factor $\ket\phi$ contributing one $\alpha_x$ and each factor $\bra\phi$ one $\overline{\alpha_x}$.

With \emph{Haar-random unitary access}, the algorithm may instead apply a Haar-random unitary $U$ on $\C^N$, its inverse $U^\dagger$, and controlled versions of both, and the state scored is $\ket\phi=U\ket{0^n}$. Both forms of access are the models of~\cite{Kretschmer}, and the second is at least as strong as the first, since $O_\phi$ can be built from one controlled application of $U^\dagger$ and one of $U$ with the help of one ancilla qubit.

To build $O_\phi$ in this way, identify $\ket\bot$ with $\ket0\ket{0^n}$ and $\ket x$ with $\ket1\ket x$, and let $cU=\ket0\bra0\otimes I+\ket1\bra1\otimes U$ be $U$ controlled by the ancilla. Then $cU^\dagger\ket{\eta_\phi}=\ket-\ket{0^n}$ for the vector $\ket{\eta_\phi}$ of \eqref{eq:prelim-canonical-oracle}, so
\[
 I-2\ket{\eta_\phi}\bra{\eta_\phi}=cU\bigl(I-2\ket-\bra-\otimes\ket{0^n}\bra{0^n}\bigr)cU^\dagger
\]
on all of $\C^2\otimes\C^N$. This operator is $O_\phi$ on the embedded register. The middle factor swaps $\ket0\ket{0^n}$ with $\ket1\ket{0^n}$ and fixes every $\ket1\ket x$ with $x\ne0^n$, and it uses no query.

\paragraph{The expected score.}
For an oracle $\omega$ with ideal distribution $p_\omega$, the \emph{expected score} of an algorithm $\mathcal A$ that outputs $Z\in\{0,1\}^n$ is
\begin{equation}\label{eq:prelim-score}
 B(\mathcal A)=\E_\omega\sum_z\Pr[Z=z\mid\omega]\,Np_\omega(z),
\end{equation}
which is \eqref{eq:intro-score}. We call $Np_\omega(z)$ the \emph{reward} of the outcome $z$, so that the score is the expected reward of the output. Outputting a uniformly random string, or a fixed string, gives $B(\mathcal A)=1$, because $\sum_zp_\omega(z)=1$ and $\E_\omega p_\omega(z)=1/N$ for every $z$.

The \emph{honest sampler} is the one-query algorithm that prepares a state whose measurement in the computational basis has the distribution $p_\omega$ and measures it. In the Fourier ensemble it applies Hadamard gates, a phase query, and Hadamard gates to $\ket{0^n}$, as in \cref{sec:intro}, and in the Haar ensemble it applies $O_\phi$ to $\ket\bot$. It outputs $z$ with probability $p_\omega(z)$, and its expected score, the \emph{ideal score}, is
\begin{equation}\label{eq:prelim-ideal-score}
 \E_\omega N\sum_zp_\omega(z)^2=
 \begin{cases}
 3-2/N&\text{in the Fourier ensemble,}\\
 2N/(N+1)&\text{in the Haar ensemble;}
 \end{cases}
\end{equation}
these values are computed in \cref{lem:moments}, and from \cref{sec:haar} on we write $\beta_N=2N/(N+1)$ for the ideal Haar score.

No algorithm scores more than $\E_\omega N\max_zp_\omega(z)$, a value that is approached by algorithms that learn the oracle; it is $\ln N+O(1)$ in the Haar ensemble (\cref{lem:heaviest}) and at most $2(n+2)$ in the Fourier ensemble (\cref{lem:fr-max-weight}).

\subsubsection{Query algorithms and a reduction to controlled queries}\label{sec:prelim-model}

We now describe the most general experiment that our theorems concern. It includes an adversary, Eve, who may share an entangled state with the device; the theorems about the score alone are the special case in which Eve's register is trivial. Throughout, \emph{device} and \emph{algorithm} refer to the same party, the one that makes the queries; we denote it by $\mathcal A$, so that $B(\mathcal A)$ in \eqref{eq:prelim-score} is the expected score of the device in the experiment.

\begin{definition}[The experiment]\label{def:experiment}
Fix $q\ge0$ and one of the two ensembles.
\begin{enumerate}
\item The device holds a register $A$ and Eve holds a register $E$, both of finite dimension, and they start in a joint state $\rho_{AE}$ that does not depend on the oracle. The oracle $\omega$ is then drawn.
\item The device applies oracle-independent operations to $A$ and makes at most $q$ queries along every path of its computation, which consists of finitely many steps. An operation may adjoin ancillas in fixed states, apply unitaries, and perform measurements whose outcomes are appended to a classical record kept in $A$, and every operation and every query may depend on the record produced before it.
\item A query applies $O_\omega$ to a register of $A$ of the appropriate type, chosen according to the record, possibly controlled by a qubit of $A$. Since $O_\omega$ is its own inverse, inverse queries are the same operation.
\item The output $Z\in\{0,1\}^n$ is a function of the final record. No operation acts on $E$, and nothing is sent from $A$ to $E$. After $Z$ has been produced, Eve receives a complete classical description of $\omega$ in a register $R$: the truth table of $f$, or the vector $\phi$.
\end{enumerate}
\end{definition}

Classical randomness of the device is included in $\rho_{AE}$, and so is any randomness shared with Eve. Haar-random unitary access, in which $U$ and $U^\dagger$ are distinct operations, is not covered by this definition and is reduced to the canonical oracle in \cref{sec:haar-unitary}.

For each oracle $\omega$ and each $z$, the experiment determines the subnormalized state $\rho^\omega_z$ of Eve's register on the event $Z=z$, whose trace is $\Pr[Z=z\mid\omega]$. Since nothing acts on $E$, we have $\sum_z\rho^\omega_z=\rho_E$ for every $\omega$, where $\rho_E$ is the reduced state of $E$ in $\rho_{AE}$, although $Z$ may be correlated with $E$. The next lemma puts every such experiment into a form with a fixed sequence of controlled queries, which is the form used in the proofs.

\begin{lemma}[Reduction to a fixed sequence of controlled queries]\label{lem:normal-form}
Consider an experiment as in \cref{def:experiment}. There are
\begin{enumerate}
\item[(a)] a finite-dimensional register $A'$, which contains a query register $\mathsf Q$ of the type on which $O_\omega$ acts and a control qubit $\mathsf C$,
\item[(b)] a unit vector $\psi_0\in A'\otimes E$ whose reduced state on $E$ is $\rho_E$,
\item[(c)] oracle-independent unitaries $U_0,\dots,U_q$ on $A'$, and
\item[(d)] orthogonal projections $\{\Pi_z\}_{z\in\{0,1\}^n}$ on $A'$ with $\sum_z\Pi_z=I$,
\end{enumerate}
with the following property. Let
\begin{equation}\label{eq:prelim-controlled-query}
 O^{\mathrm c}_\omega=\ket0\bra0_{\mathsf C}\otimes I+\ket1\bra1_{\mathsf C}\otimes O_\omega
\end{equation}
act on $\mathsf C$ and $\mathsf Q$ and trivially on the rest of $A'$, and put
\begin{equation}\label{eq:prelim-V}
 V_\omega=U_q\,O^{\mathrm c}_\omega\,U_{q-1}\,O^{\mathrm c}_\omega\cdots U_1\,O^{\mathrm c}_\omega\,U_0 .
\end{equation}
Then, for every oracle $\omega$ and every $z$,
\begin{equation}\label{eq:prelim-normal-state}
 \rho^\omega_z=\operatorname{Tr}_{A'}\Bigl[(\Pi_zV_\omega\otimes I_E)\ket{\psi_0}\bra{\psi_0}(\Pi_zV_\omega\otimes I_E)^\dagger\Bigr],
 \qquad
 \Pr[Z=z\mid\omega]=\bigl\|(\Pi_zV_\omega\otimes I_E)\psi_0\bigr\|^2 .
\end{equation}
\end{lemma}

\begin{proof}
We modify the experiment in four steps, none of which changes the subnormalized state of $E$ on the event $Z=z$, for any $\omega$ and $z$.

\emph{Step 1: the record as a register and the query rounds.} Regard the classical record as a register $\mathsf K$ of $A$ with an orthonormal basis $\{\ket r\}$ indexed by the possible records, which are finitely many because the computation has finitely many steps and every measurement has finitely many outcomes, and let the record include the number of queries made so far along its path. At every stage, the state of $AE$ is block diagonal, of the form
\[
 \sum_r\ket r\bra r_{\mathsf K}\otimes\sigma_r .
\]

We group the computation into $q$ \emph{query rounds} and a final round. Along a path, the rounds are as follows.
\begin{enumerate}
\item[(i)] For $1\le k\le q$, the $k$-th round consists of everything the device does after its $(k-1)$-st query, up to and including its $k$-th query, and the final round consists of what remains after the $q$-th query.
\item[(ii)] If the path makes exactly $j<q$ queries, then the $(j+1)$-st round consists instead of everything after its $j$-th query up to the end of the computation, and the rounds after the $(j+1)$-st, including the final round, are empty.
\end{enumerate}

Within a round, the operations other than the query, together with the recording of their outcomes, compose into a single oracle-independent completely positive trace-preserving map on $A$, possibly on an enlarged $A$, because every operation is selected by the record and the record grows with every outcome. The query of the $k$-th round is the unitary
\[
 \sum_r\ket r\bra r_{\mathsf K}\otimes Q^r_\omega ,
\]
where $Q^r_\omega$ is the identity if the path $r$ makes fewer than $k$ queries, and otherwise $O_\omega$ acting on the registers selected by $r$, possibly controlled by the qubit selected by $r$. On block-diagonal states these maps act exactly as the record-dependent operations of the experiment, and they are oracle-independent except for the queries.

\emph{Step 2: a fixed query register.} Adjoin a register $\mathsf Q$ of the type on which $O_\omega$ acts and a qubit $\mathsf C$, both in fixed states. Let $S^r$ be the oracle-independent unitary that swaps the registers selected by $r$ into $\mathsf Q$ and acts on $\mathsf C$ as follows.
\begin{enumerate}
\item[(i)] If the query of the path $r$ is controlled, then $S^r$ swaps the control qubit selected by $r$ into $\mathsf C$.
\item[(ii)] Otherwise, since $\mathsf C$ starts in the fixed state $\ket0$, the unitary $S^r$ flips it to $\ket1$ if the path makes a query and leaves it in $\ket0$ if it does not.
\end{enumerate}
Then the operator $W=\sum_r\ket r\bra r_{\mathsf K}\otimes S^r$ is an oracle-independent unitary, and
\[
 W^\dagger\,O^{\mathrm c}_\omega\,W=\sum_r\ket r\bra r_{\mathsf K}\otimes Q^r_\omega
\]
on every state that is block diagonal in $\mathsf K$ and has $\mathsf Q$ and $\mathsf C$ in their fixed states, which includes every state that occurs in the computation. The identity holds because, on each block $r$, the factor $W$, which acts first, moves the relevant registers into position, the controlled query acts on them exactly as $Q^r_\omega$ does, and $W^\dagger$ moves them back.

Since $O^{\mathrm c}_\omega$ does not act on $\mathsf K$, both applications of $W$ are controlled by the same value of the record. Thus every query step is $O^{\mathrm c}_\omega$ conjugated by oracle-independent unitaries, and these unitaries are absorbed into the operations before and after the query.

\emph{Step 3: dilation and purification.} Replace each oracle-independent completely positive trace-preserving map by a Stinespring isometry whose environment register is kept by the device and never acted on again. This does not change the reduced state of the other registers at any later time, because all later operations act trivially on the environment. Purify $\rho_{AE}$ with a register that is kept by the device and never acted on. Include every register that is adjoined during the experiment in the initial state, each in its fixed state.

The result is a register $A'$ and a unit vector $\psi_0\in A'\otimes E$ whose reduced state on $E$ is $\rho_E$, and the device now applies, alternately, oracle-independent unitaries on $A'$ and the controlled query $O^{\mathrm c}_\omega$ on $\mathsf C\mathsf Q$, that is, it applies $V_\omega$ of \eqref{eq:prelim-V}. A path that makes fewer than $q$ queries has $\mathsf C$ in the state $\ket0$ at the remaining query rounds.

\emph{Step 4: the output.} Extend the output function $Z$ arbitrarily to the records that cannot occur at the end of the computation, and let $\Pi_z$ be the projection onto the records $r$ with $Z(r)=z$, tensored with the identity on the other registers of $A'$, so that $\sum_z\Pi_z=I$.
After Step~3, the reduced state of $\mathsf K$, the original registers of $A$, and $E$ at the end of the computation equals the final state of the original experiment, and it is block diagonal in $\mathsf K$. Therefore the partial trace in \eqref{eq:prelim-normal-state} is the sum of the $E$-parts of the blocks with $Z(r)=z$, which is $\rho^\omega_z$. Its trace is $\Pr[Z=z\mid\omega]$.
\end{proof}

From now on we write $A$ for the register $A'$ of \cref{lem:normal-form}. We refer to $\Psi(\omega)=(V_\omega\otimes I_E)\psi_0$ as the \emph{final state} of the device for the oracle $\omega$, and to $\Psi_z(\omega)=(\Pi_z\otimes I_E)\Psi(\omega)$ as its \emph{output-$z$ branch}. In this notation the expected score \eqref{eq:prelim-score} reads
\begin{equation}\label{eq:prelim-score-branches}
 B(\mathcal A)=\E_\omega\sum_zNp_\omega(z)\,\|\Psi_z(\omega)\|^2 .
\end{equation}
A query that is controlled by $\mathsf C$ in the state $\ket1$ is an ordinary query, so the lemma also covers algorithms that query without control, and in particular the honest samplers of both ensembles.

\subsubsection{Amplitude amplification}\label{sec:prelim-amplification}

The squaring algorithm of \cref{sec:squaring} and the state-preparation algorithm of \cref{sec:hadamard} use amplitude amplification in the following standard form~\cite{BrassardHoyerMoscaTapp}.

\begin{lemma}[Amplitude amplification]\label{lem:amplification}
Let $\psi$ be a unit vector in a Hilbert space, let $\Pi$ be an orthogonal projection, and let $P=\|\Pi\psi\|^2>0$.
\begin{enumerate}
\item[(i)] If $P=1$, then
\[
 \bigl((2\ket\psi\bra\psi-I)(I-2\Pi)\bigr)^\ell\psi=(-1)^\ell\psi
\]
for every integer $\ell\ge0$.
\item[(ii)] If $P<1$, put
\[
 \vartheta=\arcsin\sqrt P\in(0,\pi/2),\qquad u=\frac{\Pi\psi}{\sqrt P},\qquad u^\perp=\frac{(I-\Pi)\psi}{\sqrt{1-P}},
\]
so that $\psi=\sin\vartheta\,u+\cos\vartheta\,u^\perp$; then, for every integer $\ell\ge0$,
\[
 \bigl((2\ket\psi\bra\psi-I)(I-2\Pi)\bigr)^\ell\psi=\sin\bigl((2\ell+1)\vartheta\bigr)\,u+\cos\bigl((2\ell+1)\vartheta\bigr)\,u^\perp .
\]
\end{enumerate}
\end{lemma}

\begin{proof}
If $P=1$, then $\Pi\psi=\psi$, so $(I-2\Pi)\psi=-\psi$ and $(2\ket\psi\bra\psi-I)(-\psi)=-\psi$.

For (ii), let $0<P<1$. The vectors $u$ and $u^\perp$ are orthonormal, and the plane they span contains $\psi$ and is invariant under both reflections. Indeed, $I-2\Pi$ fixes $u^\perp$ and negates $u$, while $2\ket\psi\bra\psi-I$ fixes $\psi$ and negates the vector $\cos\vartheta\,u-\sin\vartheta\,u^\perp$ orthogonal to it. On this plane, in the basis $(u,u^\perp)$, the product of the two reflections has the matrix
\[
 \begin{pmatrix}\cos2\vartheta&\sin2\vartheta\\-\sin2\vartheta&\cos2\vartheta\end{pmatrix},
\]
which is the rotation by the angle $2\vartheta$ that takes $u^\perp$ toward $u$; here $xu+yu^\perp$ is written $(x,y)$, and the first column is the image of $u$. By the addition formulas the matrix maps $(\sin\alpha,\cos\alpha)$ to $(\sin(\alpha+2\vartheta),\cos(\alpha+2\vartheta))$ for every $\alpha$; in particular it maps $\psi=(\sin\vartheta,\cos\vartheta)$ to $(\sin3\vartheta,\cos3\vartheta)$, and induction on $\ell$ gives the claim.
\end{proof}

\subsection{Min-entropy, fidelity, and purified distance}\label{sec:prelim-entropy}

The oracle register $R$ is finite in the Fourier ensemble and continuous in the Haar ensemble. To treat both cases at once, we describe the joint state of $R$, $Z$, and $E$ by the family of Eve's conditional states.

\begin{definition}[Families of positive operators and entropies]\label{def:entropies}
A \emph{family of positive operators} is a map, measurable with respect to the distribution of the oracle, $\omega\mapsto(\tau^\omega_z)_{z\in\{0,1\}^n}$ from oracles to tuples of positive semidefinite operators on $E$. It is \emph{normalized} if $\E_\omega\sum_z\operatorname{Tr}\tau^\omega_z=1$. This normalization is joint over the oracle and output; the total trace for each fixed oracle need not be $1$.
For normalized families $\rho$ and $\tau$ define
\begin{align*}
 p_{\mathrm{guess}}(Z\mid RE)_\tau&=\E_\omega\sup_{M}\sum_z\operatorname{Tr}(M_z\tau^\omega_z),\\
 \Hmin(Z\mid RE)_\tau&=-\log p_{\mathrm{guess}}(Z\mid RE)_\tau,\\
 F(\rho,\tau)&=\E_\omega\sum_z\bigl\|\sqrt{\rho^\omega_z}\sqrt{\tau^\omega_z}\bigr\|_1,\\
 P(\rho,\tau)&=\sqrt{1-F(\rho,\tau)^2},
\end{align*}
where the supremum is over all measurements $M=(M_z)_z$ on $E$, that is, over all tuples of positive semidefinite operators with $\sum_zM_z=I_E$, and where $\|T\|_1=\operatorname{Tr}\sqrt{T^\dagger T}$ is the trace norm.

For $0\le\eta\le1$, the \emph{smooth min-entropy} is
\[
 \Hmin^\eta(Z\mid RE)_\rho=\sup\bigl\{\Hmin(Z\mid RE)_\tau:\ \tau\text{ a normalized family with }P(\rho,\tau)\le\eta\bigr\}.
\]
\end{definition}

Here $F(\rho,\tau)\le1$, because
\[
 \bigl\|\sqrt{\rho^\omega_z}\sqrt{\tau^\omega_z}\bigr\|_1\le\bigl(\operatorname{Tr}\rho^\omega_z\,\operatorname{Tr}\tau^\omega_z\bigr)^{1/2},
\]
and the Cauchy--Schwarz inequality over $\omega$ and $z$, together with the normalization of both families, bounds the sum by $1$.

The family that describes an experiment consists of the states $\rho^\omega_z$ of \cref{def:experiment}, and it is normalized since $\sum_z\operatorname{Tr}\rho^\omega_z=1$ for every $\omega$.
In this family, $p_{\mathrm{guess}}(Z\mid RE)_\rho$ is the largest probability with which Eve, who knows $\omega$ and holds $E$, guesses $Z$, since she may choose her measurement as a function of $\omega$. The integrands in \cref{def:entropies} are measurable functions of $\omega$, because the supremum may be taken over a countable dense set of measurements, and the trace norm is continuous.

\begin{remark}[Relation to the standard definitions]\label{rem:entropy-conventions}
When the set of oracles is finite, as in the Fourier ensemble, a normalized family is the same thing as a density operator
\[
 \tau_{RZE}=\sum_{\omega,z}\mu(\omega)\ket\omega\bra\omega_R\otimes\ket z\bra z_Z\otimes\tau^\omega_z
\]
that is classical on $R$ and $Z$, where $\mu(\omega)$ is the probability
of the oracle $\omega$. In this case, the quantities above are the standard ones: $\Hmin(Z\mid RE)_\tau$ is the conditional min-entropy, by its characterization as a guessing probability~\cite{KoenigRennerSchaffner}, and $F$ and $P$ are the fidelity and the purified distance of the two block-diagonal operators~\cite{TomamichelColbeckRenner}.

Our smoothing takes the supremum over states that are classical on $R$ and $Z$ and normalized, whereas the standard definition also allows subnormalized states that are not classical on $RZ$. The supremum here is therefore over a smaller set, and every lower bound proved below also holds for the standard quantity, which is the one used in randomness extraction against quantum side information~\cite{TomamichelEtAlHashing}.

In the Haar ensemble, where $R$ is continuous, we take \cref{def:entropies} as the definition. If Eve receives a function $g(\phi)$ with finitely many values instead of $\phi$, then the family is replaced by the averages of $\rho^\phi_z$ over the fibers of $g$.

The guessing probability can only decrease under this averaging, and the fidelity can only increase, because the map $(a,b)\mapsto\bigl\|\sqrt a\sqrt b\bigr\|_1$ is jointly concave on pairs of positive semidefinite operators. This concavity follows from the strong concavity of the fidelity~\cite[Theorem~9.7]{NielsenChuang} and the homogeneity $\bigl\|\sqrt{sa}\sqrt{tb}\bigr\|_1=\sqrt{st}\,\bigl\|\sqrt a\sqrt b\bigr\|_1$.
Hence every lower bound on $\Hmin^\eta(Z\mid RE)$ proved for the continuous register also holds, in the standard sense, when $\phi$ is disclosed to any finite precision.
\end{remark}

We use two standard facts. First, for unit vectors $u$ and $v$ in $A\otimes E$ whose reduced states on $E$ are $\rho$ and $\tau$,
\begin{equation}\label{eq:prelim-uhlmann}
 \bigl\|\sqrt\rho\sqrt\tau\bigr\|_1\ge|\langle u,v\rangle| ,
\end{equation}
which is the elementary half of Uhlmann's theorem~\cite[Section~9.2.2]{NielsenChuang}. By homogeneity, it also holds for vectors that are not normalized, with $\rho$ and $\tau$ the corresponding positive semidefinite operators of trace $\|u\|^2$ and $\|v\|^2$.

Second, the following lemma shows that an operator acting on the device alone cannot raise the probability of any outcome of a measurement on $E$ by more than a factor of its squared norm.

\begin{lemma}[Operators on the device alone]\label{lem:local-operator}
Let $K:A\to A_1$ be a linear map between finite-dimensional spaces, let $\psi\in A\otimes E$, and let
\[
 \rho=\operatorname{Tr}_A\ket\psi\bra\psi\quad\text{and}\quad\rho_K=\operatorname{Tr}_{A_1}\bigl[(K\otimes I_E)\ket\psi\bra\psi(K\otimes I_E)^\dagger\bigr].
\]
Then $\operatorname{Tr}(M\rho_K)\le\|K\|^2\operatorname{Tr}(M\rho)$ for every positive semidefinite operator $M$ on $E$.
\end{lemma}

\begin{proof}
We have
\[
 \operatorname{Tr}(M\rho_K)=\langle\psi,(K^\dagger K\otimes M)\psi\rangle\quad\text{and}\quad\operatorname{Tr}(M\rho)=\langle\psi,(I\otimes M)\psi\rangle .
\]
The operator $(\|K\|^2I-K^\dagger K)\otimes M$ is a tensor product of two positive semidefinite operators, hence positive semidefinite, so
\[
 \langle\psi,(K^\dagger K\otimes M)\psi\rangle\le\|K\|^2\langle\psi,(I\otimes M)\psi\rangle . \qedhere
\]
\end{proof}

\subsection{Dirichlet distributions}\label{sec:prelim-dirichlet}

For $a=(a_1,\dots,a_N)$ with all $a_i\ge1$, write $|a|=a_1+\dots+a_N$. The \emph{Dirichlet distribution} $\mathrm{Dir}(a)$ is the probability distribution on the simplex $\{p\in\R^N:\ p_i\ge0,\ \sum_ip_i=1\}$ with density proportional to $\prod_ip_i^{a_i-1}$ with respect to the uniform measure on the simplex, and $\mathrm{Dir}(1,\dots,1)$ is the uniform distribution.

We write $\E_a$ for the expectation under $\mathrm{Dir}(a)$, and $\langle g,h\rangle_a=\E_a[\bar gh]$ and $\|h\|_a$ for the inner product and norm of the space of square-integrable functions on the simplex under $\mathrm{Dir}(a)$. We also write $(x)_k=x(x+1)\cdots(x+k-1)$ for the rising factorial. The following facts are classical~\cite[Chapter~2]{NgTianTang}; we include the short proofs.

\begin{lemma}[Dirichlet facts]\label{lem:dirichlet}
Let $p\sim\mathrm{Dir}(a)$.
\begin{enumerate}
\item[(i)] For nonnegative integers $b_1,\dots,b_N$,
\[
 \E_a\prod_ip_i^{b_i}=\frac{\prod_i(a_i)_{b_i}}{(|a|)_{|b|}},\qquad\text{where }|b|=\sum_ib_i .
\]
In particular,
\[
 \E_ap_i=\frac{a_i}{|a|}\quad\text{and}\quad\operatorname{Var}_a(p_i)=\frac{a_i(|a|-a_i)}{|a|^2(|a|+1)} .
\]
\item[(ii)] As measures on the simplex,
\[
 p_i\,d\mathrm{Dir}(a)=\frac{a_i}{|a|}\,d\mathrm{Dir}(a+e_i),
\]
where $e_i$ is the $i$-th unit vector. Equivalently, $\E_a[p_ih]=(a_i/|a|)\E_{a+e_i}[h]$ for every integrable $h$.
\item[(iii)] For $S\subseteq[N]$, we write $a(S)=\sum_{i\in S}a_i$ and
$p(S)=\sum_{i\in S}p_i$. If $\varnothing\ne S\subsetneq[N]$, the mass $p(S)$ has the Beta distribution with parameters $(a(S),|a|-a(S))$, whose density on $[0,1]$ is proportional to $t^{a(S)-1}(1-t)^{|a|-a(S)-1}$, and the vector $(p_i/p(S))_{i\in S}$ is independent of $p(S)$ and has the distribution $\mathrm{Dir}((a_i)_{i\in S})$.
\end{enumerate}
\end{lemma}

\begin{proof}
With respect to Lebesgue measure in the coordinates $p_1,\dots,p_{N-1}$, the density of $\mathrm{Dir}(a)$ is
\[
 \frac{\Gamma(|a|)\prod_ip_i^{a_i-1}}{\prod_i\Gamma(a_i)} .
\]

For (i), multiplying this density by $\prod_ip_i^{b_i}$ gives a multiple of the density of $\mathrm{Dir}(a+b)$, so the expectation is the ratio of the two normalizing constants,
\[
 \frac{\Gamma(|a|)\prod_i\Gamma(a_i+b_i)}{\Gamma(|a|+|b|)\prod_i\Gamma(a_i)} ,
\]
which is the stated product of rising factorials. Part~(ii) is the case $b=e_i$ of~(i), read as an identity between densities.

For (iii), write
\[
 t=p(S),\qquad p_i=tu_i\quad\text{for }i\in S,\qquad p_j=(1-t)v_j\quad\text{for }j\notin S,
\]
where $u$ and $v$ range over the simplices of dimensions $|S|-1$ and $N-|S|-1$. In the coordinates consisting of $t$ and all but one coordinate of each of $u$ and $v$, the Jacobian of this change of variables is proportional to $t^{|S|-1}(1-t)^{N-|S|-1}$, and
\[
 \prod_ip_i^{a_i-1}=t^{a(S)-|S|}(1-t)^{|a|-a(S)-N+|S|}\prod_{i\in S}u_i^{a_i-1}\prod_{j\notin S}v_j^{a_j-1}.
\]
Hence the joint density of $(t,u,v)$ is proportional to $t^{a(S)-1}(1-t)^{|a|-a(S)-1}$ times a function of $u$ times a function of $v$. Since $(t,u,v)$ ranges over a product of three sets, $t$, $u$, and $v$ are independent, $t$ has the stated Beta distribution, and $u=(p_i/p(S))_{i\in S}$ has the distribution $\mathrm{Dir}((a_i)_{i\in S})$.
\end{proof}

Part~(ii) says that multiplying the density by $p_i$ is the same as observing the outcome $i$ once: the prior $\mathrm{Dir}(a)$ becomes the posterior $\mathrm{Dir}(a+e_i)$, and the normalizing factor $a_i/|a|$ is the prior probability of that observation. This interpretation is used throughout \cref{sec:haar}.

\paragraph{The Haar measure in these terms.}
Write $\alpha_x=\sqrt{p_x}\,e^{\mathrm i\theta_x}$ for the amplitudes of a Haar-random state. Under the Haar measure, the phases $\theta_1,\dots,\theta_N$ are independent and uniform on $[0,2\pi)$, they are independent of $p=(p_1,\dots,p_N)$, and $p$ is uniform on the simplex, that is, $p\sim\mathrm{Dir}(1,\dots,1)$. This is the standard description of the Haar measure on the sphere in terms of its moduli and phases; for example, it follows by writing a vector of independent standard complex Gaussians in polar coordinates and normalizing it.

\section{The lower bound for Fourier sampling}\label{sec:fourier}

In this section we prove part~1 of \cref{thm:intro-xeb}. We fix an experiment in the Fourier ensemble, which by \cref{lem:normal-form} we may assume makes a fixed sequence of $q$ controlled queries. The experiment has device register $A$, Eve's register $E$ (which may be trivial), initial state $\psi_0\in A\otimes E$, unitaries $U_0,\dots,U_q$, projections $\{\Pi_z\}$, and final state
\[
 \Psi(f)=(V_f\otimes I_E)\psi_0,\qquad V_f=U_qO^{\mathrm c}_fU_{q-1}\cdots O^{\mathrm c}_fU_0 .
\]
The query register $\mathsf Q$ of \cref{lem:normal-form} consists of an address register $\mathsf X$ of $n$ qubits and a target qubit $\mathsf T$, so the query acts on the control qubit $\mathsf C$, on $\mathsf X$, and on $\mathsf T$.

\subsection{The recording representation}\label{sec:fourier-recording}

Let $\mathcal H=A\otimes E$. The final state is a function $f\mapsto\Psi(f)$ from $\{-1,1\}^{\{0,1\}^n}$ to $\mathcal H$. We work in the space $\mathcal F$ of all such functions, equipped with the inner product $\langle\Phi,\Phi'\rangle=\E_f\langle\Phi(f),\Phi'(f)\rangle$, so that $\|\Phi\|^2=\E_f\|\Phi(f)\|^2$ and $\|\Psi\|=1$.

For a set $S\subseteq\{0,1\}^n$, let $\chi_S(f)=\prod_{y\in S}f(y)$, with $\chi_\varnothing=1$. These are the characters of the group $\{-1,1\}^{\{0,1\}^n}$, indexed by sets. They should not be confused with the characters $\chi_z$ of \eqref{eq:prelim-walsh}, which are indexed by strings. Since the values of $f$ are independent uniform signs, $\E_f\chi_S(f)\chi_{S'}(f)=\mathbf 1[S=S']$, and the $2^N$ functions $\chi_S$ form an orthonormal basis of the scalar functions on $\{-1,1\}^{\{0,1\}^n}$. Every $\Phi\in\mathcal F$ therefore has the expansion
\begin{equation}\label{eq:fourier-expansion}
 \Phi=\sum_S\chi_S\,\Phi_S,\qquad \Phi_S=\E_f\bigl[\chi_S(f)\Phi(f)\bigr]\in\mathcal H,\qquad \|\Phi\|^2=\sum_S\|\Phi_S\|^2 .
\end{equation}

We call $S$ the \emph{recorded set} of the term $\chi_S\Phi_S$ and its elements the \emph{recorded points}. This is the representation used in the compressed oracle technique~\cite{Zhandry}. If the random function is purified by a register in the uniform superposition $2^{-N/2}\sum_f\ket f$, then $\chi_S\Phi_S$ corresponds to the vector $\ket S\otimes\Phi_S$, where the vectors $\ket S=2^{-N/2}\sum_f\chi_S(f)\ket f$ form the Fourier basis of that register. We shall not need the register itself, only the expansion \eqref{eq:fourier-expansion}.

\paragraph{Operators on $\mathcal F$.}
An operator $T$ on $\mathcal H$ acts on $\mathcal F$ pointwise, by $(T\Phi)(f)=T\Phi(f)$. This action preserves each coefficient space: $(T\Phi)_S=T\Phi_S$. Multiplication by a scalar function $m(f)$ also acts on $\mathcal F$. In particular, $\mathcal X_x$ denotes multiplication by $f(x)$. Since $f(x)\chi_S(f)=\chi_{S\triangle\{x\}}(f)$, where $\triangle$ is the symmetric difference, we have
\begin{equation}\label{eq:fourier-toggle}
 \mathcal X_x(\chi_S\,h)=\chi_{S\triangle\{x\}}\,h\qquad(h\in\mathcal H).
\end{equation}
That is, multiplying by $f(x)$ records the point $x$ if it is absent and removes it if it is present.

We split this operation into its two halves. Let $a_x$ and $a_x^\dagger$ be the operators on $\mathcal F$ defined on the basis by
\begin{equation}\label{eq:fourier-ax}
 a_x(\chi_S\,h)=\mathbf 1[x\in S]\,\chi_{S\setminus\{x\}}\,h,\qquad
 a_x^\dagger(\chi_S\,h)=\mathbf 1[x\notin S]\,\chi_{S\cup\{x\}}\,h.
\end{equation}
Thus $a_x$ removes the point $x$ and $a_x^\dagger$ records it, and each gives zero when it cannot do so. The two operators are adjoint to each other, since the inner products $\langle\chi_{S'}h',a_x\chi_Sh\rangle$ and $\langle a_x^\dagger\chi_{S'}h',\chi_Sh\rangle$ both equal $\langle h',h\rangle$ when $S'=S\setminus\{x\}$ with $x\in S$, and both vanish otherwise.

In the language of qubits, if the set $S$ is encoded by $N$ membership qubits, one for each address, then $a_x$, $a_x^\dagger$, and $a_x^\dagger a_x$ are the lowering, raising, and number operators of the $x$-th qubit, and $\mathcal X_x$ is its Pauli $X$ operator. These qubits represent the dependence on $f$ and are not a register that the algorithm can act on.

Directly from \eqref{eq:fourier-toggle} and \eqref{eq:fourier-ax},
\begin{equation}\label{eq:fourier-ax-algebra}
 \mathcal X_x=a_x+a_x^\dagger,\qquad a_x^2=0,\qquad a_xa_x^\dagger+a_x^\dagger a_x=I,\qquad a_xa_x^\dagger-a_x^\dagger a_x=I-2n_x,
\end{equation}
where the operator $n_x=a_x^\dagger a_x$ is the orthogonal projection onto the terms in which $x$ is recorded. Operators attached to distinct points $x\ne y$ commute, because they act on the membership of different points. The operator $D=\sum_xn_x$, with $D(\chi_S\,h)=|S|\,\chi_S\,h$, counts the recorded points.

Finally, for $z\in\{0,1\}^n$ we define
\begin{equation}\label{eq:fourier-bz}
 b_z=\frac1{\sqrt N}\sum_x\chi_z(x)\,a_x .
\end{equation}
This operator removes one recorded point with the amplitudes of the Hadamard basis state $\ket{\chi_z}=H^{\otimes n}\ket z=N^{-1/2}\sum_x\chi_z(x)\ket x$: the amplitude with which it removes $x$ is the amplitude of $\ket x$ in $\ket{\chi_z}$.

These operators also have a description in terms of the function $\Phi$ itself. The operator $a_x$ is the discrete derivative
\[
 (a_x\Phi)(f)=\frac12\bigl(\Phi(f^{x,+})-\Phi(f^{x,-})\bigr),
\]
where $f^{x,\pm}$ agrees with $f$ outside $x$ and takes the value $\pm1$ at $x$. This holds because
\[
 \chi_S(f^{x,\pm})=
 \begin{cases}
  \pm\chi_{S\setminus\{x\}}(f) & \text{if }x\in S,\\
  \chi_S(f) & \text{otherwise.}
 \end{cases}
\]
This description is used in \cref{sec:randomness}.

\begin{lemma}[Identities for the Fourier operators]\label{lem:fourier-identities}
For all $z,x\in\{0,1\}^n$ and all $\Phi\in\mathcal F$, the following hold.
\begin{enumerate}
\item[(i)] The operators $b_w$ satisfy
\[
 \sum_wb_w^\dagger b_w=D .
\]
\item[(ii)] The operators $b_z$ and $b_z^\dagger$ satisfy
\[
 b_zb_z^\dagger-b_z^\dagger b_z=I-\frac2ND .
\]
\item[(iii)] The vector $\Phi$ satisfies
\[
 \sum_w\bigl\|(b_w^2\mathcal X_x-\mathcal X_xb_w^2)\Phi\bigr\|^2=\frac4N\langle\Phi,(D-n_x)\Phi\rangle .
\]
\end{enumerate}
\end{lemma}

\begin{proof}
\emph{Part (i).} By \eqref{eq:fourier-bz} and the first relation in \eqref{eq:prelim-walsh},
\begin{align*}
 \sum_wb_w^\dagger b_w
 &=\frac1N\sum_{x,y}\Bigl(\sum_w\chi_w(x)\chi_w(y)\Bigr)a_x^\dagger a_y\\
 &=\sum_xa_x^\dagger a_x=D .
\end{align*}

\emph{Part (ii).} Expanding the products gives
\[
 b_zb_z^\dagger-b_z^\dagger b_z=\frac1N\sum_{x,y}\chi_z(x)\chi_z(y)(a_xa_y^\dagger-a_y^\dagger a_x).
\]
For $x\ne y$ the operators $a_x$ and $a_y^\dagger$ commute, so only the terms with $x=y$ remain. By \eqref{eq:fourier-ax-algebra}, these terms give
\[
 \frac1N\sum_x(I-2n_x)=I-\frac2ND .
\]

\emph{Part (iii).} Fix $x$ and write
\[
 b_w=N^{-1/2}\chi_w(x)a_x+c_w,\qquad c_w=N^{-1/2}\sum_{y\ne x}\chi_w(y)a_y,
\]
where $c_w$ is an operator that commutes with $a_x$, $a_x^\dagger$, and $n_x$. Since $a_x^2=0$ and $\chi_w(x)^2=1$,
\[
 b_w^2=\frac2{\sqrt N}\chi_w(x)\,a_xc_w+c_w^2 .
\]
The term $c_w^2$ commutes with $\mathcal X_x=a_x+a_x^\dagger$. Moreover, by \eqref{eq:fourier-ax-algebra},
\begin{align*}
 a_x\mathcal X_x-\mathcal X_xa_x&=a_xa_x^\dagger-a_x^\dagger a_x=I-2n_x ,\\
\intertext{so that}
 b_w^2\mathcal X_x-\mathcal X_xb_w^2&=\frac2{\sqrt N}\chi_w(x)\,(I-2n_x)\,c_w .
\end{align*}
The operator $I-2n_x$ multiplies each basis term by $\pm1$, so it preserves norms. Therefore
\begin{align*}
 \sum_w\bigl\|(b_w^2\mathcal X_x-\mathcal X_xb_w^2)\Phi\bigr\|^2
 &=\frac4N\sum_w\|c_w\Phi\|^2\\
 &=\frac4N\Bigl\langle\Phi,\sum_wc_w^\dagger c_w\Phi\Bigr\rangle .
\end{align*}
As in (i), the orthogonality relation \eqref{eq:prelim-walsh} gives
\[
 \sum_wc_w^\dagger c_w=\sum_{y\ne x}a_y^\dagger a_y=D-n_x . \qedhere
\]
\end{proof}

\paragraph{The query in the recording representation.}
By \eqref{eq:prelim-phase-kickback}, the value oracle decomposes as
\[
 O_f=P_++\sum_xf(x)P_x,\qquad P_+=I_{\mathsf X}\otimes\ket+\bra+_{\mathsf T},\qquad P_x=\ket x\bra x_{\mathsf X}\otimes\ket-\bra-_{\mathsf T},
\]
where $P_+$ and the $P_x$ are mutually orthogonal projections on $\mathsf X\mathsf T$ that sum to the identity. Hence the controlled query \eqref{eq:prelim-controlled-query} acts on $\mathcal F$ as
\begin{equation}\label{eq:fourier-query}
 Q=P_\varnothing+\sum_x\mathcal X_xP^{\mathrm c}_x,\qquad
 P_\varnothing=\ket0\bra0_{\mathsf C}\otimes I+\ket1\bra1_{\mathsf C}\otimes P_+,\qquad
 P^{\mathrm c}_x=\ket1\bra1_{\mathsf C}\otimes P_x ,
\end{equation}
where $P_\varnothing$ and the $P^{\mathrm c}_x$ are mutually orthogonal projections on $A$ summing to the identity. These projections act pointwise on $\mathcal F$, and $\mathcal X_x$ commutes with all of them.

The final state is $\Psi=U_qQU_{q-1}Q\cdots U_1QU_0\,\psi_0$, where $\psi_0$ is regarded as a constant function and the $U_t$ act pointwise. We write $\Psi_t$ for the state after $t$ queries and the unitary $U_t$ that follows them, so that $\Psi_0=U_0\psi_0$ and $\Psi_q=\Psi$.

Three facts control the recorded sets that occur in $\Psi_t$.
\begin{enumerate}
\item The constant function $\psi_0$ has the single recorded set $S=\varnothing$.
\item A pointwise unitary does not change the sets $S$ that occur.
\item By \eqref{eq:fourier-toggle}, the query changes each set by at most one point.
\end{enumerate}
By induction on $t$,
\begin{equation}\label{eq:fourier-degree}
 (\Psi_t)_S=0\quad\text{whenever }|S|>t,\qquad\text{hence}\qquad\langle\Phi,D\Phi\rangle\le t\,\|\Phi\|^2
\end{equation}
for every $\Phi$ that is obtained from $\Psi_t$ by applying operators that act pointwise, such as $P^{\mathrm c}_x$ or $\Pi_z$. The first part of \eqref{eq:fourier-degree} is the statement that the final state is a polynomial of degree at most $q$ in the values of $f$, as in the polynomial method~\cite{BealsEtAl}.

\subsection{The score bound}\label{sec:fourier-score}

Let $\Psi_z=\Pi_z\Psi$ be the output-$z$ branch of the final state. Since $\sqrt N\widehat f(z)=N^{-1/2}\sum_x\chi_z(x)f(x)$, multiplication by $\sqrt N\widehat f(z)$ is the operator $N^{-1/2}\sum_x\chi_z(x)\mathcal X_x=b_z+b_z^\dagger$.
By \eqref{eq:prelim-score-branches} and this identity,
\begin{equation}\label{eq:fourier-score-operator}
 B(\mathcal A)=\sum_z\E_f\bigl[N\widehat f(z)^2\,\|\Psi_z(f)\|^2\bigr]=\sum_z\bigl\|(b_z+b_z^\dagger)\Psi_z\bigr\|^2 .
\end{equation}

We write $D_z=b_z^\dagger b_z$, so that, by \cref{lem:fourier-identities}(i), $D=\sum_wD_w$.
The next lemma bounds the score in terms of the \emph{progress measure}
\begin{equation}\label{eq:fourier-progress}
 \lambda(\Phi)=\Bigl(\sum_z\|b_z^2\Phi\|^2\Bigr)^{1/2}.
\end{equation}

The lemma retains one nonpositive term, which the proof of the lower bound discards and which \cref{sec:randomness} uses.

\begin{lemma}[Score bound]\label{lem:fourier-score-bound}
For the final state $\Psi$ of the experiment, with $\lambda=\lambda(\Psi)$ and $\Psi_z=\Pi_z\Psi$,
\begin{equation}\label{eq:fourier-score-bound}
 B(\mathcal A)\le3-\frac2N+2\lambda+2\Bigl(1-\frac1N\Bigr)\lambda^2-2\Bigl(1-\frac1N\Bigr)\sum_z\bigl\|(D_z-I)\Psi_z\bigr\|^2 .
\end{equation}
\end{lemma}

\begin{proof}
\emph{Step 1: the branch quantity $e$.} Let
\[
 e=\Bigl(\sum_z\|b_z^2\Psi_z\|^2\Bigr)^{1/2}
\]
be the progress measure evaluated branch by branch. Since $\Pi_z$ acts on $A$ and $b_z$ acts on the dependence on $f$, the two commute, so
\begin{align*}
 e^2&=\sum_z\|\Pi_zb_z^2\Psi\|^2\\
 &\le\sum_z\|b_z^2\Psi\|^2=\lambda^2 ,
\end{align*}
where the inequality holds because each $\Pi_z$ is a projection. We prove \eqref{eq:fourier-score-bound} with $e$ in place of $\lambda$. Since the right side of that bound is increasing in $e\ge0$, and $e\le\lambda$, the lemma follows.

\emph{Step 2: the expansion of the square.} For the bound, fix $z$ and write $g=\Psi_z$. Expanding the square and using \cref{lem:fourier-identities}(ii),
\begin{equation}\label{eq:fourier-square-expansion}
 (b_z+b_z^\dagger)^2=b_z^2+b_z^{\dagger2}+b_zb_z^\dagger+b_z^\dagger b_z=b_z^2+b_z^{\dagger2}+I-\frac2ND+2D_z .
\end{equation}

\emph{Step 3: the first bound on the quadratic form.} We bound the quadratic form on $g$ of the operator $I-(2/N)D+2D_z$ in two steps, Steps~3 and~4. First, $\langle g,D_zg\rangle\le\langle g,Dg\rangle$, because $D$ is the sum of the operators $D_w$ and each $D_w=b_w^\dagger b_w$ has a nonnegative quadratic form. Hence
\begin{equation}\label{eq:fourier-step1}
 \Bigl\langle g,\Bigl(I-\frac2ND+2D_z\Bigr)g\Bigr\rangle\le\|g\|^2+2\Bigl(1-\frac1N\Bigr)\langle g,D_zg\rangle .
\end{equation}

\emph{Step 4: the key identity and the second bound.} Second, applying \cref{lem:fourier-identities}(ii) once more,
\begin{align*}
 D_z^2&=b_z^\dagger(b_zb_z^\dagger)b_z=b_z^\dagger\Bigl(D_z+I-\frac2ND\Bigr)b_z\\
 &=b_z^{\dagger2}b_z^2+D_z-\frac2Nb_z^\dagger Db_z,
\end{align*}
and rearranging gives the identity
\begin{equation}\label{eq:fourier-key-identity}
 I+b_z^{\dagger2}b_z^2-D_z=(D_z-I)^2+\frac2N\,b_z^\dagger Db_z .
\end{equation}
Taking the quadratic form on $g$ and using $\langle b_zg,Db_zg\rangle\ge0$, we obtain
\begin{equation}\label{eq:fourier-step2}
 \langle g,D_zg\rangle\le\|g\|^2+\|b_z^2g\|^2-\|(D_z-I)g\|^2 .
\end{equation}
Dropping its last term gives the inequality \eqref{eq:ov-positivity} of
the overview.

\emph{Step 5: summation over $z$.} Combining \eqref{eq:fourier-step1} and \eqref{eq:fourier-step2},
\begin{align*}
 \Bigl\langle g,\Bigl(I-\frac2ND+2D_z\Bigr)g\Bigr\rangle
 &\le\Bigl(3-\frac2N\Bigr)\|g\|^2+2\Bigl(1-\frac1N\Bigr)\|b_z^2g\|^2\\
 &\qquad-2\Bigl(1-\frac1N\Bigr)\|(D_z-I)g\|^2 .
\end{align*}
We sum over $z$, using that the squared norms $\|\Psi_z\|^2$ sum to $1$ and the definition of $e$. We find that the contribution of the terms $I-(2/N)D+2D_z$ in \eqref{eq:fourier-square-expansion} to \eqref{eq:fourier-score-operator} is at most
\[
 3-\frac2N+2\Bigl(1-\frac1N\Bigr)e^2-2\Bigl(1-\frac1N\Bigr)\sum_z\|(D_z-I)\Psi_z\|^2 .
\]

\emph{Step 6: the interference term.} The remaining terms contribute
\begin{align*}
 \sum_z\langle\Psi_z,(b_z^2+b_z^{\dagger2})\Psi_z\rangle
 &=2\operatorname{Re}\sum_z\langle\Psi_z,b_z^2\Psi_z\rangle\\
 &\le2\Bigl(\sum_z\|\Psi_z\|^2\Bigr)^{1/2}\Bigl(\sum_z\|b_z^2\Psi_z\|^2\Bigr)^{1/2}=2e .
\end{align*}
Here the first equality holds because $b_z^{\dagger2}$ is the adjoint of $b_z^2$, and the inequality is the Cauchy--Schwarz inequality. Adding this to the bound of Step 5 proves the bound with $e$ in place of $\lambda$.
\end{proof}

\subsection{Growth of the progress measure}\label{sec:fourier-growth}

\begin{lemma}[Growth]\label{lem:fourier-growth}
For $0\le t<q$,
\[
 \lambda(\Psi_{t+1})\le\lambda(\Psi_t)+2\sqrt{t/N} .
\]
Consequently,
\begin{equation}\label{eq:fourier-growth}
\begin{split}
 \lambda(\Psi)\le\bar\lambda_q
 &:=\frac2{\sqrt N}\sum_{t=0}^{q-1}\sqrt t\\
 &\le\frac43\cdot\frac{q^{3/2}}{\sqrt N} .
\end{split}
\end{equation}
\end{lemma}

\begin{proof}
\emph{Step 1: the vector $\mathcal L\Phi$.} It is convenient to regard the family $(b_z^2\Phi)_z$ as the single vector $\mathcal L\Phi=\sum_z\ket z\otimes b_z^2\Phi$ in $\C^N\otimes\mathcal F$, so that $\lambda(\Phi)=\|\mathcal L\Phi\|$. An operator $T$ on $\mathcal H$, acting pointwise, commutes with every $b_z$, so $\mathcal L(T\Phi)=(I\otimes T)\mathcal L\Phi$.
In particular, the unitaries $U_t$ do not change $\lambda$, and it suffices to bound the effect of one query $Q$ on $\Psi_t$.

\emph{Step 2: moving $\mathcal L$ past one query.} By \eqref{eq:fourier-query}, and because $b_z^2$ commutes with the pointwise projections $P_\varnothing$ and $P^{\mathrm c}_x$,
\[
 \mathcal LQ\Psi_t-(I\otimes Q)\mathcal L\Psi_t=\sum_z\ket z\otimes\sum_x(b_z^2\mathcal X_x-\mathcal X_xb_z^2)P^{\mathrm c}_x\Psi_t .
\]

Moreover, for fixed $z$, the vectors $(b_z^2\mathcal X_x-\mathcal X_xb_z^2)P^{\mathrm c}_x\Psi_t$ for distinct $x$ are orthogonal, because the operator $b_z^2\mathcal X_x-\mathcal X_xb_z^2$ acts only on the dependence on $f$ while $P^{\mathrm c}_x$ projects onto mutually orthogonal subspaces of $A$.

\emph{Step 3: the degree bound.} Therefore, by Step~2 and \cref{lem:fourier-identities}(iii),
\begin{align*}
 \bigl\|\mathcal LQ\Psi_t-(I\otimes Q)\mathcal L\Psi_t\bigr\|^2
 &=\sum_x\sum_z\bigl\|(b_z^2\mathcal X_x-\mathcal X_xb_z^2)P^{\mathrm c}_x\Psi_t\bigr\|^2\\
 &=\frac4N\sum_x\langle P^{\mathrm c}_x\Psi_t,(D-n_x)P^{\mathrm c}_x\Psi_t\rangle ,
\end{align*}
and, by \eqref{eq:fourier-degree} applied to $P^{\mathrm c}_x\Psi_t$ after dropping the nonpositive contribution of $-n_x$, the last expression is at most $(4t/N)\sum_x\|P^{\mathrm c}_x\Psi_t\|^2\le4t/N$. The last inequality holds because the $P^{\mathrm c}_x$ are mutually orthogonal projections and $\|\Psi_t\|=1$.

\emph{Step 4: the triangle inequality.} Since $Q$ is unitary, $\|(I\otimes Q)\mathcal L\Psi_t\|=\|\mathcal L\Psi_t\|$, and the triangle inequality, together with the bound of Step 3, gives
\[
 \lambda(Q\Psi_t)\le\lambda(\Psi_t)+2\sqrt{t/N} .
\]
This is the first claim, because $\lambda(\Psi_{t+1})=\lambda(U_{t+1}Q\Psi_t)=\lambda(Q\Psi_t)$, where the second equality holds by Step 1.

\emph{Step 5: the sum of the increments.} Initially, $\Psi_0=U_0\psi_0$ is a constant function, so every $a_x$ maps it to zero and $\lambda(\Psi_0)=0$. Summing the increments gives the middle expression in \eqref{eq:fourier-growth}. Finally, the inequality $\sqrt t\le\int_t^{t+1}\sqrt s\,ds$ gives
\[
 \sum_{t=0}^{q-1}\sqrt t\le\int_0^q\sqrt s\,ds=\frac23q^{3/2} . \qedhere
\]
\end{proof}

\subsection{The lower bound}\label{sec:fourier-theorem}

\begin{theorem}[Fourier lower bound]\label{thm:fourier-lower}
Every algorithm that makes at most $q$ queries in the Fourier ensemble satisfies
\begin{equation}\label{eq:fourier-lower-explicit}
 B(\mathcal A)\le3-\frac2N+\frac83\cdot\frac{q^{3/2}}{\sqrt N}+\frac{32}9\Bigl(1-\frac1N\Bigr)\frac{q^3}N .
\end{equation}
Consequently, if $B(\mathcal A)\ge3+\varepsilon$ for some $0<\varepsilon\le4$, then
\[
 q\ge(3\varepsilon/16)^{2/3}N^{1/3} .
\]
\end{theorem}

\begin{proof}
\emph{Step 1: the explicit bound.} By \cref{lem:fourier-score-bound}, dropping its last term,
\[
 B(\mathcal A)\le3-\frac2N+2\lambda+2\Bigl(1-\frac1N\Bigr)\lambda^2,\qquad \lambda=\lambda(\Psi)\le\frac43q^{3/2}N^{-1/2},
\]
where the bound on $\lambda$ holds by \cref{lem:fourier-growth}. The right side of the first inequality is increasing in $\lambda\ge0$, so substituting the bound on $\lambda$ gives \eqref{eq:fourier-lower-explicit}.

\emph{Step 2: a lower bound on $\lambda$.} If $B(\mathcal A)\ge3+\varepsilon$, then, after dropping $-2/N$ and bounding $1-1/N$ by $1$, the first inequality of Step 1 gives $\varepsilon\le2\lambda+2\lambda^2$. Solving the quadratic inequality gives
\begin{align*}
 \lambda&\ge\frac{\sqrt{1+2\varepsilon}-1}2=\frac{\varepsilon}{\sqrt{1+2\varepsilon}+1}\\
 &\ge\frac\varepsilon4 ,
\end{align*}
where the last step uses $\varepsilon\le4$.

Hence, by the bound on $\lambda$ in Step 1,
\[
 \frac\varepsilon4\le\frac43q^{3/2}N^{-1/2},\qquad\text{that is,}\qquad q^{3/2}\ge\frac{3\varepsilon}{16}\sqrt N ,
\]
which is the claim.
\end{proof}

The restriction $\varepsilon\le4$ costs nothing, since a score of $3+\varepsilon$ with $\varepsilon>4$ is in particular a score of at least $7$, so that $q\ge(3/4)^{2/3}N^{1/3}$ by the case $\varepsilon=4$. The proof also recovers, in full generality, the one-query bound
of~\cite{Kretschmer}. Indeed, an algorithm that makes one query of any kind, controlled or not, has $\lambda(\Psi)=0$, because at most one point is ever recorded and $b_z^2$ removes two. Hence $B(\mathcal A)\le3-2/N$ by \cref{lem:fourier-score-bound}.

\begin{remark}[Limited independence]\label{rem:fourier-independence}
For a fixed algorithm, the score is the expectation over $f$ of $\sum_zNp_f(z)\,\|\Psi_z(f)\|^2$, which is a polynomial of degree at most $2q+2$ in the values $f(x)$, because $V_f$ has degree at most $q$ and $Np_f(z)$ has degree $2$. Since $f(x)^2=1$, such a polynomial is a linear combination of the characters $\chi_S$ with $|S|\le2q+2$, and the expectation of $\chi_S$ with $S\ne\varnothing$ is $0$ under every random function whose values are $(2q+2)$-wise independent and uniform. The score of $\mathcal A$, and hence \cref{thm:fourier-lower}, is therefore unchanged when the uniformly random function is replaced by such a function.
\end{remark}

\subsection{The exact identity for one phase query}\label{sec:fourier-one-query}

For algorithms that make a single phase query, the score is given by an exact formula, which we stated in \cref{sec:overview-one-query}. An algorithm makes a \emph{single phase query} if $q=1$ and the state $U_0\psi_0$ has the control qubit $\mathsf C$ in the state $\ket1$ and the target qubit $\mathsf T$ in the state $\ket-$. For such a state, $P_\varnothing U_0\psi_0=0$, so by \eqref{eq:fourier-query} the final state is
\begin{equation}\label{eq:fourier-one-query-state}
 \Psi(f)=\sum_xf(x)\,v_x,\qquad v_x=(U_1P^{\mathrm c}_xU_0\otimes I_E)\psi_0 ,
\end{equation}
which is linear in $f$. In particular, the final state when the oracle is the character $\chi_w$ is $\Psi(\chi_w)=\sum_x\chi_w(x)v_x$, and, for every $f$,
\[
 \Psi(f)=\sum_w\widehat f(w)\Psi(\chi_w)\qquad\text{because}\qquad f=\sum_w\widehat f(w)\chi_w .
\]

\begin{proposition}[One phase query]\label{prop:fourier-one-query}
Let $\mathcal A$ make a single phase query, and let
\[
 \sigma=\frac1N\sum_z\Pr[Z=z\mid f=\chi_z]
\]
be the probability with which it identifies a uniformly random character. Then the following hold.
\begin{enumerate}
\item[(i)] The score of $\mathcal A$ is
\begin{equation}\label{eq:fourier-one-query-identity}
 B(\mathcal A)=1-\frac2N+2\sigma .
\end{equation}
\item[(ii)] In particular, $B(\mathcal A)\le3-2/N$, and equality holds if and only if $\mathcal A$ outputs $w$ with certainty whenever the oracle is $\chi_w$.
\item[(iii)] In the case of equality, $\Psi_z(f)=\widehat f(z)\Psi(\chi_z)$ for all $f$ and $z$, and Eve's conditional states are $\rho^f_z=p_f(z)\,\rho_E$. That is, the output is an exact sample from $p_f$ and is independent of Eve's register, even given $f$.
\end{enumerate}
\end{proposition}

\begin{proof}
\emph{Step 1: every term of $\Psi$ has one recorded point.} By \eqref{eq:fourier-one-query-state}, $\Psi=\sum_x\chi_{\{x\}}v_x$, so every term of $\Psi$ has exactly one recorded point. Hence $D\Psi=\Psi$, $b_z^2\Psi=0$, and
\[
 b_z\Psi=\frac1{\sqrt N}\sum_x\chi_z(x)\,\chi_\varnothing\,v_x=\frac1{\sqrt N}\Psi(\chi_z),
\]
which is a constant function.

Consequently, applying $\Pi_z$, which commutes with $b_z$ and $D$, gives $D\Psi_z=\Psi_z$, $b_z^2\Psi_z=0$, and
\[
 \|b_z\Psi_z\|^2=\frac1N\|\Pi_z\Psi(\chi_z)\|^2=\frac1N\Pr[Z=z\mid f=\chi_z] .
\]

\emph{Step 2: the exact score.} We now evaluate \eqref{eq:fourier-score-operator} exactly, using the expansion \eqref{eq:fourier-square-expansion}. For each $z$, the terms of that expansion contribute as follows.
\begin{enumerate}
\item The terms $b_z^2$ and $b_z^{\dagger2}$ contribute $\langle\Psi_z,b_z^2\Psi_z\rangle+\langle b_z^2\Psi_z,\Psi_z\rangle=0$.
\item The term $I-(2/N)D$ contributes $(1-2/N)\|\Psi_z\|^2$.
\item The term $2D_z$ contributes $2\|b_z\Psi_z\|^2$.
\end{enumerate}
Summing over $z$ gives \eqref{eq:fourier-one-query-identity}.

\emph{Step 3: the case of equality.} Since $\sigma\le1$, the score is at most $3-2/N$, with equality if and only if $\Pr[Z=w\mid f=\chi_w]=1$ for every $w$. In that case, $\Pi_z\Psi(\chi_w)=0$ for $z\ne w$, so by linearity
\[
 \Psi_z(f)=\sum_w\widehat f(w)\Pi_z\Psi(\chi_w)=\widehat f(z)\Psi(\chi_z) .
\]

Finally, $\Psi(\chi_z)=(V_{\chi_z}\otimes I_E)\psi_0$ with $V_{\chi_z}$ unitary on $A$, so its reduced state on $E$ is $\rho_E$, and, by \eqref{eq:prelim-normal-state}, Eve's conditional states are
\[
 \rho^f_z=\widehat f(z)^2\rho_E=p_f(z)\rho_E . \qedhere
\]
\end{proof}

The identity \eqref{eq:fourier-one-query-identity} is specific to phase queries. For a query whose target qubit is in a general state, the state after the query also has terms with no recorded point, and the score can be smaller than $1-2/N+2\sigma$. The example in \cref{sec:overview-one-query} has $\sigma=1$ and score $1$.

\section{The lower bound for Haar-random states}\label{sec:haar}

In this section we prove the lower bound for Haar-random states. As in the Fourier case, we bound the score in terms of a progress measure and then bound its change under a query. The identity $\sum_x\alpha_x\overline{\alpha_x}=1$ shows that counting repeated factors depends on how the polynomial is written. We use unpaired amplitudes and degrees for the corresponding Dirichlet measures, and write $\lambda_t$ for the progress measure after $t$ queries. For an algorithm given $s$ copies and making $q$ queries to $O_\phi$, we prove
\begin{align*}
 B(\mathcal A)&\le\beta_N+2\sqrt2\,\lambda_q+72\lambda_q^2,\\
 \lambda_{t+1}&\le\lambda_t+c_0\sqrt{\frac{s+2t+1}{N}}\qquad(0\le t<q),
\end{align*}
where $c_0$ is absolute and $\beta_N$ is the ideal score defined below. With no copies, $\lambda_0=0$, so summing gives $\lambda_q=O(q^{3/2}/\sqrt N)$. Exceeding the ideal score by a constant therefore requires $\Omega(N^{1/3})$ queries. The score bound uses \cref{lem:haar-subset}, and the query bound uses both \cref{lem:haar-subset,lem:haar-coherent}. We give their statements and the resulting score and query arguments here; supporting calculations and longer block estimates are in Appendix~\ref{app:haar-deferred}. The extension to unitary access is in \cref{sec:haar-unitary}.

Throughout, $\ket\phi$ is Haar-random, and we write
\[
 \ket\phi=\sum_{x\in[N]}\alpha_x\ket x,\qquad p_x=|\alpha_x|^2,\qquad \alpha_x=\sqrt{p_x}e^{\mathrm i\theta_x},
\]
so that the phases are independent and uniform and $p$ is uniform on the simplex (\cref{sec:prelim-dirichlet}). We write $\beta_N=2N/(N+1)$ for the ideal score, $e_i$ for the $i$-th unit vector of $\Z^N$, and, as in \cref{sec:prelim-dirichlet}, $|a|$ for the total parameter of a Dirichlet distribution $\mathrm{Dir}(a)$. For a vector $\gamma\in\Z^N$ we write $|\gamma|_1=\sum_i|\gamma_i|$.

Fix an experiment, which by \cref{lem:normal-form} we may assume makes a fixed sequence of $q$ controlled queries. The device register $A$ contains the control qubit $\mathsf C$ and the query register $\mathsf Q=\C\ket\bot\oplus\C^N$. Eve's register $E$ may be trivial. The final state is $\Psi(\phi)=(V_\phi\otimes I_E)\psi_0$, where
\[
 V_\phi=U_qO^{\mathrm c}_\phi U_{q-1}\cdots O^{\mathrm c}_\phi U_0 .
\]
If the algorithm receives $s$ copies of $\ket\phi$, then the initial vector is $\ket\phi^{\otimes s}\otimes\psi_0$ instead, with the copies in $s$ further registers of $A$. The proof of \cref{lem:normal-form} applies to this initial state without change. As in \cref{sec:fourier}, we write $\mathcal H=A\otimes E$ for the space in which the final state takes its values.

\subsection[Polynomials, patterns, and the level]{Polynomials, patterns, and the level}\label{sec:haar-level}

\paragraph{Polynomials.}
For a finite-dimensional space $\mathcal V$, which will be $\mathcal H$ or a tensor factor of it, let $\mathcal P_{\le d}(\mathcal V)$ be the space of functions $F$ from the unit sphere of $\C^N$ to $\mathcal V$ that are polynomials of total degree at most $d$ in the $2N$ variables $\alpha_x,\overline{\alpha_x}$, with the inner product and norm
\[
 \langle F,G\rangle=\E_\phi\langle F(\phi),G(\phi)\rangle_\mathcal V,\qquad \|F\|^2=\E_\phi\|F(\phi)\|_\mathcal V^2 .
\]
Two polynomials are identified when they agree on the sphere. We regard the elements of $\mathcal P_{\le d}(\mathcal V)$ as vector-valued polynomials; they need not be normalized for each $\phi$. A final state is a unit vector in this space, because $\|\Psi(\phi)\|=1$ for every $\phi$.

\begin{lemma}[Degree growth]\label{lem:haar-degree}
After $t$ queries, the state of an algorithm that received $s$ copies lies in $\mathcal P_{\le s+2t}(\mathcal H)$.
\end{lemma}

\begin{proof}
The initial vector
\[
 \ket\phi^{\otimes s}\otimes\psi_0=\sum_{x\in[N]^s}\alpha_{x_1}\cdots\alpha_{x_s}\ket{x_1,\dots,x_s}\otimes\psi_0
\]
has degree $s$. A fixed unitary maps each coordinate of the state to a fixed linear combination of coordinates and does not raise the degree. By \eqref{eq:prelim-oracle-expansion}, every matrix entry of $O_\phi$, and hence of $O^{\mathrm c}_\phi$, is a polynomial of degree at most two, so a query raises the degree by at most two.
\end{proof}

\paragraph{Patterns and sectors.}
We first group monomials by their unpaired amplitudes, then expand their remaining factors in $p$ in orthogonal polynomials. The unpaired factors determine the Dirichlet distribution used for the inner product. A \emph{pattern} is a vector $\gamma\in\Z^N$, and its \emph{phase monomial} is
\begin{align*}
 u_\gamma(\phi)&=\prod_{i:\gamma_i>0}\alpha_i^{\gamma_i}\prod_{i:\gamma_i<0}\overline{\alpha_i}^{\,-\gamma_i}\\
 &=\Bigl(\prod_ip_i^{|\gamma_i|/2}\Bigr)e^{\mathrm i\sum_i\gamma_i\theta_i} .
\end{align*}
A positive entry $\gamma_i$ stands for $\gamma_i$ unpaired factors $\alpha_i$, and a negative entry for $|\gamma_i|$ factors $\overline{\alpha_i}$.

Every monomial in the variables $\alpha,\overline\alpha$ factors uniquely as $u_\gamma\prod_ip_i^{m_i}$ with $m_i\ge0$.
Explicitly, if it contains $\alpha_i^{b_i}\overline{\alpha_i}^{c_i}$, then $\gamma_i=b_i-c_i$ and $m_i=\min(b_i,c_i)$. Its total degree is $|\gamma|_1+2\sum_im_i$.

The \emph{sector} of $\gamma$ is the set of functions $u_\gamma h(p)$, so that $h$ depends on $\phi$ only through the probabilities $p_1,\dots,p_N$.

To each pattern we attach the Dirichlet parameters
\[
 a(\gamma)=(1+|\gamma_1|,\dots,1+|\gamma_N|),\qquad |a(\gamma)|=N+|\gamma|_1,
\]
and the weight
\[
 M_\gamma=\E_\phi|u_\gamma|^2=\frac{\prod_i|\gamma_i|!}{(N)_{|\gamma|_1}} ,
\]
where the value of $M_\gamma$ follows from the identity $|u_\gamma|^2=\prod_ip_i^{|\gamma_i|}$ and \cref{lem:dirichlet}(i).

Within the sector of $\gamma$, we use the inner product of square-integrable functions under $\mathrm{Dir}(a(\gamma))$. We write $\mathcal R^a_r$ for the space of polynomials in $p$ of degree at most $r$ that are orthogonal, under $\mathrm{Dir}(a)$, to all polynomials in $p$ of degree at most $r-1$, with the convention $\mathcal R^a_{-1}=\{0\}$; its elements have \emph{exact degree $r$} for $\mathrm{Dir}(a)$. We write $\Pi^a_{\le r}$ and $\Pi^a_r$ for the orthogonal projections onto the polynomials of degree at most $r$ and onto $\mathcal R^a_r$ in that space. These notions depend on $a$ only through the inner product $\langle\cdot,\cdot\rangle_a$.

A polynomial in $p$ with values in $\mathcal V$ is treated coefficientwise with respect to a fixed orthonormal basis of $\mathcal V$, and all statements about polynomials in $p$ below extend to $\mathcal V$-valued ones in this way.

\begin{lemma}[Blocks]\label{lem:haar-blocks}
The following hold.
\begin{enumerate}
\item[(i)] For every $d\ge0$,
\begin{equation}\label{eq:haar-blocks}
 \mathcal P_{\le d}(\mathcal V)=\bigoplus_{|\gamma|_1+2r\le d}u_\gamma \mathcal R^{a(\gamma)}_r\otimes \mathcal V ,
\end{equation}
where the summands, which we call \emph{blocks}, are pairwise orthogonal.
\item[(ii)] For every pattern $\gamma$ and every function $h$ of $p$ with $\|h\|_{a(\gamma)}<\infty$, in particular for every polynomial in $p$,
\begin{equation}\label{eq:haar-norm-conversion}
 \|u_\gamma h\|^2=M_\gamma\,\|h\|_{a(\gamma)}^2 .
\end{equation}
\end{enumerate}
\end{lemma}

\begin{proof}
\emph{Step 1: decomposition into orthogonal sectors.} By the factorization of monomials, every element of $\mathcal P_{\le d}(\mathcal V)$ is a sum of terms $u_\gamma h$ with $h$ a polynomial in $p$ of degree at most $(d-|\gamma|_1)/2$. We treat scalar coefficients; the $\mathcal V$-valued case follows coefficientwise.

Terms with distinct patterns are orthogonal. Indeed, for $\gamma\ne\gamma'$ and functions $h,h'$ of $p$, conditioning on $p$ and averaging over the phases gives
\[
 \langle u_\gamma h,u_{\gamma'}h'\rangle=\E_p\Bigl[\overline{h(p)}h'(p)\prod_ip_i^{(|\gamma_i|+|\gamma'_i|)/2}\,\E_\theta e^{\mathrm i\sum_i(\gamma'_i-\gamma_i)\theta_i}\Bigr]=0 ,
\]
because the phases are independent and uniform and some $\gamma'_i-\gamma_i$ is nonzero.

\emph{Step 2: the norm conversion.} For \eqref{eq:haar-norm-conversion}, averaging over the phases gives
\[
 \|u_\gamma h\|^2=\E_{\mathrm{Dir}(1,\dots,1)}\Bigl[\prod_ip_i^{|\gamma_i|}|h|^2\Bigr].
\]
Applying \cref{lem:dirichlet}(ii) once for each of the $|\gamma|_1$ factors $p_i$ raises the parameter from $(1,\dots,1)$ to $a(\gamma)$, and each application contributes the factor $a_i/|a|$ for the parameters current at that moment. For the outcome $i$ these numerators are $1,2,\dots,|\gamma_i|$, and, in some order, the denominators are $N,N+1,\dots,N+|\gamma|_1-1$, so the product of all the factors is $M_\gamma$.

\emph{Step 3: orthogonality within a sector.} Finally, within a sector, decompose $h=\sum_rh_r$ with $h_r\in \mathcal R^{a(\gamma)}_r$. The same computation gives
\[
 \langle u_\gamma h_r,u_\gamma h_t\rangle=M_\gamma\langle h_r,h_t\rangle_{a(\gamma)}=0\qquad\text{for }r\ne t,
\]
and $h_r$ is represented by a polynomial of degree at most $r$, so the degree count of Step~1 gives the decomposition \eqref{eq:haar-blocks}.
\end{proof}

When a query changes one entry of the pattern by one unit, we shall use \eqref{eq:haar-norm-conversion} together with \cref{lem:dirichlet}(ii) in the following form: for every pattern $\gamma$ and outcome $i$ with $|\gamma_i\pm1|=|\gamma_i|+1$,
\begin{equation}\label{eq:haar-density-conversion}
 M_\gamma\,p_i\,d\mathrm{Dir}(a(\gamma))=M_{\gamma\pm e_i}\,d\mathrm{Dir}(a(\gamma\pm e_i)),
\end{equation}
and symmetrically with $p_i^{-1}$ when the magnitude of the entry decreases; both sides are $|u_{\gamma\pm e_i}|^2$ times the distribution of $p$ under the Haar measure.

\paragraph{The level.}
The level adds the degree under the Dirichlet measure to the number of unpaired factors beyond the first at each outcome. The first factor is not counted, because the honest sampler has one unpaired factor in each output branch. A block in $\mathcal P_{\le d}$ satisfies $|\gamma|_1+2r\le d$, while its level is $\operatorname{rep}(\gamma)+r$.

\begin{definition}[Level]\label{def:haar-level}
For a pattern $\gamma$ let
\[
 \operatorname{rep}(\gamma)=\sum_i\max(|\gamma_i|-1,0)
\]
count the occurrences beyond the first at each outcome. The \emph{level} of the block indexed by $(\gamma,r)$ in \eqref{eq:haar-blocks} is $\operatorname{rep}(\gamma)+r$. Let $P_k$ be the orthogonal projection onto the sum of the blocks of level $k$, and let $C=\sum_kkP_k$ be the \emph{level operator}. The progress measure of a polynomial $F$ is
\[
 \lambda(F)=\|C^{1/2}F\|=\Bigl(\sum_kk\,\|P_kF\|^2\Bigr)^{1/2}.
\]
\end{definition}

A block has level zero when its pattern is square-free ($|\gamma_i|\le1$ for all $i$) and its polynomial factor is constant. For example, $p_z=1/N+(p_z-1/N)$ splits into levels zero and one. In the sector of $\alpha_z$, the Dirichlet measure is $\mathrm{Dir}(a(e_z))$, under which $p_z$ has mean $2/(N+1)$. Hence
\[
 \alpha_zp_z=\frac2{N+1}\alpha_z+\alpha_z\left(p_z-\frac2{N+1}\right).
\]
The two summands again have levels zero and one. These are decompositions of vector-valued polynomials. The normalization $\sum_zp_z=1$ has level zero, while $\alpha_z^2$ has level one.

Three facts about these projections are used repeatedly.
\begin{enumerate}
\item Each $P_k$ acts on $\mathcal P_{\le d}(\mathcal V)$ as the sum of the projections onto the blocks of \eqref{eq:haar-blocks} of level $k$, so it maps $\mathcal P_{\le d}(\mathcal V)$ to itself.
\item A $\phi$-independent operator on $\mathcal V$ acts on the coefficients only, hence maps each block to itself and commutes with every $P_k$ and with $C$.
\item Adjoining or removing a register in a fixed state, such as the query register in the state $\ket\bot$, changes neither the pattern nor the polynomial factor of any block, so the map $J:G\mapsto G\otimes\ket\bot$ and its adjoint also commute with every $P_k$.
\end{enumerate}

\paragraph{Level zero and the ideal score.}
The score is the quadratic form, on $\mathcal P_{\le d}(\mathcal H)$, of the multiplication operator $X=\sum_zNp_z\Pi_z$, defined on all square-integrable $\mathcal H$-valued functions of $\phi$ by
\[
 (XF)(\phi)=\sum_zNp_z(\phi)\Pi_zF(\phi) .
\]
It raises the degree of a polynomial by two, and only its quadratic form is used. Indeed,
\begin{equation}\label{eq:haar-score-form}
 \langle F,XF\rangle=N\sum_z\E_\phi\bigl[p_z\|\Pi_zF(\phi)\|^2\bigr],
\end{equation}
which is the expected score \eqref{eq:prelim-score-branches} when $F$ is a final state. Since $X$ multiplies by a function of $p$ and by projections on $\mathcal H$, it preserves the pattern of every block and its branch, and it can change only the degree under the Dirichlet measure.

\begin{proposition}[Level zero and the ideal score]\label{prop:haar-level-zero}
If $F_0=P_0F_0$, then
\[
 \langle F_0,XF_0\rangle\le\beta_N\|F_0\|^2,
\]
with equality for the honest sampler. More precisely, for a square-free pattern $\gamma$, an output $z$, and a constant $v\in \mathcal H$,
\begin{align*}
 \langle u_\gamma\Pi_zv,X\,u_\gamma\Pi_zv\rangle&=\frac{N(1+|\gamma_z|)}{N+|\gamma|_1}\,\|u_\gamma\Pi_zv\|^2\\
 &\le
 \begin{cases}
 \|u_\gamma\Pi_zv\|^2 & \text{if }\gamma_z=0,\\
 \beta_N\|u_\gamma\Pi_zv\|^2 & \text{if }|\gamma_z|=1.
 \end{cases}
\end{align*}
\end{proposition}

\begin{proof}
\emph{Step 1: reduction to single terms.} A polynomial of level zero is a sum of terms $u_\gamma\Pi_zv$ with $\gamma$ square-free and $v\in \mathcal H$ constant, and we may take the $v$ from a fixed orthonormal basis of $\mathcal H$ adapted to the decomposition $\mathcal H=\bigoplus_z\Pi_z\mathcal H$. Distinct triples $(z,\gamma,v)$ give orthogonal terms, and $X$ preserves all three data, so the quadratic form is a sum over the terms.

For one term, by \eqref{eq:haar-norm-conversion} and the mean $\E_{a(\gamma)}p_z=(1+|\gamma_z|)/|a(\gamma)|$ from \cref{lem:dirichlet}(i), we have
\begin{align*}
 \langle u_\gamma\Pi_zv,X\,u_\gamma\Pi_zv\rangle&=N\,M_\gamma\,\E_{a(\gamma)}[p_z]\,\|\Pi_zv\|^2\\
 &=\frac{N(1+|\gamma_z|)}{N+|\gamma|_1}\,\|u_\gamma\Pi_zv\|^2 .
\end{align*}

\emph{Step 2: the two bounds.} If $\gamma_z=0$, then the ratio is $N/(N+|\gamma|_1)\le1\le\beta_N$. If $|\gamma_z|=1$, then $|\gamma|_1\ge1$ and the ratio is at most $2N/(N+1)=\beta_N$.

Finally, the honest sampler prepares $\ket\phi=O_\phi\ket\bot$ with its single query and measures it, so its output-$z$ branch is $\alpha_z\ket z$, of pattern $e_z$ with constant polynomial factor, for which the ratio equals $\beta_N$.
\end{proof}

\begin{lemma}[The level of the copies]\label{lem:haar-copies}
For every $s\ge0$,
\[
 \lambda(\ket\phi^{\otimes s})^2=\frac{s(s-1)}{N+s-1} .
\]
\end{lemma}

The proof is in Appendix~\ref{app:haar-copies}.

The identity follows by expanding in the basis of the copy registers. Conditional on $p$, the squared amplitudes give $s$ independent samples, so the squared progress measure of the copies is the expected number of samples that repeat an earlier outcome.

\paragraph{The main theorem.}
Combining the score and growth estimates with \cref{lem:haar-copies} gives us the following theorem.

\begin{theorem}[Haar lower bound for the canonical oracle]\label{thm:haar-lower}
Let an algorithm receive $s$ copies of $\ket\phi$ and make at most $q$ queries to $O_\phi$. Then
\[
 \lambda(\Psi)\le\Lambda_{s,q}\qquad\text{and}\qquad B(\mathcal A)\le\beta_N+2\sqrt2\,\Lambda_{s,q}+72\,\Lambda_{s,q}^2,
\]
where
\[
 \Lambda_{s,q}=\sqrt{\frac{s(s-1)}{N+s-1}}+\frac{c_0}{\sqrt N}\sum_{t=0}^{q-1}\sqrt{s+2t+1} .
\]
\end{theorem}

\begin{proof}
\emph{Step 1: the initial level.} Let $\Psi_t$ be the state after $t$ queries and the fixed unitary that follows them, as in \cref{sec:fourier-recording}, so that $\Psi_0=U_0(\ket\phi^{\otimes s}\otimes\psi_0)$ and $\Psi_q=\Psi$. Each $\Psi_t$ is a unit vector in $\mathcal P_{\le s+2t}(\mathcal H)$ by \cref{lem:haar-degree}. The fixed unitaries commute with $C$ and preserve the norm, and adjoining the $\phi$-independent vector $\psi_0$ does not change the level, so, by \cref{lem:haar-copies},
\[
 \lambda(\Psi_0)^2=\frac{s(s-1)}{N+s-1} .
\]

\emph{Step 2: the effect of the queries.} By the growth estimate at degree $s+2t$,
\[
 \lambda(\Psi_{t+1})\le\lambda(\Psi_t)+c_0\sqrt{\frac{s+2t+1}N} .
\]
Summing over $t$ gives $\lambda(\Psi)\le\Lambda_{s,q}$, which is the first claim.

Finally, the score estimate applies to $\Psi$, and its right side is increasing in $\lambda(\Psi)$.
\end{proof}

\begin{corollary}[Query and copy bounds]\label{cor:haar-birthday}
Let $0<\varepsilon\le1$, and let an algorithm as in \cref{thm:haar-lower} have $B(\mathcal A)\ge2+\varepsilon$. Then
\[
 s+q\sqrt{s+q+1}\ge c_1\varepsilon\sqrt N
\]
for an absolute constant $c_1>0$. In particular, the following hold.
\begin{enumerate}
\item[(i)] With no copies, $q\ge c_2\varepsilon^{2/3}N^{1/3}$ for an absolute constant $c_2>0$.
\item[(ii)] With no queries, $s\ge c_1\varepsilon\sqrt N$.
\end{enumerate}
\end{corollary}

\begin{proof}
\emph{Step 1: a lower and an upper bound on $\Lambda$.} Since $\beta_N\le2$, the theorem gives
\[
 \varepsilon\le2\sqrt2\Lambda+72\Lambda^2,\qquad\text{where }\Lambda=\Lambda_{s,q} .
\]
If $\Lambda\le1$, then this gives $\varepsilon\le75\Lambda$. If $\Lambda>1$, then $\Lambda>\varepsilon$. In both cases $\Lambda\ge\varepsilon/75$.

On the other hand, each summand is at most $\sqrt{s+2q+1}$, and the empty sum is zero when $q=0$. Bounding $s(s-1)/(N+s-1)$ by $s^2/N$ and using $s+2q+1\le2(s+q+1)$, we obtain
\begin{align*}
 \Lambda&\le\frac s{\sqrt N}+\frac{c_0\,q\sqrt{s+2q+1}}{\sqrt N}\\
 &\le\frac{(1+\sqrt2c_0)\bigl(s+q\sqrt{s+q+1}\bigr)}{\sqrt N} .
\end{align*}

\emph{Step 2: the three claims.} The two bounds give the first claim with \mbox{$c_1=1/(75(1+\sqrt2c_0))$}. With $s=0$, the first claim reads $q\sqrt{q+1}\ge c_1\varepsilon\sqrt N$, and the bound $q\sqrt{q+1}\le\sqrt2\,q^{3/2}$ gives item~(i). With $q=0$, the first claim is item~(ii).
\end{proof}

Both extremes of the corollary have the right exponent of $N$: $O(\sqrt N)$ measured copies suffice to exceed the ideal score by a constant, since a repeated outcome identifies an outcome of large probability~\cite{Kretschmer}, and $O(N^{1/3})$ queries suffice by \cref{thm:squaring-scores} in \cref{sec:squaring}. The exponent of $\varepsilon$ is not known to be sharp.

\subsection{Dirichlet energy and concentration}\label{sec:haar-energy}

To bound the score, we first bound the mass of a set of outcomes weighted by a polynomial in $p$. Throughout this subsection, $a_i\ge1$. We prove the following lemma using integration by parts.

\begin{lemma}[The mass of a set of outcomes]\label{lem:haar-subset}
Let $S\subseteq[N]$, and write $p(S)=\sum_{i\in S}p_i$ and $a(S)=\sum_{i\in S}a_i$. For every polynomial $h$ in $p$ of degree at most $r$,
\[
 \E_a\bigl[p(S)|h|^2\bigr]\le8\,\frac{a(S)+r}{|a|}\,\|h\|_a^2 .
\]
\end{lemma}

\paragraph{Tangent derivatives.}
The partial derivatives $\partial_i=\partial/\partial p_i$ of a polynomial can depend on how it is represented away from the simplex. We instead differentiate along $e_i-p$, which points from $p$ toward vertex $i$, and define
\begin{equation}\label{eq:haar-tangent}
 T_ig=\partial_ig-\sum_{j=1}^Np_j\,\partial_jg .
\end{equation}
For $0\le t\le1$, the point $(1-t)p+te_i$ lies in the simplex. Since $T_ig(p)$ is the derivative at $t=0$ of the polynomial $g((1-t)p+te_i)$, it depends only on the restriction of $g$ to the simplex. The operators satisfy $\sum_ip_iT_i=0$.

\begin{lemma}[Integration by parts]\label{lem:haar-ibp}
For all polynomials $g,v$ and all $i\in[N]$,
\begin{equation}\label{eq:haar-ibp}
 \E_a\bigl[p_i\,(T_ig)\,v\bigr]=\E_a\bigl[g\bigl((|a|p_i-a_i)v-p_i\,T_iv\bigr)\bigr] .
\end{equation}
\end{lemma}

We verify the identity on monomials using the Dirichlet moment formula, which avoids boundary terms. It is bilinear, with no conjugation, and applies to conjugates by substituting $\bar v$ for $v$.

The proof is in Appendix~\ref{app:haar-ibp}.

\paragraph{{\boldmath The form $\Gamma$ and the operator $L_a$.}}
For polynomials $f,g$, define the form $\Gamma$ and the operator $L_a$ by
\begin{equation}\label{eq:haar-energy-def}
\begin{aligned}
 \Gamma(f,g)&=\sum_{i,j}p_i(\delta_{ij}-p_j)\,\overline{\partial_if}\,\partial_jg,\\
 L_ag&=\sum_{i,j}p_i(\delta_{ij}-p_j)\,\partial_i\partial_jg+\sum_i(a_i-|a|p_i)\,\partial_ig .
\end{aligned}
\end{equation}

Expanding the square shows that on the simplex
\begin{align*}
 \Gamma(g,g)&=\sum_ip_i|\partial_ig|^2-\Bigl|\sum_jp_j\partial_jg\Bigr|^2\\
 &=\sum_ip_i|T_ig|^2\ge0,\\
 \text{and more generally}\quad\Gamma(f,g)&=\sum_ip_i\overline{T_if}\,T_ig .
\end{align*}

On the simplex the operator $L_a$ can be written with the tangent derivatives alone,
\begin{equation}\label{eq:haar-La-tangent}
 L_ag=\sum_i\bigl(p_iT_i^2g+a_iT_ig\bigr),
\end{equation}
so it is well defined on functions on the simplex, since each $T_i$ is. The calculation is in Appendix~\ref{app:haar-La}.

Since $L_ap_i=a_i-|a|p_i$, summing over $i\in S$ gives $L_ap(S)=a(S)-|a|p(S)$.

\begin{lemma}[Symmetry and eigenvalues]\label{lem:haar-energy}
For all polynomials $f,g$:
\begin{enumerate}
\item[(i)] the identity $\E_a\Gamma(f,g)=-\langle f,L_ag\rangle_a$ holds; in particular, the operator $L_a$ is symmetric for $\langle\cdot,\cdot\rangle_a$ and $\langle g,L_ag\rangle_a\le0$;
\item[(ii)] the operator $L_a$ maps polynomials of degree at most $m$ to polynomials of degree at most $m$, and $L_ag=-r(|a|+r-1)g$ for $g\in \mathcal R^a_r$;
\item[(iii)] for every polynomial $h$ of degree at most $r$,
\[
 \E_a\Gamma(h,h)\le r(|a|+r-1)\,\|h\|_a^2 .
\]
\end{enumerate}
\end{lemma}

The proof is in Appendix~\ref{app:haar-energy}.

The proof of \cref{lem:haar-subset} is in Appendix~\ref{app:haar-subset}.

\subsection{The score in terms of the level}\label{sec:haar-extraction}

For a unit polynomial $F$, write $F_0=P_0F$ and $F_+=F-F_0$. Expanding along this decomposition gives
\begin{equation}\label{eq:haar-score-split}
 \langle F,XF\rangle=\langle F_0,XF_0\rangle+2\operatorname{Re}\langle F_0,XF_+\rangle+\langle F_+,XF_+\rangle .
\end{equation}
The first term is at most the ideal score by \cref{prop:haar-level-zero}. We bound the cross term by $O(\lambda(F))$ and the last term by $O(\lambda(F)^2)$, using \cref{lem:haar-subset} for the latter.

\begin{lemma}[The cross term]\label{lem:haar-cross}
For every $F\in\mathcal P_{\le d}(\mathcal H)$ with $F_0=P_0F$ and $F_+=F-F_0$,
\[
 |\langle F_0,XF_+\rangle|\le\sqrt2\,\|F_0\|\,\|F_+\|\qquad\text{and}\qquad\|F_+\|\le\lambda(F).
\]
\end{lemma}

\begin{proof}
\emph{Step 1: reduction to one sector.} Since $X$ preserves branches and patterns and $F_0$ lies in sectors with square-free patterns and constant polynomial factors, the pairing is a sum over such sectors. In one such sector, write $a=a(\gamma)$. Let $v$ be the constant coefficient of $\Pi_zF_0$ and $g$ the polynomial factor of $\Pi_zF_+$, with $\E_ag=0$ coefficientwise.
Then $\E_a\langle v,g\rangle=0$, so only the centered part of $p_z$ contributes:
\begin{align*}
 \bigl|N\,\E_a[p_z\langle v,g\rangle]\bigr|&=\bigl|N\,\E_a[(p_z-\E_ap_z)\langle v,g\rangle]\bigr|\\
 &\le N\sqrt{\operatorname{Var}_a(p_z)}\,\|v\|\,\|g\|_a .
\end{align*}

\emph{Step 2: the variance and the first claim.} By \cref{lem:dirichlet}(i), with $a_z\le2$ and $|a|\ge N$,
\begin{align*}
 N^2\operatorname{Var}_a(p_z)&=\frac{N^2a_z(|a|-a_z)}{|a|^2(|a|+1)}\\
 &\le\frac{2N^2}{|a|(|a|+1)}\le2 .
\end{align*}
Multiplying by $M_\gamma$ and applying the Cauchy--Schwarz inequality over the sectors gives the first claim.

For the second claim, we have
\begin{align*}
 \|F_+\|^2&=\sum_{k\ge1}\|P_kF\|^2\\
 &\le\sum_{k\ge1}k\|P_kF\|^2=\lambda(F)^2 . \qedhere
\end{align*}
\end{proof}

\begin{lemma}[The positive-level piece]\label{lem:haar-global}
For every $F\in\mathcal P_{\le d}(\mathcal H)$ with $F_+=F-P_0F$,
\[
 \langle F_+,XF_+\rangle\le72\,\lambda(F)^2 .
\]
\end{lemma}

\begin{proof}
\emph{Step 1: reduction to one sector.} Fix a branch $z$ and a pattern $\gamma$, and write $a=a(\gamma)$. Let the component of $\Pi_zF_+$ in this sector be
\[
 u_\gamma\sum_rh_r,\qquad h_r\in \mathcal R^a_r,
\]
with $\mathcal H$-valued coefficients treated coefficientwise. Since $X$ preserves branches and patterns, the quadratic form is a sum over sectors, and within a sector it is, by \eqref{eq:haar-norm-conversion},
\[
 N\,M_\gamma\,\E_a\Bigl[p_z\Bigl\|\sum_rh_r\Bigr\|^2\Bigr] .
\]

\emph{Step 2: multiplication by $p_z$.} Multiplication by $p_z$ connects only equal or adjacent degrees under the Dirichlet measure. Indeed, $p_zh_t$ has degree at most $t+1$, so $\langle h_r,p_zh_t\rangle_a=0$ for $r>t+1$, and symmetry gives the same conclusion for $t>r+1$. Write $w_r=\E_a[p_z\|h_r\|^2]$. Applying the Cauchy--Schwarz inequality to $\sqrt{p_z}h_r$ and $\sqrt{p_z}h_{r\pm1}$ gives $|\langle h_r,p_zh_{r\pm1}\rangle_a|\le\sqrt{w_rw_{r\pm1}}$. Since $2\sqrt{w_rw_{r+1}}\le w_r+w_{r+1}$, we obtain
\begin{align*}
 \E_a\Bigl[p_z\Bigl\|\sum_rh_r\Bigr\|^2\Bigr]&\le\sum_rw_r+2\sum_r\sqrt{w_rw_{r+1}}\\
 &\le3\sum_rw_r .
\end{align*}

\emph{Step 3: comparison with the level.} Each $w_r$ is the quantity bounded in \cref{lem:haar-subset} for the set $S=\{z\}$, whose parameter is $a_z=1+|\gamma_z|$, so that lemma and $|a|\ge N$ give
\[
 Nw_r\le8(1+|\gamma_z|+r)\|h_r\|_a^2 .
\]
Now compare with the level $k=\operatorname{rep}(\gamma)+r$ of the block $(\gamma,r)$, which is at least $1$ for every block of $F_+$. Since $|\gamma_z|\le1+\max(|\gamma_z|-1,0)\le1+\operatorname{rep}(\gamma)$, we have $1+|\gamma_z|+r\le2+k\le3k$.
Here the level does not count the first occurrence of $z$ in the pattern, as in \cref{prop:haar-level-zero}, but it counts every further occurrence.

Finally, summing over the sectors, which are orthogonal, converting back by $M_\gamma$, and using that $\Pi_z$ commutes with every $P_k$, we obtain
\[
 \langle F_+,XF_+\rangle\le24\sum_{k\ge1}3k\,\|P_kF\|^2=72\,\lambda(F)^2 . \qedhere
\]
\end{proof}

\begin{theorem}[Score bound]\label{thm:haar-extraction}
For every unit vector $F\in\mathcal P_{\le d}(\mathcal H)$,
\[
 \langle F,XF\rangle\le\beta_N+2\sqrt2\,\lambda(F)+72\,\lambda(F)^2 .
\]
\end{theorem}

\begin{proof}
We bound the three terms in \eqref{eq:haar-score-split} as follows.
\begin{enumerate}
\item The first term is at most $\beta_N\|F_0\|^2\le\beta_N$ by \cref{prop:haar-level-zero}.
\item The second term is at most $2\sqrt2\,\|F_0\|\,\|F_+\|\le2\sqrt2\lambda(F)$ by \cref{lem:haar-cross}.
\item The third term is at most $72\lambda(F)^2$ by \cref{lem:haar-global}. \qedhere
\end{enumerate}
\end{proof}

\subsection[Projection estimates for Dirichlet polynomials]{Projection estimates for Dirichlet polynomials}\label{sec:haar-lowering}

We next prove the estimates needed to bound the effect of a query on the level. Multiplying by an amplitude changes the pattern and therefore the Dirichlet inner product. Fix $a_i\ge1$. Weighting $\mathrm{Dir}(a)$ by $p_i$ and normalizing gives $\mathrm{Dir}(a+e_i)$. For $g\in\mathcal R^a_r$ with $r\ge1$, the posterior decomposition has only degrees $r$ and $r-1$. With weights $a_i/|a|$, the sum of the lower-degree squared norms equals $r/(|a|+2r-1)$ times $\|g\|_a^2$, and $r/(|a|+2r-1)\le r/N$ when $|a|\ge N$. We then use this identity to bound the degree-$r+1$ part of $\sum_i p_ih_i$ for polynomials $h_i$ of degree at most $r$.

\begin{lemma}[Two adjacent degrees]\label{lem:haar-adjacent}
Let $i\in[N]$.
\begin{enumerate}
\item[(i)] If $g\in \mathcal R^a_r$, then $g\in \mathcal R^{a+e_i}_r\oplus \mathcal R^{a+e_i}_{r-1}$.
\item[(ii)] Conversely, if $h\in \mathcal R^{a+e_i}_r$, then $p_ih\in \mathcal R^a_{r+1}\oplus \mathcal R^a_r$.
\end{enumerate}
\end{lemma}

\begin{proof}
For the first claim, $g$ has degree at most $r$, so it suffices to show that $g$ is orthogonal under $\mathrm{Dir}(a+e_i)$ to every polynomial $v$ of degree at most $r-2$. By \cref{lem:dirichlet}(ii),
\[
 \langle v,g\rangle_{a+e_i}=\frac{|a|}{a_i}\E_a[p_i\bar vg]=0,
\]
because $p_i\bar v$ has degree at most $r-1$ and $g\in \mathcal R^a_r$.

For the second claim, $p_ih$ has degree at most $r+1$, and for $v$ of degree at most $r-1$,
\[
 \langle v,p_ih\rangle_a=\frac{a_i}{|a|}\langle v,h\rangle_{a+e_i}=0
\]
because $h\in \mathcal R^{a+e_i}_r$.
\end{proof}

\begin{theorem}[The projection onto the lower degree]\label{thm:haar-projection}
Let $g\in \mathcal R^a_r$ with $r\ge1$ and let $i\in[N]$. Then
\[
 \Pi^{a+e_i}_{\le r-1}\,g=\frac{T_ig+rg}{|a|+2r-1}\in \mathcal R^{a+e_i}_{r-1} .
\]
\end{theorem}

\begin{proof}
\emph{Step 1: the degree.} Choose a representative of $g$ of degree at most $r$ and split off its homogeneous part $g_r$ of degree $r$, writing $g=g_r+g_{<r}$. Euler's identity $\sum_jp_j\partial_jg_r=rg_r$ shows that the top parts cancel in
\[
 T_ig+rg=\partial_ig+\Bigl(rg_{<r}-\sum_jp_j\partial_jg_{<r}\Bigr),
\]
which therefore has degree at most $r-1$.

\emph{Step 2: the pairing with polynomials of lower degree.} We next show that for every polynomial $v$ of degree at most $r-1$,
\begin{equation}\label{eq:haar-key-pairing}
 \E_a\bigl[p_i(T_ig+rg)v\bigr]=(|a|+2r-1)\,\E_a[p_igv] .
\end{equation}
By \eqref{eq:haar-ibp},
\begin{align*}
 \E_a[p_i(T_ig)v]&=|a|\,\E_a[p_igv]-a_i\E_a[gv]\\
 &\quad-\E_a[g\,p_i\partial_iv]+\E_a\Bigl[g\,p_i\sum_jp_j\partial_jv\Bigr].
\end{align*}
The second and third terms vanish, because $v$ and $p_i\partial_iv$ have degree at most $r-1$ and $g\in \mathcal R^a_r$.

In the last term, write $w=\sum_jp_j\partial_jv-(r-1)v$. If $v$ is split into its homogeneous part of degree $r-1$ and the rest, then Euler's identity shows that $w$ has degree at most $r-2$. Hence $p_iw$ has degree at most $r-1$ and is orthogonal to $g$, and the last term equals $(r-1)\E_a[p_igv]$. Adding $r\E_a[p_igv]$ to both sides gives \eqref{eq:haar-key-pairing}.

\emph{Step 3: the projection.} Finally, by \cref{lem:dirichlet}(ii), $\E_a[p_i\,\cdot\,]$ is $a_i/|a|$ times the expectation under $\mathrm{Dir}(a+e_i)$, so \eqref{eq:haar-key-pairing} with $\bar v$ in place of $v$ states that
\[
 g-\frac{T_ig+rg}{|a|+2r-1}
\]
is orthogonal under $\mathrm{Dir}(a+e_i)$ to every polynomial of degree at most $r-1$. Since this quotient has degree at most $r-1$, it is the orthogonal projection of $g$ onto that space.

Moreover, the projection lies in $\mathcal R^{a+e_i}_{r-1}$, and not only among the polynomials of degree at most $r-1$, because, by \cref{lem:haar-adjacent}, the polynomial $g$ has no posterior components of degree $r-2$ or lower.
\end{proof}

\begin{theorem}[Average squared norm of the lower-degree projection]\label{thm:lowering}
For every $r\ge1$ and $g\in \mathcal R^a_r$,
\[
 \sum_{i=1}^N\frac{a_i}{|a|}\,\bigl\|\Pi^{a+e_i}_{\le r-1}\,g\bigr\|_{a+e_i}^2=\frac r{|a|+2r-1}\,\|g\|_a^2 .
\]
\end{theorem}

\begin{proof}
\emph{Step 1: reduction to a prior expectation.} Write $b=|a|+2r-1$. If we substitute \cref{thm:haar-projection} and convert each posterior norm to a prior expectation by \cref{lem:dirichlet}(ii), then the weight $a_i/|a|$ cancels against the conversion factor $|a|/a_i$:
\[
 \sum_i\frac{a_i}{|a|}\bigl\|\Pi^{a+e_i}_{\le r-1}g\bigr\|_{a+e_i}^2=\frac1{b^2}\,\E_a\Bigl[\sum_ip_i|T_ig+rg|^2\Bigr].
\]

On the simplex, expanding the square gives
\begin{align*}
 \sum_ip_i|T_ig+rg|^2&=\Gamma(g,g)+2r\operatorname{Re}\Bigl(\bar g\sum_ip_iT_ig\Bigr)+r^2|g|^2\\
 &=\Gamma(g,g)+r^2|g|^2,
\end{align*}
because $\sum_ip_iT_ig=0$ and $\sum_ip_i=1$.

\emph{Step 2: the expectation.} By \cref{lem:haar-energy}(i) and (ii),
\[
 \E_a\Gamma(g,g)=r(|a|+r-1)\|g\|_a^2 ,
\]
and $r(|a|+r-1)+r^2=rb$. Dividing by $b^2$ gives the theorem.
\end{proof}

At $r=1$ the identity can be verified directly. A centered linear statistic $g=\sum_jc_jp_j$ with $\sum_ja_jc_j=0$ lies in $\mathcal R^a_1$, its projection onto the constants under the posterior $\mathrm{Dir}(a+e_i)$ is its posterior mean $c_i/(|a|+1)$, and \cref{lem:dirichlet}(i) gives
\[
 \sum_i\frac{a_i}{|a|}\Bigl|\frac{c_i}{|a|+1}\Bigr|^2=\frac1{|a|+1}\|g\|_a^2 .
\]

\begin{lemma}[Sums over outcomes]\label{lem:haar-coherent}
Let $h_1,\dots,h_N$ be polynomials in $p$ of degree at most $r$, and let $Y=\sum_{i=1}^Np_ih_i$. Then
\[
 \bigl\|\Pi^a_{r+1}Y\bigr\|_a^2\le\frac{r+1}{|a|+2r+1}\sum_{i=1}^N\E_a\bigl[p_i|h_i|^2\bigr].
\]
\end{lemma}

\begin{proof}
\emph{Step 1: the pairing in the posterior spaces.} Since
\[
 \bigl\|\Pi^a_{r+1}Y\bigr\|_a=\sup_{g\in \mathcal R^a_{r+1},\ \|g\|_a=1}|\langle g,Y\rangle_a| ,
\]
we fix such a $g$. Converting each term to the posterior by \cref{lem:dirichlet}(ii), and inserting the projection $\Pi^{a+e_i}_{\le r}$, which is allowed because $h_i$ has degree at most $r$, we obtain
\begin{align*}
 \langle g,Y\rangle_a&=\sum_i\E_a[p_i\bar gh_i]\\
 &=\sum_i\frac{a_i}{|a|}\langle g,h_i\rangle_{a+e_i}\\
 &=\sum_i\frac{a_i}{|a|}\bigl\langle\Pi^{a+e_i}_{\le r}g,h_i\bigr\rangle_{a+e_i}.
\end{align*}

\emph{Step 2: the two factors.} By the Cauchy--Schwarz inequality in the weighted direct sum of the posterior spaces,
\[
 |\langle g,Y\rangle_a|^2\le\Bigl(\sum_i\frac{a_i}{|a|}\bigl\|\Pi^{a+e_i}_{\le r}g\bigr\|_{a+e_i}^2\Bigr)\Bigl(\sum_i\frac{a_i}{|a|}\|h_i\|_{a+e_i}^2\Bigr).
\]

The first factor equals $(r+1)/(|a|+2r+1)$ by \cref{thm:lowering} with $r+1$ in place of $r$. The second factor equals $\sum_i\E_a[p_i|h_i|^2]$ by \cref{lem:dirichlet}(ii).
\end{proof}

\subsection{The effect of one query on the level}\label{sec:haar-growth}

We now use \cref{lem:haar-subset,lem:haar-coherent} to bound the effects of appending $\ket\phi$ and applying $\bra\phi$. Each map changes the level by at most one. Bounds on the components that move between levels then bound the commutator with $C^{1/2}$, and hence the change in progress.

\paragraph{The oracle in terms of two maps.}
Let $\mathcal W$ be the space of all registers other than the query register $\mathsf Q$, so that $\mathcal V=\mathcal W\otimes\mathsf Q$ after reordering. Define, pointwise in $\phi$, the maps $L_+:\mathcal W\to\mathcal W\otimes\mathsf Q$ and $L_-:\mathcal W\otimes\mathsf Q\to\mathcal W$ by
\begin{align*}
 L_+G&=G\otimes\ket\phi=\sum_i\alpha_i\,G\otimes\ket i,\\
 L_-F&=\sum_i\overline{\alpha_i}\,F_i\qquad\text{for }F=F_\bot\otimes\ket\bot+\sum_iF_i\otimes\ket i,
\end{align*}
together with $J G=G\otimes\ket\bot$ and its adjoint $J^\dagger F=F_\bot$. The maps $L_+$ and $J$ are isometries pointwise in $\phi$, and $L_-=L_+^\dagger$ both pointwise and for the inner product of vector-valued polynomials.

For $N=2$, consider $G_\pm=\alpha_1\ket1\pm\alpha_2\ket2$. Both are normalized and have level zero, but
\[
 L_-G_+=p_1+p_2=1,\qquad L_-G_-=p_1-p_2=2p_1-1.
\]
Since $p_1$ is uniform on $[0,1]$, the first output has level zero and the second has squared level-one norm $\int_0^1(2p-1)^2\,dp=1/3$, attaining \cref{lem:haar-coherent} at $r=0$. The input components are orthogonal in the query register, but their images under $L_-$ are added in the same space. We must therefore bound their sum before taking its norm. We use $L_-$ to analyze the oracle; the full query remains unitary.

We write
\[
 K_\phi=\frac{J-L_+}{\sqrt2},\qquad K_\phi^\dagger=\frac{J^\dagger-L_-}{\sqrt2} .
\]
The map $K_\phi$ appends $\ket{\eta_\phi}=(\ket\bot-\ket\phi)/\sqrt2$ from \eqref{eq:prelim-canonical-oracle}, and its adjoint applies $\bra{\eta_\phi}$. Thus $K_\phi K_\phi^\dagger$ is the projection onto $\ket{\eta_\phi}$ in the query register, and
\begin{equation}\label{eq:haar-oracle-ER}
 O_\phi=I-2K_\phi K_\phi^\dagger .
\end{equation}
The maps $J$ and $J^\dagger$ commute with every $P_k$, by the third fact after \cref{def:haar-level}, so the analysis of a query reduces to that of $L_+$ and $L_-$.

\paragraph{The effect of one factor on a block.}
Fix a block $u_\gamma h$ with $h\in\mathcal R^a_r$, where $a=a(\gamma)$ and $k=\operatorname{rep}(\gamma)+r$. The map $L_-$ multiplies by $\overline{\alpha_i}$ at outcome $i$. This changes $\gamma_i$ to $\gamma_i-1$ and the polynomial factor as in \cref{lem:haar-adjacent}. The table lists the resulting degrees under the Dirichlet measure and levels.
\begin{center}
\small\renewcommand{\arraystretch}{1.15}
\begin{tabular}{@{}cccccc@{}}
\toprule
input entry $\gamma_i$ & new entry & change of $\operatorname{rep}$ & measure & degrees & levels\\
\midrule
$\gamma_i\le-1$ & $\gamma_i-1$ & $+1$ & $a\to a+e_i$ & $r,\ r-1$ & $k+1,\ k$\\
$\gamma_i=0$ & $-1$ & $0$ & $a\to a+e_i$ & $r,\ r-1$ & $k,\ k-1$\\
$\gamma_i=1$ & $0$ & $0$ & $a\to a-e_i$ & $r,\ r+1$ & $k,\ k+1$\\
$\gamma_i\ge2$ & $\gamma_i-1$ & $-1$ & $a\to a-e_i$ & $r,\ r+1$ & $k-1,\ k$\\
\bottomrule
\end{tabular}
\end{center}
In the first two rows the magnitude $|\gamma_i|$ grows, the new Dirichlet measure is the posterior $\mathrm{Dir}(a+e_i)$, and part~(i) of \cref{lem:haar-adjacent} splits $h$ into degrees $r$ and $r-1$. In the last two rows the magnitude shrinks, the factor $\alpha_i\overline{\alpha_i}=p_i$ appears, and part~(ii) of the lemma splits $p_ih$ into degrees $r$ and $r+1$ for $\mathrm{Dir}(a-e_i)$. Multiplication by $\alpha_i$, which is what $L_+$ does, obeys the same table with the signs of the entries reversed.

Two consequences of the table are used repeatedly. First, a single factor changes the level by at most one. Second, the level increases only in the first and third rows, that is, only when the new factor falls on an outcome that already appears in the pattern. Combined with \cref{lem:haar-subset}, which bounds the weight that a polynomial of degree $d$ places on the occupied outcomes, this gives the rate $(d+1)/N$ in \cref{lem:haar-up}, and hence the rate
$\sqrt{(d+1)/N}$ in \cref{thm:haar-one-query}.

For $L_+$, we use \cref{lem:haar-subset}. For $L_-$, repeated unpaired factors use the same estimate, while paired conjugate factors require \cref{lem:haar-coherent}. As the $G_\pm$ example shows, terms contributing to the same output block must be added before taking norms.

\begin{lemma}[Raising the level by one]\label{lem:haar-up}
For $L\in\{L_+,L_-\}$ and every polynomial $F\in\mathcal P_{\le d}$ in the domain of $L$,
\[
 \sum_{k\ge0}\bigl\|P_{k+1}LP_kF\bigr\|^2\le36\,\frac{d+1}N\,\|F\|^2 .
\]
\end{lemma}

The proof is in Appendix~\ref{app:haar-up}.

For $L_+$, only occupied outcomes can raise the level, so \cref{lem:haar-subset} bounds their total squared norm by $16d/N$ times the input norm squared. For $L_-$, repeated unpaired factors use the same estimate in the measure $\mathrm{Dir}(a-e_i)$, while paired conjugate factors use \cref{lem:haar-coherent}. Their bounds are $16(d+1)/N$ and $(d+1)/(2N)$, respectively. Within an output block we add these two contributions before taking the norm; distinct output blocks are orthogonal. The appendix verifies that every input component is counted at most once, giving the stated constant $36$.

\begin{lemma}[Lowering the level by one]\label{lem:haar-down}
For $L\in\{L_+,L_-\}$ and every polynomial $F\in\mathcal P_{\le d}$ in the domain of $L$,
\[
 \sum_{k\ge1}\bigl\|P_{k-1}LP_kF\bigr\|^2\le36\,\frac{d+2}N\,\|F\|^2 .
\]
\end{lemma}

The proof is in Appendix~\ref{app:haar-down}.

This follows by duality from \cref{lem:haar-up}, since $L_+^\dagger=L_-$ and the output has degree at most $d+1$. Applying the raising estimate at that degree gives $d+2$ in the numerator.

\begin{lemma}[One map and the level]\label{lem:haar-commutator}
For $L\in\{L_+,L_-,J,J^\dagger\}$ and every polynomial $F\in\mathcal P_{\le d}$ in the domain of $L$,
\[
 \bigl\|C^{1/2}LF-LC^{1/2}F\bigr\|\le12\sqrt{\frac{d+2}N}\,\|F\| .
\]
\end{lemma}

\begin{proof}
For $J$ and $J^\dagger$ the left side is zero. Let $L\in\{L_+,L_-\}$.

Since $L$ maps each block to adjacent levels only, inserting the level decomposition on both sides of $L$ gives
\begin{align*}
 C^{1/2}L-LC^{1/2}&=\sum_{k\ge0}\bigl(\sqrt{k+1}-\sqrt k\bigr)P_{k+1}LP_k\\
 &\quad+\sum_{k\ge0}\bigl(\sqrt k-\sqrt{k+1}\bigr)P_kLP_{k+1},
\end{align*}
and every scalar factor has absolute value at most one.

The terms of the first sum lie at distinct levels, so its norm on $F$ is at most
\[
 \Bigl(\sum_k\|P_{k+1}LP_kF\|^2\Bigr)^{1/2}\le6\sqrt{\frac{d+1}N}\,\|F\|
\]
by \cref{lem:haar-up}. Likewise, the second sum is at most $6\sqrt{(d+2)/N}\,\|F\|$ by \cref{lem:haar-down}. The triangle inequality gives the claim.
\end{proof}

\begin{theorem}[One-query growth]\label{thm:haar-one-query}
There is an absolute constant $c_0$ such that for all \mbox{$F\in\mathcal P_{\le d}(\mathcal V)$},
\[
 \lambda(O_\phi F)\le\lambda(F)+c_0\sqrt{\frac{d+1}N}\,\|F\| ,
\]
and the same holds for the controlled query $O^{\mathrm c}_\phi$.
\end{theorem}

\begin{proof}
By \eqref{eq:haar-oracle-ER},
\begin{align*}
 C^{1/2}O_\phi-O_\phi C^{1/2}&=-2\Bigl((C^{1/2}K_\phi-K_\phi C^{1/2})K_\phi^\dagger\\
 &\qquad+K_\phi(C^{1/2}K_\phi^\dagger-K_\phi^\dagger C^{1/2})\Bigr),
\end{align*}
and we bound the two terms on $F\in\mathcal P_{\le d}(\mathcal V)$ separately.

\emph{Step 1: the two terms.} For the first term, we have $K_\phi^\dagger F\in\mathcal P_{\le d+1}$ and $\|K_\phi^\dagger F\|\le\|F\|$, because $K_\phi^\dagger$ is a contraction pointwise in $\phi$. Moreover, $K_\phi=(J-L_+)/\sqrt2$ with $J$ commuting with $C^{1/2}$, so \cref{lem:haar-commutator} for $L_+$ at degree $d+1$ gives
\[
 \bigl\|(C^{1/2}K_\phi-K_\phi C^{1/2})K_\phi^\dagger F\bigr\|\le\frac{12}{\sqrt2}\sqrt{\frac{d+3}N}\,\|F\| .
\]

For the second term, the map $K_\phi$ is an isometry pointwise and $K_\phi^\dagger=(J^\dagger-L_-)/\sqrt2$, so the lemma for $L_-$ at degree $d$ gives
\[
 \bigl\|K_\phi(C^{1/2}K_\phi^\dagger-K_\phi^\dagger C^{1/2})F\bigr\|\le\frac{12}{\sqrt2}\sqrt{\frac{d+2}N}\,\|F\| .
\]

\emph{Step 2: the level after the query.} Using $d+3\le3(d+1)$ and $d+2\le2(d+1)$, the two terms together give
\begin{align*}
 \bigl\|C^{1/2}O_\phi F-O_\phi C^{1/2}F\bigr\|&\le12\sqrt2\,(\sqrt3+\sqrt2)\sqrt{\frac{d+1}N}\,\|F\|\\
 &\le54\sqrt{\frac{d+1}N}\,\|F\| .
\end{align*}

Since $O_\phi$ preserves norms for every $\phi$, it also preserves the norm of a polynomial. Hence
\begin{align*}
 \lambda(O_\phi F)&=\|C^{1/2}O_\phi F\|\\
 &\le\|O_\phi C^{1/2}F\|+\|C^{1/2}O_\phi F-O_\phi C^{1/2}F\|\\
 &=\lambda(F)+\|C^{1/2}O_\phi F-O_\phi C^{1/2}F\| .
\end{align*}

\emph{Step 3: the controlled query.} The control projections are $\phi$-independent and commute with $C^{1/2}$, so
\[
 C^{1/2}O^{\mathrm c}_\phi-O^{\mathrm c}_\phi C^{1/2}=\ket1\bra1\otimes(C^{1/2}O_\phi-O_\phi C^{1/2}) .
\]
Applying the bound just proved to $(\ket1\bra1\otimes I)F\in\mathcal P_{\le d}(\mathcal V)$, whose norm is at most $\|F\|$, gives the same estimate, and $O^{\mathrm c}_\phi$ is again unitary pointwise.
\end{proof}

The proof gives $c_0=54$ and completes \cref{thm:haar-lower}. Only the rank-one part enters the commutator, since the identity commutes with $C^{1/2}$. Inverse queries are also covered because $O_\phi^{-1}=O_\phi$.

\begin{remark}[State designs]\label{rem:haar-designs}
For a fixed algorithm given $s$ copies and making $q$ queries to the canonical oracle, $\Psi(\phi)$ has degree at most $s+2q$ (\cref{lem:haar-degree}). The score integrand $\sum_zNp_\phi(z)\|\Psi_z(\phi)\|^2$ therefore has degree at most $2s+4q+2$, since $Np_\phi(z)=N\alpha_z\overline{\alpha_z}$ has degree two. Averaging over a uniform global phase eliminates monomials with unequal numbers of amplitudes and conjugates. Every remaining monomial has degree at most $s+2q+1$ in each, so an exact $(s+2q+1)$-design preserves its Haar expectation. The score, and hence the bound of \cref{thm:haar-lower} for the supplied $s$ copies, is unchanged when the state is drawn from such a design and multiplied by an independent uniform global phase.
\end{remark}

\subsection{Haar-random unitary access}\label{sec:haar-unitary}

With Haar-random unitary access the algorithm applies a Haar-random $U$, its inverse, and controlled versions of both, and the state scored is $\ket\phi=U\ket{0^n}$. The following simulation result reduces this model to the canonical oracle.

The conditional distribution below is the random state-preparation oracle of \cite[Definition~8]{Kretschmer}. Fix any unitary $V$ with $V\ket{0^n}=\ket\phi$, and let $W$ fix $\ket{0^n}$ and be Haar-random on its orthogonal complement. Then $U_\phi=VW$, and Haar invariance makes its distribution independent of the choice of $V$.

\begin{theorem}[{\cite[Definition~8 and Theorem~19]{Kretschmer}}]\label{thm:kretschmer-simulation}
For a unit vector $\ket\phi\in\C^N$, let $U_\phi$ be a random unitary with $U_\phi\ket{0^n}=\ket\phi$, distributed according to the Haar measure on the set of such unitaries.
\begin{enumerate}
\item For every algorithm $\mathcal A$ that makes $q$ queries to $U_\phi$, including inverse and controlled queries, there is an algorithm $\mathcal B$ that makes $2q$ queries to $O_\phi$ such that, for every $\ket\phi$, the channel from the input to the premeasurement state of $\mathcal A$, averaged over $U_\phi$, and the corresponding channel of $\mathcal B$ are within diamond distance $(10q+4)/\sqrt N$. Here the diamond distance is the completely bounded trace norm of the difference, without a factor $1/2$.
\item If $\ket\phi$ is Haar-random, then $U_\phi$ is a Haar-random unitary.
\end{enumerate}
\end{theorem}

Converting a distance between channels into a difference of scores requires a bound on the largest possible reward, which is $N$ times the probability of the heaviest outcome.

\begin{lemma}[The heaviest outcome]\label{lem:heaviest}
Let $H_N=\sum_{j=1}^N1/j$. For a Haar-random state,
\[
 \E_\phi\bigl[N\max_zp_z\bigr]=H_N\le\ln N+1 .
\]
\end{lemma}

The proof is in Appendix~\ref{app:haar-heaviest}.

\begin{theorem}[Haar lower bound for unitary access]\label{thm:haar-unitary}
Every algorithm that makes $q$ queries to a Haar-random unitary, with inverse and controlled access, satisfies
\[
 B(\mathcal A)\le\beta_N+2\sqrt2\,\Lambda_{0,2q}+72\,\Lambda_{0,2q}^2+\frac{(5q+2)H_N}{\sqrt N} .
\]
Consequently, there is an absolute constant $c_3>0$ such that, for every $0<\varepsilon\le1$ and every $N$ with $(5N^{1/3}+2)(\ln N+1)\le\varepsilon\sqrt N/2$, a score of $2+\varepsilon$ requires at least $c_3\,\varepsilon^{2/3}N^{1/3}$ queries.
\end{theorem}

\begin{proof}
\emph{Step 1: a fixed state.} Fix $\phi$ and let $Z$ and $Z'$ be the outputs of $\mathcal A$, averaged over $U_\phi$, and of the simulator $\mathcal B$ of \cref{thm:kretschmer-simulation}. The diamond distance bounds the trace norm of the difference of the premeasurement states, and the total variation distance of the output distributions is at most half of that trace norm. Since the reward $z\mapsto Np_z$ is nonnegative and at most $N\max_zp_z$, the two scores differ by at most the total variation distance times $N\max_zp_z$, that is,
\[
 N\,\E[p_Z\mid\phi]\le N\,\E[p_{Z'}\mid\phi]+\frac{5q+2}{\sqrt N}N\max_zp_z .
\]

\emph{Step 2: averaging over the state.} We now average over the Haar-random $\phi$, under which $U_\phi$ is Haar-random. The simulator $\mathcal B$ is a mixture of algorithms of the form of \cref{def:experiment}, each making $2q$ queries to $O_\phi$, and the expected score is linear in the mixture. Hence \cref{thm:haar-lower} with $s=0$ and $2q$ queries bounds the first term on the right; \cref{lem:heaviest} bounds the second, which gives the display in the statement.

\emph{Step 3: the consequence.} We may assume $q\le N^{1/3}$, because for $q>N^{1/3}$ the claim holds for every $c_3\le1$. Then, by the hypothesis on $N$, the last term of the display in the statement is at most $(5N^{1/3}+2)(\ln N+1)/\sqrt N\le\varepsilon/2$, and $\beta_N\le2$ gives $\varepsilon/2\le2\sqrt2\,\Lambda_{0,2q}+72\,\Lambda_{0,2q}^2$.
The proof of \cref{cor:haar-birthday}, with $\varepsilon/2$ in place of $\varepsilon$ and $2q$ in place of $q$, gives $2q\ge c_2(\varepsilon/2)^{2/3}N^{1/3}$, so the claim holds with $c_3=\min\{1,c_2/2^{5/3}\}$.
\end{proof}

Together with \cref{thm:squaring-scores} in \cref{sec:squaring}, which applies to unitary access without change because the squaring algorithm uses $U$ only to prepare $\ket\phi$ and to reflect about it, this proves part~2 of \cref{thm:intro-xeb} and the Haar case of the final statement of that theorem.

The quantity $H_N=\Theta(\log N)$ is also the supremum of the expected scores of all algorithms: by \eqref{eq:prelim-score} no algorithm scores more than $\E_\phi\,N\max_zp_z=H_N$, and an algorithm with enough queries can identify the heaviest outcome with probability close to one.

\section{The squaring algorithm}\label{sec:squaring}

In this section we prove part~3 of \cref{thm:intro-xeb} and the formal version of \cref{cor:collision}. The algorithm and its analysis are the same for the Fourier ensemble, for the canonical oracle, and for Haar-random unitary access, because the algorithm uses the oracle only to prepare a state whose measurement gives the ideal distribution, and to reflect about it. We treat the three forms of access together and write $\omega$ for the oracle, $p=p_\omega$ for the ideal distribution, and $p_z=p(z)$. The normalized moments of $p$ and the \emph{collision distribution} are
\begin{equation}\label{eq:squaring-moments-def}
 M_j=N^{j-1}\sum_zp_z^j\qquad(j\ge1),\qquad r_p(z)=\frac{p_z^2}{\sum_yp_y^2}=\frac{Np_z^2}{M_2},
\end{equation}
where $r_p$ is the distribution of the common value of two independent samples from $p$ conditioned on their being equal. The ideal score \eqref{eq:prelim-ideal-score} is $\E_\omega M_2$, and $M_2\ge1$ for every distribution by the Cauchy--Schwarz inequality
\[
 1=\Bigl(\sum_zp_z\Bigr)^2\le N\sum_zp_z^2 .
\]

\subsection{The three forms of access}\label{sec:squaring-access}

In each case the oracle provides access for state preparation and reflection: there is a unit vector $\ket{\psi_\omega}\in\C^N$ whose computational-basis measurement has distribution $p_\omega$, one query prepares $\ket{\psi_\omega}$, and two queries implement the reflection $2\ket{\psi_\omega}\bra{\psi_\omega}-I$.

\emph{Fourier ensemble.} Let $\Delta_f\ket x=f(x)\ket x$ be the phase query, which is one query to the value oracle with the target qubit in $\ket-$ by \eqref{eq:prelim-phase-kickback}, and let $U_f=H^{\otimes n}\Delta_fH^{\otimes n}$, a unitary with $U_f^\dagger=U_f$. Then
\[
 \ket{\psi_f}=U_f\ket{0^n}=\sum_z\widehat f(z)\ket z,
\]
because $H^{\otimes n}\ket{0^n}=N^{-1/2}\sum_x\ket x$ and $H^{\otimes n}\ket x=N^{-1/2}\sum_z\chi_z(x)\ket z$.

\emph{Canonical oracle.} Here $\ket{\psi_\phi}=\ket\phi=O_\phi\ket\bot$. The reflection about $\ket\phi$ acts on the query register $\mathsf Q=\C\ket\bot\oplus\C^N$, but the algorithm below never leaves $\C^N$ once $\ket\phi$ has been prepared, because $\ket\phi\in\C^N$ and both reflections it uses map $\C^N$ to itself.

\emph{Haar-random unitary access.} Here $\ket{\psi_U}=U\ket{0^n}$.

The reflections in the three cases are
\begin{align*}
 2\ket{\psi_f}\bra{\psi_f}-I&=U_f(2\ket{0^n}\bra{0^n}-I)U_f,&&\text{(Fourier ensemble)}\\
 2\ket\phi\bra\phi-I_{\mathsf Q}&=O_\phi(2\ket\bot\bra\bot-I_{\mathsf Q})O_\phi,&&\text{(canonical oracle)}\\
 2\ket{\psi_U}\bra{\psi_U}-I&=U(2\ket{0^n}\bra{0^n}-I)U^\dagger,&&\text{(unitary access)}
\end{align*}
where the second identity holds by \eqref{eq:prelim-oracle-expansion}.

For $j\ge1$ let $\kappa_j=j!$ in the two Haar cases, where $p$ is uniform on the simplex, and, in the Fourier ensemble, $\kappa_j=(2j-1)!!=1\cdot3\cdots(2j-1)$. These are the limits of $\E M_j$ as $N\to\infty$, and \cref{lem:moments} gives $\E M_j\le\kappa_j$ for every $j$ together with the exact values of $\E M_2$ and $\E M_3$. In particular, $\kappa_2$, which is $2$ for Haar access and $3$ for the Fourier ensemble, is the limit of the ideal score.

\subsection{The algorithm}\label{sec:squaring-algorithm}

The algorithm stores the outcomes of $m$ measured samples and amplifies the component of one more sample that lies in the stored set. The probability of the stored set is a random variable whose mean is about $\kappa_2m/N$ (\cref{lem:squaring-table}), and the algorithm chooses its number of amplification steps as if this probability were exactly $P_\star=\kappa_2m/N$; it never estimates the actual probability.

\begin{definition}[The squaring algorithm $\mathsf{Sq}_m$]\label{def:squaring}
Let $m\ge1$ be such that $P_\star=\kappa_2m/N\le1/4$, let $\vartheta_0=\arcsin\sqrt{P_\star}$, let $\ell$ be a nonnegative integer nearest to $\pi/(4\vartheta_0)-1/2$, and put $L=2\ell+1$, so that $|L\vartheta_0-\pi/2|\le\vartheta_0$. The algorithm proceeds as follows.
\begin{enumerate}
\item Prepare and measure $m$ copies of $\ket{\psi_\omega}$, obtaining outcomes $Y_1,\dots,Y_m$, and store the set $S=\{Y_1,\dots,Y_m\}$.
\item Prepare one more copy of $\ket{\psi_\omega}$.
\item Apply $\ell$ times the operator
\[
 (2\ket{\psi_\omega}\bra{\psi_\omega}-I)(I-2\Pi_S),\qquad\text{where }\Pi_S=\sum_{z\in S}\ket z\bra z .
\]
\item Measure in the computational basis and output the outcome $Z$.
\end{enumerate}
\end{definition}

The reflection $I-2\Pi_S$ uses only the stored set, so the number of queries is $m+1+2\ell$. Since $|L\vartheta_0-\pi/2|\le\vartheta_0$ gives $2\ell\le\pi/(2\vartheta_0)$, and $\arcsin u\ge u$ gives $\vartheta_0\ge\sqrt{P_\star}$, the number of queries is at most
\begin{equation}\label{eq:squaring-queries}
 m+1+\frac\pi{2\vartheta_0}\le m+1+\frac\pi2\sqrt{\frac N{\kappa_2m}} .
\end{equation}
The number of amplification steps depends only on $N$, $m$, and the ensemble, through $P_\star=\kappa_2m/N$; it does not depend on the oracle or on the stored outcomes. The rest of the section proves that the resulting output distribution is close to the collision distribution, which is proportional to $p_z^2$.

\subsection{Moments of the two ensembles}\label{sec:squaring-moments}

\begin{lemma}[Moments]\label{lem:moments}
The following hold.
\begin{enumerate}
\item[(i)] In the Haar ensemble, for all $j\ge1$,
\begin{align*}
 &\E M_j=\frac{N^{j-1}\,j!}{(N+1)(N+2)\cdots(N+j-1)}\le j!,\qquad
 \E M_j^2\le(j!)^2+\frac{(2j)!}N,\quad\text{together with}\\
 &\Var M_j\le\frac{(2j)!+j^2(j!)^2}N ,\quad\text{and in particular}\\
 &\E M_2=\frac{2N}{N+1},\qquad \E M_3=\frac{6N^2}{(N+1)(N+2)},\qquad \E M_4\le24,\quad\text{and}\\
 &\E(M_2-2)^2=\frac{4N^2-4N+24}{(N+1)(N+2)(N+3)}\le\frac4N .
\end{align*}
\item[(ii)] In the Fourier ensemble, $\E\bigl[(\sqrt N\widehat f(z))^{2j}\bigr]\le\kappa_j$ for every $z$ and $j\ge1$, hence $\E M_j\le\kappa_j$ and $\E M_j^2\le\kappa_{2j}$, and
\[
 \E M_2=3-\frac2N,\qquad \E M_3=15-\frac{30}N+\frac{16}{N^2},
\]
\[
 \Var M_2\le\frac{24}N,\qquad \E(M_2-3)^2\le\frac{28}N,\qquad \Var M_3\le\frac{10^4}N .
\]
\end{enumerate}
\end{lemma}

The proof is in Appendix~\ref{app:squaring-moments}.

The two limiting distributions are visible in these values: $Np_z$ is nearly an exponential random variable in the Haar ensemble, with moments $j!$, and $\sqrt N\widehat f(z)$ is nearly a standard Gaussian in the Fourier ensemble, with even moments $(2j-1)!!$.

\subsection{The output distribution for a fixed stored set}\label{sec:squaring-fixed-table}

Fix the oracle and the outcomes $Y_1,\dots,Y_m$, and let $P_S=p(S)=\sum_{z\in S}p_z$ be the mass of the stored set, in the notation of \cref{lem:dirichlet}(iii); it is positive because every stored outcome has positive probability. For a set $T$ with $p(T)>0$ let
\[
 \mu_T(z)=\frac{p_z\mathbf 1[z\in T]}{p(T)}
\]
be $p$ conditioned on $T$. We use $\mu_S$ and, when $P_S<1$, $\mu_{S^{\mathrm c}}$, where $S^{\mathrm c}$ is the complement of $S$. Put $g(P)=\sin^2\bigl(L\arcsin\sqrt P\bigr)$.

\begin{lemma}[Output distribution for a fixed stored set]\label{lem:squaring-fixed}
Conditioned on the oracle and on $S$, the output of $\mathsf{Sq}_m$ has the distribution
\[
 g(P_S)\,\mu_S+(1-g(P_S))\,\mu_{S^{\mathrm c}},
\]
where the second term is absent if $P_S=1$.
\end{lemma}

\begin{proof}
Apply \cref{lem:amplification} with $\psi=\ket{\psi_\omega}$ and $\Pi=\Pi_S$, so that $P=\|\Pi_S\psi_\omega\|^2=P_S$. If $P_S=1$, the state after step~3 of \cref{def:squaring} is $\pm\ket{\psi_\omega}$ and the output has distribution $p=\mu_S$. Otherwise the state is $\sin(L\vartheta)u+\cos(L\vartheta)u^\perp$, where
\[
 \vartheta=\arcsin\sqrt{P_S},\qquad u=\frac{\Pi_S\psi_\omega}{\sqrt{P_S}},\qquad u^\perp=\frac{(I-\Pi_S)\psi_\omega}{\sqrt{1-P_S}} .
\]

When $P_S<1$, measuring this state gives $z\in S$ with probability $\sin^2(L\vartheta)\,p_z/P_S$ and $z\notin S$ with probability $\cos^2(L\vartheta)\,p_z/(1-P_S)$, because $u$ and $u^\perp$ are supported on $S$ and on its complement.
\end{proof}

\begin{lemma}[Amplification error]\label{lem:squaring-amplification-error}
Let $0<P_\star\le1/4$ and let $L$ be as in \cref{def:squaring}. Then, for every $P\in(0,1]$,
\[
 1-g(P)\le18\,(P/P_\star-1)^2+3P_\star .
\]
\end{lemma}

The two terms have different origins: the first is the change of the amplification angle when the actual probability $P$ differs from the design value $P_\star$, and the second is the rounding of the number of steps to an integer.

\begin{proof}
\emph{Step 1: reduction to the error in the angle.} Let $\vartheta=\arcsin\sqrt P\in(0,\pi/2]$. Then
\begin{align*}
 1-g(P)&=\cos^2(L\vartheta)=\sin^2(L\vartheta-\pi/2)\\
 &\le(L\vartheta-\pi/2)^2 ,
\end{align*}
and, by the triangle inequality and the property $|L\vartheta_0-\pi/2|\le\vartheta_0$ of \cref{def:squaring},
\begin{align*}
 |L\vartheta-\pi/2|&\le L|\vartheta-\vartheta_0|+|L\vartheta_0-\pi/2|\\
 &\le L|\vartheta-\vartheta_0|+\vartheta_0 ,\qquad\text{and therefore}\\
 1-g(P)&\le2L^2(\vartheta-\vartheta_0)^2+2\vartheta_0^2 .
\end{align*}

\emph{Step 2: bounds on $\vartheta_0$ and on $L$.} Since $\arcsin$ is convex on $[0,1]$ with $\arcsin(1/2)=\pi/6$, we have $\arcsin u\le(\pi/3)u$ for $0\le u\le1/2$. Hence
\[
 \vartheta_0\le(\pi/3)\sqrt{P_\star},\qquad 2\vartheta_0^2\le(2\pi^2/9)P_\star\le3P_\star .
\]
Moreover, $|L\vartheta_0-\pi/2|\le\vartheta_0$ gives $L\le\pi/(2\vartheta_0)+1$, and with $\sqrt{P_\star}\le\vartheta_0$ and $\sqrt{P_\star}\le1/2$ this gives $L\sqrt{P_\star}\le\pi/2+1/2$.

\emph{Step 3: the cases $P\le1/2$ and $P>1/2$.} If $P\le1/2$, then $\sqrt P$ and $\sqrt{P_\star}$ lie in $[0,1/\sqrt2]$, where the derivative of $\arcsin$ is at most $\sqrt2$, so
\begin{align*}
 |\vartheta-\vartheta_0|&\le\sqrt2\,\bigl|\sqrt P-\sqrt{P_\star}\bigr|=\sqrt{2P_\star}\,\bigl|\sqrt{P/P_\star}-1\bigr|\\
 &\le\sqrt{2P_\star}\,|P/P_\star-1| ,
\end{align*}
using $|\sqrt y-1|\le|y-1|$ for $y\ge0$. Therefore, by the bound on $L\sqrt{P_\star}$ in Step 2,
\begin{align*}
 2L^2(\vartheta-\vartheta_0)^2&\le4(L\sqrt{P_\star})^2(P/P_\star-1)^2\\
 &\le4(\pi/2+1/2)^2(P/P_\star-1)^2\\
 &\le18\,(P/P_\star-1)^2 .
\end{align*}

If instead $P>1/2$, then $P/P_\star>2$, so $1-g(P)\le1\le(P/P_\star-1)^2$, and the claim follows in both cases.
\end{proof}

\subsection{The output distribution of the algorithm}\label{sec:squaring-output-law}

We now average over the stored outcomes. Fix the oracle, let $Y_1,\dots,Y_m$ be independent with distribution $p$, and put
\[
 W_i=Np_{Y_i},\qquad \overline W=\frac1m\sum_{i=1}^mW_i .
\]
Let $\mathcal D$ be the event that $Y_1,\dots,Y_m$ are distinct, and let
\[
 \widetilde\mu_Y(z)=\frac{\sum_iW_i\mathbf 1[Y_i=z]}{m\overline W}
\]
be the probability distribution that gives each stored outcome a weight proportional to its reward; it is well defined because every $W_i$ is positive. On the event $\mathcal D$, $\widetilde\mu_Y$ equals $\mu_S$, the output distribution of the algorithm conditioned on the amplified sample lying in $S$ (\cref{lem:squaring-fixed}).

For a single stored set the normalization $\overline W$ has fluctuations of order $m^{-1/2}$, but the average of $\widetilde\mu_Y$ over the stored outcomes is within $O(1/m)$ of the collision distribution. The reason, made precise in the proof of the next lemma, is that conditioned on one stored outcome $z$, the normalization $m\overline W$ is $Np_z$ plus the sum of the rewards of the other $m-1$ outcomes, which is independent of $z$, so Jensen's inequality bounds the average from below in terms of the mean $(m-1)M_2$ of that sum.

\begin{lemma}[Identities for the stored set]\label{lem:squaring-table}
For every distribution $p$:
\begin{enumerate}
\item[(i)] $\E_Y\overline W=M_2$ and $\Var_Y(\overline W)\le M_3/m$;
\item[(ii)] on $\mathcal D$, $\widetilde\mu_Y=\mu_S$ and $P_S=m\overline W/N$;
\item[(iii)] $\Pr_Y[\mathcal D^{\mathrm c}]\le m(m-1)M_2/(2N)$;
\item[(iv)] the average $\E_Y\widetilde\mu_Y$ satisfies
\[
 \TV\bigl(\E_Y\widetilde\mu_Y,\,r_p\bigr)\le\frac{M_3}{mM_2^2},\qquad\text{where }\TV(a,b)=\frac12\sum_z|a(z)-b(z)| .
\]
\end{enumerate}
\end{lemma}

\begin{proof}
Part~(i) holds because $\E W_1=\sum_zp_z\cdot Np_z=M_2$, $\Var W_1\le\E W_1^2=\sum_zp_zN^2p_z^2=M_3$, and the $W_i$ are independent. On $\mathcal D$,
\[
 \sum_iW_i\mathbf 1[Y_i=z]=Np_z\mathbf 1[z\in S],\qquad m\overline W=N\sum_{z\in S}p_z=NP_S ,
\]
so $\widetilde\mu_Y(z)=p_z\mathbf 1[z\in S]/P_S=\mu_S(z)$, which proves (ii). Part~(iii) follows from the union bound over pairs, since $\Pr[Y_i=Y_j]=\sum_zp_z^2=M_2/N$.

For part~(iv), write $w_z=Np_z$ and $\bar\mu=\E_Y\widetilde\mu_Y$, which is a probability distribution. Since the $Y_i$ are exchangeable,
\begin{align*}
 \bar\mu(z)&=m\,\E_Y\Bigl[\frac{W_1\mathbf 1[Y_1=z]}{W_1+W'}\Bigr]\\
 &=m\,p_z\,w_z\,\E\Bigl[\frac1{w_z+W'}\Bigr],\qquad W'=\sum_{i=2}^mW_i ,
\end{align*}
where $W'$ is independent of $Y_1$ and has mean $(m-1)M_2$ by (i). This holds for every $z$ with $p_z>0$; if $p_z=0$, then $\bar\mu(z)=r_p(z)=0$.

Since $t\mapsto1/t$ is convex on $t>0$, Jensen's inequality gives
\[
 \E\Bigl[\frac1{w_z+W'}\Bigr]\ge\frac1{w_z+(m-1)M_2} .
\]
Substituting this into the display above for $\bar\mu(z)$ and using $r_p(z)=p_zw_z/M_2$, we obtain
\begin{align*}
 \bar\mu(z)&\ge\frac{m\,p_zw_z}{w_z+(m-1)M_2}\\
 &=r_p(z)\,\frac{mM_2}{w_z+(m-1)M_2} .
\end{align*}

We distinguish two cases.
\begin{enumerate}
\item If $w_z\le M_2$, the fraction is at least $1$ and $\bar\mu(z)\ge r_p(z)$.
\item If $w_z>M_2$, then
\begin{align*}
 r_p(z)-\bar\mu(z)&\le r_p(z)\,\frac{w_z-M_2}{w_z+(m-1)M_2}\\
 &\le r_p(z)\,\frac{w_z-M_2}{mM_2} ,
\end{align*}
where the second inequality holds because $w_z+(m-1)M_2>mM_2$.
\end{enumerate}

Since $\bar\mu$ and $r_p$ are probability distributions, $\TV(\bar\mu,r_p)$ is the sum over $z$ of the positive parts of $r_p(z)-\bar\mu(z)$. Moreover,
\begin{equation}\label{eq:squaring-collision-score}
 \sum_zr_p(z)w_z=\frac{N^2\sum_zp_z^3}{N\sum_zp_z^2}=\frac{M_3}{M_2},
\end{equation}
so
\[
 \TV(\bar\mu,r_p)\le\frac1{mM_2}\sum_zr_p(z)\,w_z=\frac{M_3}{mM_2^2} . \qedhere
\]
\end{proof}

\begin{theorem}[Output distribution of the squaring algorithm]\label{thm:output-law}
For each oracle $\omega$ let $a_\omega(z)=\Pr[Z=z\mid\omega]$ be the output distribution of $\mathsf{Sq}_m$, where $m\ge1$ satisfies $\kappa_2m/N\le1/4$. Then
\begin{equation}\label{eq:output-law}
 \E_\omega\TV(a_\omega,r_{p_\omega})\le\frac1m\,\E_\omega\Bigl[\frac{M_3}{M_2^2}\Bigr]+\binom m2\frac{\E M_2}N+\frac{18}{\kappa_2^2}\Bigl(\frac{\E M_3}m+\E(M_2-\kappa_2)^2\Bigr)+\frac{3\kappa_2m}N .
\end{equation}
In particular, if $N\ge64$ and $m=\lceil N^{1/3}\rceil$, then $\mathsf{Sq}_m$ makes at most $3N^{1/3}$ queries and
\begin{equation}\label{eq:output-law-rate}
 \E_\omega\TV(a_\omega,r_{p_\omega})\le
 \begin{cases}
 40\,N^{-1/3}&\text{in the Haar ensemble,}\\
 60\,N^{-1/3}&\text{in the Fourier ensemble.}
 \end{cases}
\end{equation}
\end{theorem}

\begin{proof}
\emph{Step 1: the triangle inequality.} Fix $\omega$, write $p=p_\omega$, let $a_{\omega,S}$ be the output distribution conditioned on $S$, which \cref{lem:squaring-fixed} describes, and let $\bar\mu=\E_Y\widetilde\mu_Y$. Since $a_\omega=\E_Ya_{\omega,S}$ and the total variation distance is convex, the triangle inequality and \cref{lem:squaring-table}(iv) give
\begin{align*}
 \TV(a_\omega,r_p)&\le\TV(a_\omega,\bar\mu)+\TV(\bar\mu,r_p)\\
 &\le\E_Y\TV(a_{\omega,S},\widetilde\mu_Y)+\frac{M_3}{mM_2^2} .
\end{align*}

\emph{Step 2: the two events.} We bound $\TV(a_{\omega,S},\widetilde\mu_Y)$ separately on $\mathcal D$ and on $\mathcal D^{\mathrm c}$.
\begin{enumerate}
\item On $\mathcal D$ we have $\widetilde\mu_Y=\mu_S$ by \cref{lem:squaring-table}(ii), and
\[
 a_{\omega,S}-\mu_S=(1-g(P_S))(\mu_{S^{\mathrm c}}-\mu_S),
\]
where $\mu_S$ and $\mu_{S^{\mathrm c}}$ are distributions with disjoint supports when $P_S<1$, and therefore $\TV(a_{\omega,S},\widetilde\mu_Y)=1-g(P_S)$. When $P_S=1$, both sides are $0$, because then $a_{\omega,S}=p=\mu_S$ by \cref{lem:squaring-fixed} and $g(1)=\sin^2(L\pi/2)=1$ since $L$ is odd.
\item On $\mathcal D^{\mathrm c}$, the total variation distance is at most $1$.
\end{enumerate}
Therefore
\[
 \TV(a_\omega,r_p)\le\E_Y\bigl[\mathbf 1_{\mathcal D}\,(1-g(P_S))\bigr]+\Pr\nolimits_Y[\mathcal D^{\mathrm c}]+\frac{M_3}{mM_2^2} .
\]

\emph{Step 3: the amplification bound and averaging over the oracle.} On $\mathcal D$ we have $P_S/P_\star=\overline W/\kappa_2$ by \cref{lem:squaring-table}(ii), so \cref{lem:squaring-amplification-error} gives
\[
 1-g(P_S)\le\frac{18}{\kappa_2^2}(\overline W-\kappa_2)^2+3P_\star ,
\]
and
\begin{align*}
 \E_Y(\overline W-\kappa_2)^2&=\Var_Y(\overline W)+(M_2-\kappa_2)^2\\
 &\le M_3/m+(M_2-\kappa_2)^2
\end{align*}
by \cref{lem:squaring-table}(i).

With part (iii) of the same lemma, the bounds of Step 2 and the two bounds above give
\[
 \TV(a_\omega,r_p)\le\frac{M_3}{mM_2^2}+\binom m2\frac{M_2}N+\frac{18}{\kappa_2^2}\Bigl(\frac{M_3}m+(M_2-\kappa_2)^2\Bigr)+\frac{3\kappa_2m}N ,
\]
and averaging over $\omega$ proves \eqref{eq:output-law}.

\emph{Step 4: the parameters of the second statement.} Let $N\ge64$ and $m=\lceil N^{1/3}\rceil$, so $m\ge N^{1/3}$ and
\begin{align*}
 m&\le N^{1/3}+1\le\frac54N^{1/3},\\
 \kappa_2m/N&\le\frac{15}4N^{-2/3}\le1/4 .
\end{align*}
By \eqref{eq:squaring-queries} the number of queries is at most
\begin{align*}
 m+1+\frac\pi2\sqrt{N/(\kappa_2m)}&\le\Bigl(\frac54+\frac\pi{2\sqrt2}\Bigr)N^{1/3}+1\\
 &\le3N^{1/3},
\end{align*}
since $N^{1/3}\ge4$. In \eqref{eq:output-law} we use $M_2\ge1$, the moments of \cref{lem:moments}, and the bounds
\begin{alignat*}{2}
 m&\ge N^{1/3}, &\qquad m(m-1)&\le\frac{25}{16}N^{2/3},\\
 m/N&\le\frac54N^{-2/3}\le\frac5{16}N^{-1/3}, &\qquad N^{-1}&\le\frac1{16}N^{-1/3} .
\end{alignat*}

\emph{Step 5: the Haar and Fourier ensembles.} The four terms of \eqref{eq:output-law} are at most the four summands on the left, in the Haar ensemble (first row) and in the Fourier ensemble (second row):
\begin{align*}
 \frac6m+\frac{m(m-1)}N+\Bigl(\frac{27}m+\frac{18}N\Bigr)+\frac{6m}N&\le\Bigl(6+\frac{25}{16}+27+\frac{18}{16}+\frac{30}{16}\Bigr)N^{-1/3}\le40N^{-1/3},\\
 \frac{15}m+\frac{3m(m-1)}{2N}+\Bigl(\frac{30}m+\frac{56}N\Bigr)+\frac{9m}N&\le\Bigl(15+\frac{75}{32}+30+\frac{56}{16}+\frac{45}{16}\Bigr)N^{-1/3}\le60N^{-1/3} . \qedhere
\end{align*}
\end{proof}

For $m$ of order $N^{1/3}$ the first three terms of \eqref{eq:output-law} are of order $N^{-1/3}$ and the fourth is of order $N^{-2/3}$: the first is the bias of the averaged stored-set distribution, the second the probability of a repeated stored outcome, the third the failure of the amplification caused by the deviation of $P_S$ from $P_\star$, and the fourth the rounding of the number of amplification steps.

\subsection{The score of the output distribution}\label{sec:squaring-scores}

A small error in total variation does not by itself control the score, since the reward $Np_z$ of a single outcome may be as large as $N$. The score that an output distribution $a$ loses relative to $r_p$ comes from the outcomes to which $a$ assigns less probability than $r_p$, and the next lemma makes use of this fact.

\begin{lemma}[Total variation and the score]\label{lem:squaring-tv-score}
For all distributions $p$ and $a$ on $\{0,1\}^n$,
\[
 \sum_za(z)\,Np_z\ge\frac{M_3}{M_2}-\sqrt{\frac{M_4}{M_2}\,\TV(a,r_p)} .
\]
Consequently, an algorithm whose output distribution for the oracle $\omega$ is $a_\omega$ has
\[
 B(\mathcal A)\ge\E_\omega\Bigl[\frac{M_3}{M_2}\Bigr]-\sqrt{\E_\omega\Bigl[\frac{M_4}{M_2}\Bigr]\,\E_\omega\TV(a_\omega,r_{p_\omega})} .
\]
\end{lemma}

\begin{proof}
The score of the collision distribution is $M_3/M_2$ by \eqref{eq:squaring-collision-score}.

\emph{Step 1: the shortfall of $a$ relative to $r_p$.} Let $d_z=\max\{r_p(z)-a(z),0\}$, so that
\[
 0\le d_z\le r_p(z),\qquad \sum_zd_z=\TV(a,r_p),\qquad a(z)\ge r_p(z)-d_z .
\]
Then, since the score of $r_p$ is $M_3/M_2$,
\[
 \sum_za(z)Np_z\ge\frac{M_3}{M_2}-\sum_zNp_zd_z .
\]

\emph{Step 2: the Cauchy--Schwarz bound.} By the Cauchy--Schwarz inequality, and then by the properties of $d_z$ in Step 1,
\begin{align*}
 \sum_zNp_zd_z&\le\Bigl(\sum_zN^2p_z^2d_z\Bigr)^{1/2}\Bigl(\sum_zd_z\Bigr)^{1/2}\\
 &\le\Bigl(\sum_zN^2p_z^2r_p(z)\Bigr)^{1/2}\TV(a,r_p)^{1/2},
\end{align*}
and
\[
 \sum_zN^2p_z^2r_p(z)=\frac{N^3\sum_zp_z^4}{N\sum_zp_z^2}=\frac{M_4}{M_2} .
\]
With Step 1, this proves the first claim.

For the second claim, average over $\omega$ and then bound $\E_\omega\bigl[(M_4/M_2)^{1/2}\TV(a_\omega,r_{p_\omega})^{1/2}\bigr]$ by the Cauchy--Schwarz inequality.
\end{proof}

\begin{lemma}[The limiting scores]\label{lem:squaring-limits}
Recall that $(\kappa_3,\kappa_4)=(6,24)$ in the Haar ensemble and $(15,105)$ in the Fourier ensemble, so that $\kappa_3/\kappa_2$ equals $3$ and $5$, and $\kappa_4/\kappa_2$ equals $12$ and $35$. Then
\[
 \E\Bigl[\frac{M_3}{M_2}\Bigr]\ge\frac{\kappa_3}{\kappa_2}-\frac{C_{\mathrm{lim}}}{\sqrt N},\qquad
 \E\Bigl[\frac{M_4}{M_2}\Bigr]\le\frac{\kappa_4}{\kappa_2}+\frac{C_{\mathrm{lim}}}{\sqrt N}
\]
for an absolute constant $C_{\mathrm{lim}}$.
\end{lemma}

\begin{proof}
\emph{Step 1: comparison with $\E M_j/\kappa_2$.} Let $j\in\{3,4\}$. Since $M_2\ge1$,
\[
 \Bigl|\frac{M_j}{M_2}-\frac{M_j}{\kappa_2}\Bigr|=\frac{M_j\,|M_2-\kappa_2|}{\kappa_2M_2}\le\frac{M_j\,|M_2-\kappa_2|}{\kappa_2},
\]
so, by the Cauchy--Schwarz inequality,
\[
 \Bigl|\E\Bigl[\frac{M_j}{M_2}\Bigr]-\frac{\E M_j}{\kappa_2}\Bigr|\le\frac{\bigl(\E M_j^2\;\E(M_2-\kappa_2)^2\bigr)^{1/2}}{\kappa_2} .
\]

\emph{Step 2: the moment bounds.} By \cref{lem:moments},
\[
 \E M_3\ge\kappa_3-30/N,\qquad \E M_4\le\kappa_4,\qquad \E(M_2-\kappa_2)^2\le28/N,
\]
and $\E M_3^2$ and $\E M_4^2$ are bounded by absolute constants, in both ensembles; in the Haar ensemble the first bound holds because $\kappa_3-\E M_3=6(3N+2)/((N+1)(N+2))\le18/N$. With Step 1, this gives both bounds.
\end{proof}

\begin{theorem}[Scores of the squaring algorithm]\label{thm:squaring-scores}
Let $N\ge64$ and $m=\lceil N^{1/3}\rceil$. Then the squaring algorithm has the following properties.
\begin{enumerate}
\item[(i)] It makes $O(N^{1/3})$ queries.
\item[(ii)] Its output distribution satisfies $\E_\omega\TV(a_\omega,r_{p_\omega})=O(N^{-1/3})$.
\item[(iii)] Its expected score is at least $5-O(N^{-1/6})$ in the Fourier ensemble and at least $3-O(N^{-1/6})$ for either form of Haar access.
\end{enumerate}
\end{theorem}

\begin{proof}
Combine \cref{thm:output-law} with \cref{lem:squaring-tv-score,lem:squaring-limits}. The score is then at least
\[
 \frac{\kappa_3}{\kappa_2}-O(N^{-1/2})-\Bigl(\bigl(\kappa_4/\kappa_2+O(N^{-1/2})\bigr)\,O(N^{-1/3})\Bigr)^{1/2} . \qedhere
\]
\end{proof}

This proves part~3 of \cref{thm:intro-xeb}. The algorithm reaches every fixed score below $5$ in the Fourier ensemble and below $3$ in the Haar ensemble, which are the limiting scores of the collision distribution. Together with \cref{thm:fourier-lower} and, for the Haar ensemble, \cref{thm:haar-lower}, its \cref{cor:haar-birthday}, and \cref{thm:haar-unitary}, it also gives the formal version of \cref{cor:collision}.

\begin{corollary}[Sampling the collision distribution]\label{cor:collision-formal}
Fix $0<\nu<4/35$ in the Fourier ensemble, or $0<\nu<1/12$ for either form of Haar access.
\begin{enumerate}
\item[(i)] Every algorithm with no supplied copies whose output distribution satisfies $\E_\omega\TV(a_\omega,r_{p_\omega})\le\nu$ has expected score at~least
\[
 \frac{\kappa_3}{\kappa_2}-\sqrt{(\kappa_4/\kappa_2)\,\nu}-O(N^{-1/2}),
\]
which exceeds the ideal score by a positive constant for all large $N$, and therefore makes $\Omega(N^{1/3})$ queries.
\item[(ii)] For $N\ge64$, the squaring algorithm with $m=\lceil N^{1/3}\rceil$ makes $O(N^{1/3})$ queries and has $\E_\omega\TV(a_\omega,r_{p_\omega})=O(N^{-1/3})$.
\end{enumerate}
\end{corollary}

\begin{proof}
The score bound follows from \cref{lem:squaring-tv-score,lem:squaring-limits}, using
\[
 \sqrt{(c+t)\nu}-\sqrt{c\nu}=\frac{t\sqrt\nu}{\sqrt{c+t}+\sqrt c}\le\frac{t\sqrt\nu}{2\sqrt c}
\]
with $c=\kappa_4/\kappa_2$ and $t=C_{\mathrm{lim}}/\sqrt N$.

The thresholds are where the bound meets the ideal score:
\begin{align*}
 5-\sqrt{35\nu}>3&\quad\text{if and only if}\quad\nu<4/35,\\
 3-\sqrt{12\nu}>2&\quad\text{if and only if}\quad\nu<1/12 .
\end{align*}

The query lower bound then follows from \cref{thm:fourier-lower} in the Fourier ensemble and from \cref{cor:haar-birthday}, with no copies, and \cref{thm:haar-unitary} in the Haar ensemble. Part~(ii) restates parts~(i) and~(ii) of \cref{thm:squaring-scores}.
\end{proof}

The lower bounds of this paper concern the expected score, and \cref{sec:intro-results} explains how they apply to the empirical score of a list of samples. In the other direction, the empirical score of independent runs of the squaring algorithm is, with high probability, at least any fixed value below its limiting score.

\begin{corollary}[Empirical scores]\label{cor:empirical}
Let $N\ge64$ and $m=\lceil N^{1/3}\rceil$, let $Z_1,\dots,Z_k$ be the outputs of $k$ independent runs of $\mathsf{Sq}_m$ on the same oracle $\omega$, each with a fresh stored set, and let
\[
 \widehat B_k=\frac1k\sum_{i=1}^kNp_\omega(Z_i)
\]
be their empirical score. Then, for $0<t\le1$,
\[
 \Pr\Bigl[\widehat B_k\le\frac{\kappa_3}{\kappa_2}-t\Bigr]\le\frac C{t^2}\Bigl(\frac1k+N^{-1/6}\Bigr)
\]
for an absolute constant $C$, where the probability is over the oracle and the $k$ runs, and $\kappa_3/\kappa_2$ is $5$ in the Fourier ensemble and $3$ for either form of Haar access. In particular, for every fixed $b<\kappa_3/\kappa_2$ and $\zeta>0$, $k=O(1/\zeta)$ runs suffice for the empirical score to be at least $b$ with probability at least $1-\zeta$, for all large $N$.
\end{corollary}

\begin{proof}
\emph{Step 1: reduction by Chebyshev's inequality.} Let
\[
 B_\omega=\sum_za_\omega(z)Np_\omega(z),\qquad Q_\omega=\sum_za_\omega(z)N^2p_\omega(z)^2
\]
be the mean and the second moment of the reward of one run conditioned on $\omega$. Conditioned on $\omega$, the rewards $Np_\omega(Z_i)$ are independent with mean $B_\omega$ and variance at most $Q_\omega$, so Chebyshev's inequality gives
\[
 \Pr[\widehat B_k\le B_\omega-t/2\mid\omega]\le\frac{4Q_\omega}{kt^2},
\]
and therefore
\[
 \Pr\Bigl[\widehat B_k\le\frac{\kappa_3}{\kappa_2}-t\Bigr]\le\Pr\Bigl[B_\omega\le\frac{\kappa_3}{\kappa_2}-\frac t2\Bigr]+\frac{4\,\E_\omega Q_\omega}{kt^2} .
\]

\emph{Step 2: the second moment.} Write $X^\ast_\omega=N\max_zp_\omega(z)$. Since $\sum_zr_p(z)N^2p_z^2=M_4/M_2\le M_4$, as in the proof of \cref{lem:squaring-tv-score}, and $\sum_z|a_\omega(z)-r_p(z)|=2\TV(a_\omega,r_p)$, we have
\[
 Q_\omega\le M_4+2\TV(a_\omega,r_p)(X^\ast_\omega)^2 .
\]
Since $\TV(a_\omega,r_p)\le1$, the Cauchy--Schwarz inequality gives
\[
 \E_\omega Q_\omega\le\kappa_4+2\sqrt{\E_\omega\TV(a_\omega,r_p)}\,\sqrt{\E_\omega (X^\ast_\omega)^4} .
\]

In both ensembles, $\Pr[X^\ast_\omega\ge u]\le2Ne^{-u/2}$ for $u\ge0$. For Fourier sampling this is part of \cref{lem:fr-max-weight}. For a Haar-random state, \cref{lem:dirichlet}(iii) with $S=\{z\}$ shows that $p_z$ has the density $(N-1)(1-s)^{N-2}$ on $[0,1]$, so, for $0\le u\le N$,
\begin{align*}
 \Pr[Np_z\ge u]&=(1-u/N)^{N-1}\le e^{-u(N-1)/N}\\
 &\le e^{-u/2}\qquad\text{since $N-1\ge N/2$},
\end{align*}
and $\Pr[Np_z\ge u]=0$ for $u>N$; the union bound over $z$ then applies. Hence, with $u_0=2\ln(2N)\le2(n+1)$, the tail bound gives
\begin{align*}
 \E (X^\ast_\omega)^4&\le u_0^4+\int_{u_0}^\infty4u^3\cdot2Ne^{-u/2}\,du\\
 &=u_0^4+\int_0^\infty4(v+u_0)^3e^{-v/2}\,dv\\
 &\le u_0^4+32u_0^3+1536 ,
\end{align*}
using $(v+u_0)^3\le4(v^3+u_0^3)$ in the last step. With $\E_\omega\TV(a_\omega,r_p)=O(N^{-1/3})$ from \cref{thm:output-law}, this gives $\E_\omega Q_\omega\le\kappa_4+O(n^2N^{-1/6})$, which is bounded by an absolute constant.

\emph{Step 3: the conditional mean.} By \cref{lem:squaring-tv-score} and $M_2\ge1$,
\[
 B_\omega\ge\frac{M_3}{M_2}-\sqrt{M_4\TV(a_\omega,r_p)} .
\]
Let $\mathcal E_1$ be the event $|M_3/M_2-\kappa_3/\kappa_2|\ge t/4$, and let $\mathcal E_2$ be the event $M_4\TV(a_\omega,r_p)\ge t^2/16$. By this bound on $B_\omega$ and the union bound,
\[
 \Pr\Bigl[B_\omega\le\frac{\kappa_3}{\kappa_2}-\frac t2\Bigr]\le\Pr[\mathcal E_1]+\Pr[\mathcal E_2],
\]
and, by Markov's inequality for nonnegative random variables,
\begin{align*}
 \Pr[\mathcal E_1]&\le\frac4t\,\E\Bigl|\frac{M_3}{M_2}-\frac{\kappa_3}{\kappa_2}\Bigr|,\\
 \Pr[\mathcal E_2]&\le\frac{16}{t^2}\,\E\bigl[M_4\TV(a_\omega,r_p)\bigr] .
\end{align*}
Since $M_2\ge1$,
\[
 \Bigl|\frac{M_3}{M_2}-\frac{\kappa_3}{\kappa_2}\Bigr|\le|M_3-\kappa_3|+\frac{\kappa_3|M_2-\kappa_2|}{\kappa_2},
\]
whose expectation is $O(N^{-1/2})$, because ${\E|M_3-\kappa_3|}\le(\Var M_3+(\E M_3-\kappa_3)^2)^{1/2}$ and ${\E|M_2-\kappa_2|}\le(\E(M_2-\kappa_2)^2)^{1/2}$, while \cref{lem:moments} gives $\Var M_3\le10^4/N$, $|\E M_3-\kappa_3|\le30/N$, and $\E(M_2-\kappa_2)^2\le28/N$ in both ensembles. Moreover, using $\TV(a_\omega,r_p)\le1$,
\[
 \E[M_4\TV(a_\omega,r_p)]\le\sqrt{\E M_4^2\,\E\TV(a_\omega,r_p)}=O(N^{-1/6}) .
\]

\emph{Step 4: conclusion.} Collecting the bounds and using $t\le1$ proves the first claim. The second follows by taking $t=\min\{1,\kappa_3/\kappa_2-b\}$ and $k\ge2C/(\zeta t^2)$, for $N$ with $CN^{-1/6}/t^2\le\zeta/2$.
\end{proof}

\subsection[Heavy outputs]{Heavy outputs}\label{sec:heavy-outputs}
The collision distribution also determines the probability of producing a heavy output. The original heavy output generation (HOG) task asks for a list with at least $2/3$ of its outputs above the median ideal probability~\cite[Problem~1]{AaronsonChen}. Here we study the expected heavy probability for the Haar ensemble. This differs from linear XEB, which weights an output by its ideal probability~\cite{AaronsonGunn,Kretschmer}. Honest sampling already exceeds $2/3$ in the limit; our question is how many queries are needed to improve on its heavy probability by a fixed amount.

Let $N=2^n\ge2$ and $p_z=|\langle z|\phi\rangle|^2$, where $\ket\phi$ is Haar-random, or $\ket\phi=U\ket{0^n}$ for Haar-random $U$. Queries may be controlled, and unitary queries may also be inverse queries. The other initial registers have a joint state independent of $\phi$ or $U$; supplied copies, when allowed, are $\ket\phi^{\otimes s}$ in the device registers, tensored with that state. Define
\begin{equation}\label{eq:heavy-threshold}
 \theta_N=1-2^{-1/(N-1)},\qquad
 B^{\mathrm{hog}}_\tau(\mathcal A)=\Pr[p_Z>\tau].
\end{equation}
The probability averages over the oracle and the algorithm's randomness. The fixed threshold $\theta_N$ is the median of the marginal law $\operatorname{Beta}(1,N-1)$. For comparison, let $\widehat\theta_N=(p_{(N/2)}+p_{(N/2+1)})/2$ be the empirical median, where the probabilities are sorted in increasing order. Honest sampling at the fixed threshold has probability
\begin{equation}\label{eq:heavy-honest}
 \beta_N^{\mathrm{hog}}=\frac{1+(N-1)\theta_N}{2}
 \le h_{\mathrm{hon}}:=\frac{1+\ln2}{2}=0.84657\ldots,
 \qquad \beta_N^{\mathrm{hog}}\longrightarrow h_{\mathrm{hon}}.
\end{equation}

\begin{theorem}[Heavy-output query bounds]\label{thm:heavy-outputs}
With $\Lambda_{s,q}$ as in \cref{thm:haar-lower}, an algorithm given $s$ copies and making at most $q$ canonical queries satisfies the first bound below; an algorithm making at most $q$ unitary queries, without supplied copies, satisfies the second:
\begin{align}
 B^{\mathrm{hog}}_{\theta_N}(\mathcal A)
 &\le \beta_N^{\mathrm{hog}}+\Lambda_{s,q},
 \label{eq:heavy-canonical}\\
 B^{\mathrm{hog}}_{\theta_N}(\mathcal A)
 &\le \beta_N^{\mathrm{hog}}+\Lambda_{0,2q}
       +\frac{5q+2}{\sqrt N}.
 \label{eq:heavy-unitary}
\end{align}
Set
\begin{equation}\label{eq:heavy-limit}
 h_\star=\frac{2+2\ln2+(\ln2)^2}{4}=0.96669\ldots.
\end{equation}
For either threshold $\tau=\theta_N$ or $\tau=\widehat\theta_N$, and every fixed $h\in(h_{\mathrm{hon}},h_\star)$, achieving $B^{\mathrm{hog}}_\tau(\mathcal A)\ge h$ without supplied copies has query complexity $\Theta_h(N^{1/3})$ in either Haar access model. For $N\ge64$, $\mathsf{Sq}_m$ with $m=\lceil N^{1/3}\rceil$ makes at most $3N^{1/3}$ queries and its heavy probability tends to $h_\star$ under both threshold conventions.
\end{theorem}

The proof is in Appendix~\ref{app:heavy-outputs}. The lower bound uses the same progress measure as \cref{sec:haar}, but the heavy-output reward is bounded by one. Its expectation is therefore controlled by the distance from the normalized level-zero part. The upper bound follows from the collision distribution approximation in \cref{thm:output-law}. The appendix also gives a sharper bound for arbitrary joint processing of supplied copies with no queries.

\section{Certified randomness at the ideal score}\label{sec:randomness}

In this section we prove \cref{thm:intro-randomness}. We consider the experiment of \cref{def:experiment} as a sequence of controlled queries and oracle-independent unitaries, as supplied by \cref{lem:normal-form}, with device register $A$, Eve's register $E$, and initial state $\psi_0\in A\otimes E$, whose reduced state on $E$ is $\rho_E$. The final state and its branches are
\[
 \Psi(\omega)=(V_\omega\otimes I_E)\psi_0,\qquad \Psi_z(\omega)=(\Pi_z\otimes I_E)\Psi(\omega),
\]
and the reduced states of the branches on $E$ are Eve's conditional states $\rho^\omega_z$. The device makes at most $q$ queries, and its expected score is at least the ideal score minus $\delta$.

Both ensembles are treated by the same three steps.
\begin{enumerate}
\item The score bound of the corresponding lower-bound section, with the terms that its proof discards now retained, shows that the branches $\Psi_z$ are close to \emph{auxiliary vectors} $G_z$.
\item For every fixed oracle, the auxiliary vector is the image of $\psi_0$ under an operator on the device alone whose norm is at most $\mathrm{poly}(q)\sqrt{p_\omega(z)}$.
\item The family of auxiliary vectors, normalized, is a state within purified distance $\eta$ of the actual one, and in this state Eve's guessing probability is at most $\mathrm{poly}(q)$
times her probability of guessing an ideal sample, provided $\eta\le1/2$.
\end{enumerate}
The third step is common to both ensembles, and we present it first.

\subsection[Min-entropy from auxiliary vectors]{Min-entropy from auxiliary vectors}\label{sec:randomness-assembly}

\begin{proposition}[Min-entropy from auxiliary vectors]\label{prop:randomness-assembly}
Suppose that for every oracle $\omega$ and every $z$ there is a linear operator $K_z(\omega)$ on $A$, depending measurably on $\omega$, such that the vectors $G_z(\omega)=(K_z(\omega)\otimes I_E)\psi_0$ satisfy
\begin{equation}\label{eq:assembly-hypotheses}
 \|K_z(\omega)\|^2\le\kappa\,p_\omega(z)\quad\text{for all $\omega$ and $z$},\qquad
 \E_\omega\sum_z\|\Psi_z(\omega)-G_z(\omega)\|^2\le\eta_0^2
\end{equation}
for some $\kappa>0$ and $0\le\eta_0<1$. Let $\xi=\E_\omega\sum_z\|G_z(\omega)\|^2$, which is finite, because $\sum_z\|G_z(\omega)\|^2\le\kappa\sum_zp_\omega(z)=\kappa$, and at least $(1-\eta_0)^2>0$ by the second hypothesis in \eqref{eq:assembly-hypotheses}, and let $\tau^\omega_z$ be the reduced state of $G_z(\omega)$ on $E$ divided by $\xi$. Then the following hold.
\begin{enumerate}
\item[(i)] The family $\tau$ is a normalized family of conditional states with $P(\rho,\tau)\le\eta_0$.
\item[(ii)] The family $\tau$ satisfies
\begin{equation}\label{eq:assembly-domination}
 \tau^\omega_z\preceq\frac\kappa\xi\,p_\omega(z)\,\rho_E\qquad\text{for all $\omega$ and $z$},
\end{equation}
so that for every $\omega$ and every measurement $(M_z)_z$ on $E$,
\begin{equation}\label{eq:assembly-guess}
 \sum_z\operatorname{Tr}(M_z\tau^\omega_z)\le\frac{\kappa}\xi\max_zp_\omega(z) .
\end{equation}
\item[(iii)] Consequently,
\begin{equation}\label{eq:assembly-entropy}
 \Hmin^{\eta}(Z\mid RE)_\rho\ge n-\log\kappa-\log\Bigl(N\,\E_\omega\max_zp_\omega(z)\Bigr)-\log\frac1\xi\qquad\text{for every }\eta\in[\eta_0,1].
\end{equation}
\item[(iv)] If $\E_\omega\sum_z\langle\Psi_z(\omega),G_z(\omega)\rangle=\xi$, as when, for each $z$, the function $\omega\mapsto G_z(\omega)$ is the image of $\omega\mapsto\Psi_z(\omega)$ under an orthogonal projection of the space of square-integrable functions of $\omega$, then
\[
 \xi=1-\E_\omega\sum_z\|\Psi_z(\omega)-G_z(\omega)\|^2\ge1-\eta_0^2 .
\]
\end{enumerate}
\end{proposition}

The family $\tau$ is normalized by the total weight $\xi$ rather than oracle by oracle. This is what makes it a normalized family in the sense of \cref{def:entropies}, and it allows the fidelity of $\rho$ and $\tau$ to be bounded from below by a single inner product. It is a comparison family for the smoothing in \cref{def:entropies} and not the state of another experiment: for a fixed $\omega$, the weight $\sum_z\operatorname{Tr}\tau^\omega_z$ need not be $1$.

\begin{proof}
\emph{Step 1: the scalar $\xi$.} We have
\[
 \E_\omega\sum_z\|\Psi_z(\omega)\|^2=1,\qquad
 \E_\omega\sum_z\|\Psi_z(\omega)-G_z(\omega)\|^2\le\eta_0^2,
\]
where the second bound is the second hypothesis in \eqref{eq:assembly-hypotheses}, and by definition $\xi=\E_\omega\sum_z\|G_z(\omega)\|^2$. Hence, by the triangle inequality in the direct sum over $z$ and the average over $\omega$,
\begin{align*}
 \sqrt\xi&\ge\Bigl(\E_\omega\sum_z\|\Psi_z(\omega)\|^2\Bigr)^{1/2}-\Bigl(\E_\omega\sum_z\|\Psi_z(\omega)-G_z(\omega)\|^2\Bigr)^{1/2}\\
 &\ge1-\eta_0>0 .
\end{align*}
Thus $\xi\ge(1-\eta_0)^2>0$, as stated, and the rescaled vectors $G_z(\omega)/\sqrt\xi$ satisfy $\E_\omega\sum_z\bigl\|G_z(\omega)/\sqrt\xi\bigr\|^2=1$.

\emph{Step 2: the overlap $c$.} Let $c=\xi^{-1/2}\,\E_\omega\sum_z\langle G_z(\omega),\Psi_z(\omega)\rangle$ be the overlap of the branches with the rescaled vectors. For every scalar $t$, expanding the square and using the two normalizations of Step 1 and the definition of $c$, we obtain
\begin{align*}
 \E_\omega\sum_z\bigl\|\Psi_z(\omega)-tG_z(\omega)/\sqrt\xi\bigr\|^2
 &=1-2\operatorname{Re}(\bar tc)+|t|^2\\
 &=|c-t|^2+1-|c|^2 .
\end{align*}
The squared distance on the left is therefore minimized over $t$ at $t=c$, where it equals $1-|c|^2$, so $1-|c|^2$ is at most its value at $t=\sqrt\xi$. That value is the left side of the second hypothesis in \eqref{eq:assembly-hypotheses}, and hence
\begin{equation}\label{eq:assembly-overlap}
 1-|c|^2\le\E_\omega\sum_z\|\Psi_z(\omega)-G_z(\omega)\|^2\le\eta_0^2 .
\end{equation}

\emph{Step 3: the bound on $P(\rho,\tau)$.} The family $\tau$ is normalized by construction, since
\[
 \E_\omega\sum_z\operatorname{Tr}\tau^\omega_z=\E_\omega\sum_z\|G_z(\omega)\|^2/\xi=1 .
\]
For each $\omega$ and $z$, the vectors $\Psi_z(\omega)$ and $G_z(\omega)/\sqrt\xi$ in $A\otimes E$ have reduced states $\rho^\omega_z$ and $\tau^\omega_z$ on $E$, so \eqref{eq:prelim-uhlmann} gives
\[
 \|\sqrt{\rho^\omega_z}\sqrt{\tau^\omega_z}\|_1\ge\frac{|\langle\Psi_z(\omega),G_z(\omega)\rangle|}{\sqrt\xi} .
\]
Summing over $z$, averaging over $\omega$, applying the triangle inequality to the sum and the average, and using \eqref{eq:assembly-overlap}, we obtain
\begin{align*}
 F(\rho,\tau)&\ge\frac1{\sqrt\xi}\,\Bigl|\E_\omega\sum_z\langle\Psi_z(\omega),G_z(\omega)\rangle\Bigr|=|c|\\
 &\ge\sqrt{1-\eta_0^2},
\end{align*}
hence $P(\rho,\tau)\le\eta_0$.

\emph{Step 4: the domination bound and the entropy bound.} For \eqref{eq:assembly-domination}, fix $\omega$ and $z$. By \cref{lem:local-operator} applied to $K_z(\omega)$ and the first hypothesis in \eqref{eq:assembly-hypotheses}, the reduced state $\tilde\rho^\omega_z$ of $G_z(\omega)$ on $E$ satisfies
\[
 \operatorname{Tr}(M\tilde\rho^\omega_z)\le\kappa\,p_\omega(z)\operatorname{Tr}(M\rho_E)\qquad\text{for every positive semidefinite $M$},
\]
which is the operator inequality $\tilde\rho^\omega_z\preceq\kappa\,p_\omega(z)\rho_E$. Dividing by $\xi$ gives \eqref{eq:assembly-domination}. For a measurement $(M_z)_z$, it follows that
\[
 \sum_z\operatorname{Tr}(M_z\tau^\omega_z)\le\frac\kappa\xi\sum_zp_\omega(z)\operatorname{Tr}(M_z\rho_E)\le\frac\kappa\xi\max_zp_\omega(z),
\]
because $\sum_z\operatorname{Tr}(M_z\rho_E)=\operatorname{Tr}\rho_E=1$.

Taking the supremum over measurements and averaging over $\omega$ gives
\[
 p_{\mathrm{guess}}(Z\mid RE)_\tau\le(\kappa/\xi)\,\E_\omega\max_zp_\omega(z),
\]
and \eqref{eq:assembly-entropy} follows from \cref{def:entropies}, since $P(\rho,\tau)\le\eta_0\le\eta$ and $\log N=n$.

\emph{Step 5: part (iv).} Finally, assume the hypothesis of (iv), so that $\E_\omega\sum_z\langle\Psi_z(\omega),G_z(\omega)\rangle=\xi$. Expanding the squares and using $\E_\omega\sum_z\|\Psi_z(\omega)\|^2=1$ gives
\[
 \E_\omega\sum_z\|\Psi_z(\omega)-G_z(\omega)\|^2=1-2\xi+\xi=1-\xi . \qedhere
\]
\end{proof}

Nothing in the proposition or its proof depends on the dimensions of $A$ and $E$. The proposition also explains why smoothing is needed: the actual branches need not satisfy a bound like \eqref{eq:assembly-guess}, because a device may output a string known to Eve with a small probability, which changes its score only slightly but makes that branch fully predictable, and the nearby state removes such branches.

\subsection{Fourier sampling}\label{sec:randomness-fourier}

We use the notation of \cref{sec:fourier}: the space $\mathcal F$ of functions of $f$ with values in $A\otimes E$, the expansion \eqref{eq:fourier-expansion} in recorded sets, the operators $a_x$, $b_z$, and $D_z=b_z^\dagger b_z$, the final state $\Psi$ with branches $\Psi_z$, the progress measure $\lambda$, and the bound $\bar\lambda_q$ of \cref{lem:fourier-growth}.

\begin{theorem}[Certified randomness from Fourier sampling]\label{thm:fourier-randomness}
Let $q\ge1$ and $\delta\ge0$, and consider an experiment in the Fourier ensemble in which the device makes at most $q$ queries and has score $B(\mathcal A)\ge3-2/N-\delta$. Let $\eta_0$ and $\bar\lambda_q$ be given by
\begin{equation}\label{eq:fr-eta0}
\begin{aligned}
 \eta_0&=\bar\lambda_q+\biggl(\frac\delta{2(1-1/N)}+\frac{\bar\lambda_q}{1-1/N}+\bar\lambda_q^2\biggr)^{1/2},\\
 \bar\lambda_q&=\frac2{\sqrt N}\sum_{t=0}^{q-1}\sqrt t .
\end{aligned}
\end{equation}
If $\eta_0<1$, then for every $\eta\in[\eta_0,1]$,
\begin{equation}\label{eq:fr-main-bound}
 \Hmin^\eta(Z\mid RE)_\rho\ge n-4\log q-\log(n+2)-3-2\log\frac1{1-\eta_0} .
\end{equation}
Moreover, $\bar\lambda_1=0$, and
\begin{equation}\label{eq:fr-eta0-simple}
 \eta_0\le\sqrt\delta+2\,q^{3/4}N^{-1/4}+3\,q^{3/2}N^{-1/2} .
\end{equation}
\end{theorem}

The proof occupies this subsection up to \cref{lem:fr-max-weight}. We first identify the auxiliary vectors, then bound the operators that produce them, and finally apply \cref{prop:randomness-assembly}.

To identify the auxiliary vectors, we show that a near-ideal score forces each branch to be close to the vector obtained by removing one recorded point in the state $\ket{\chi_z}$ and multiplying by $\sqrt N\widehat f(z)$, which is the structure described in \cref{sec:overview-fourier-randomness}.

\begin{lemma}[Auxiliary vectors]\label{lem:fr-comparison}
Assume $B(\mathcal A)\ge3-2/N-\delta$, let $\eta_0$ be as in \eqref{eq:fr-eta0}, and put
\[
 G_z=(b_z+b_z^\dagger)\,b_z\Psi_z=b_z^2\Psi_z+D_z\Psi_z .
\]
Then
\begin{align*}
 \sum_z\|(D_z-I)\Psi_z\|^2&\le\frac\delta{2(1-1/N)}+\frac{\bar\lambda_q}{1-1/N}+\bar\lambda_q^2,\\
 \Bigl(\sum_z\|\Psi_z-G_z\|^2\Bigr)^{1/2}&\le\eta_0 .
\end{align*}
\end{lemma}

\begin{proof}
\emph{Step 1: the first bound.} The hypothesis and \eqref{eq:fourier-score-bound} give
\[
 2\Bigl(1-\frac1N\Bigr)\sum_z\|(D_z-I)\Psi_z\|^2\le\delta+2\lambda+2\Bigl(1-\frac1N\Bigr)\lambda^2,
\]
where $\lambda=\lambda(\Psi)$ satisfies $\lambda\le\bar\lambda_q$ by \cref{lem:fourier-growth}. Dividing by $2(1-1/N)$ gives the first bound.

\emph{Step 2: the second bound.} Since $G_z=b_z^2\Psi_z+D_z\Psi_z$, we have $\Psi_z-G_z=-(D_z-I)\Psi_z-b_z^2\Psi_z$, and the triangle inequality in the direct sum over $z$ gives
\[
 \Bigl(\sum_z\|\Psi_z-G_z\|^2\Bigr)^{1/2}\le\Bigl(\sum_z\|(D_z-I)\Psi_z\|^2\Bigr)^{1/2}+\Bigl(\sum_z\|b_z^2\Psi_z\|^2\Bigr)^{1/2} .
\]
The last sum is at most $\lambda^2$, because $\Pi_z$ commutes with $b_z^2$, so that
\[
 \sum_z\|b_z^2\Psi_z\|^2=\sum_z\|\Pi_zb_z^2\Psi\|^2\le\sum_z\|b_z^2\Psi\|^2 ,
\]
and the right side is $\lambda^2$ by the definition of $\lambda$. Hence, by the first bound and $\lambda\le\bar\lambda_q$, the left side of the triangle inequality is at most $\eta_0$.
\end{proof}

\paragraph{The device as a bounded polynomial.}
For a fixed $f$, the vector $G_z(f)$ can be written in terms of the unitary $V_f$ of the device. Let $(\partial_x\Phi)(f)=\bigl(\Phi(f^{x,+})-\Phi(f^{x,-})\bigr)/2$ be the discrete derivative of \cref{sec:fourier-recording}, which is the operator $a_x$, and let $\partial_{\chi_z}=N^{-1/2}\sum_x\chi_z(x)\,\partial_x$ be the derivative along the character $\chi_z$, which is the operator $b_z$. Both act on functions of $f$ with values in any vector space, and they commute with the pointwise application of any linear map.

The final state is the pointwise image of the operator-valued function $f\mapsto V_f$ under the linear map $T\mapsto(T\otimes I_E)\psi_0$, and $\Pi_z$ commutes with $b_z$, so
\[
 (b_z\Psi_z)(f)=(\Pi_z(\partial_{\chi_z}V)(f)\otimes I_E)\psi_0 .
\]
Multiplying by $\sqrt N\widehat f(z)$, which is the operator $b_z+b_z^\dagger$, gives
\begin{equation}\label{eq:fr-Gz-pointwise}
\begin{aligned}
 G_z(f)&=\sqrt N\widehat f(z)\,\bigl(\Pi_z(\partial_{\chi_z}V)(f)\otimes I_E\bigr)\psi_0=(K_z(f)\otimes I_E)\psi_0,\\
 K_z(f)&=\sqrt N\widehat f(z)\,\Pi_z\,(\partial_{\chi_z}V)(f),
\end{aligned}
\end{equation}
where $(\partial_{\chi_z}V)(f)$ is an operator on $A$. We now bound its norm, for every $f$, by Markov's inequality for polynomials.

\begin{lemma}[The multilinear extension]\label{lem:fr-polynomial}
Write $V_f=\sum_{|S|\le q}\widehat V_S\chi_S(f)$ for the expansion \eqref{eq:fourier-expansion} of the operator-valued function $f\mapsto V_f$. The coefficients $\widehat V_S$ vanish for $|S|>q$ because $V_f$ is a product of fixed unitaries and of $q$ factors $P_\varnothing+\sum_xf(x)P^{\mathrm c}_x$, each affine in the values of $f$. Define
\[
 \widetilde V(u)=\sum_{|S|\le q}\widehat V_S\prod_{y\in S}u_y\qquad\text{for }u\in\R^{\{0,1\}^n}.
\]
Then:
\begin{enumerate}
\item[(i)] $\widetilde V(f)=V_f$ for every $f$;
\item[(ii)] $\|\widetilde V(u)\|\le1$ for every $u\in[-1,1]^{\{0,1\}^n}$;
\item[(iii)] $(\partial\widetilde V/\partial u_x)(f)=(\partial_xV)(f)$ for every $f$ and $x$.
\end{enumerate}
\end{lemma}

\begin{proof}
Part~(i) is the expansion.

For (ii), let $\tilde f$ be a random function whose values are independent signs with $\E\tilde f(y)=u_y$. Then $\E\chi_S(\tilde f)=\prod_{y\in S}u_y$, so $\widetilde V(u)=\E V_{\tilde f}$ is an average of unitaries and has norm at most $1$.

For (iii), $\widetilde V$ is affine in each coordinate: $\widetilde V(u)=\widetilde V_0(u)+u_x\widetilde V_1(u)$ with $\widetilde V_0,\widetilde V_1$ independent of $u_x$, so $\partial\widetilde V/\partial u_x=\widetilde V_1$, while, by (i),
\[
 (\partial_xV)(f)=\frac12\bigl(\widetilde V(f^{x,+})-\widetilde V(f^{x,-})\bigr)=\widetilde V_1(f) . \qedhere
\]
\end{proof}

The extension must be the multilinear one: another polynomial that agrees with $V$ on the signs, such as the product of the operators $P_\varnothing+\sum_xu_xP^{\mathrm c}_x$, is also bounded on the cube, but its partial derivatives at the signs need not equal $\partial_xV$.

\begin{lemma}[Markov's inequality at an endpoint]\label{lem:fr-markov}
Let $\mathcal Y$ be a finite-dimensional normed space, and let $g(t)=\sum_{k=0}^dy_kt^k$ with $y_k\in\mathcal Y$. Then
\[
 \|g'(1)\|\le d^2\max_{-1\le t\le1}\|g(t)\| .
\]
\end{lemma}

\begin{proof}
For $d=0$ there is nothing to prove, so we let $d\ge1$.

\emph{Step 1: interpolation at the Chebyshev nodes.} Let $t_j=\cos(j\pi/d)$ for $j=0,\dots,d$, so that $1=t_0>t_1>\dots>t_d=-1$, and let
\[
 \ell_j(t)=\prod_{k\ne j}\frac{t-t_k}{t_j-t_k}
\]
be the Lagrange polynomials for these nodes. Interpolation at $d+1$ nodes is exact for polynomials of degree at most $d$, coordinate by coordinate, so
\[
 g=\sum_j\ell_jg(t_j)\qquad\text{and}\qquad\|g'(1)\|\le\Bigl(\sum_j|\ell_j'(1)|\Bigr)\max_j\|g(t_j)\| .
\]

\emph{Step 2: the signs of the derivatives $\ell_j'(1)$.} We claim that $(-1)^j\ell_j'(1)>0$ for every $j$.
\begin{enumerate}
\item For $j=0$, we have $\ell_0(1)=1$ and $\ell_0'(1)=\sum_{k\ge1}(1-t_k)^{-1}>0$.
\item For $j\ge1$, $\ell_j$ vanishes at $t_0=1$, so
\[
 \ell_j'(1)=\frac{\prod_{k\notin\{0,j\}}(1-t_k)}{\prod_{k\ne j}(t_j-t_k)},
\]
whose numerator is positive and whose denominator has exactly $j$ negative factors, namely those with $k<j$.
\end{enumerate}
Hence $\sum_j|\ell_j'(1)|=\sum_j(-1)^j\ell_j'(1)$.

\emph{Step 3: the value of this sum.} The Chebyshev polynomial $T_d$, defined by $T_d(\cos\vartheta)=\cos(d\vartheta)$, has degree $d$ and satisfies $T_d(t_j)=(-1)^j$, so interpolation gives
\[
 T_d=\sum_j(-1)^j\ell_j\qquad\text{and}\qquad\sum_j(-1)^j\ell_j'(1)=T_d'(1) .
\]
Differentiating $T_d(\cos\vartheta)=\cos(d\vartheta)$ gives $T_d'(\cos\vartheta)=d\sin(d\vartheta)/\sin\vartheta$ for $0<\vartheta<\pi$, and letting $\vartheta\to0$ gives $T_d'(1)=d^2$.
\end{proof}

\begin{lemma}[Derivative bound]\label{lem:fr-derivative}
For every $f$ and $z$,
\[
 \|(\partial_{\chi_z}V)(f)\|\le2q^2/\sqrt N .
\]
Consequently, $\|K_z(f)\|^2\le4q^4p_f(z)$.
\end{lemma}

\begin{proof}
\emph{Step 1: two polynomials in one variable.} Fix $f$ and $z$, and split the addresses into
\[
 \mathcal I_+=\{x:\chi_z(x)=f(x)\},\qquad \mathcal I_-=\{x:\chi_z(x)=-f(x)\}.
\]
For real $t$, let $u^\pm(t)$ be the point of $\R^{\{0,1\}^n}$ with coordinates $tf(x)$ for $x\in \mathcal I_\pm$ and $f(x)$ otherwise, and let $g_\pm(t)=\widetilde V(u^\pm(t))$. Each monomial $\prod_{y\in S}u_y$ with $|S|\le q$ becomes a constant multiple of $t^{|S\cap \mathcal I_\pm|}$, so $g_\pm$ is a polynomial in $t$ of degree at most $q$ with operator coefficients. Moreover, $\|g_\pm(t)\|\le1$ for $|t|\le1$ by \cref{lem:fr-polynomial}(ii), because $u^\pm(t)$ then lies in the cube.

\emph{Step 2: the derivative along $\chi_z$.} By the chain rule and \cref{lem:fr-polynomial}(iii),
\[
 g_\pm'(1)=\sum_{x\in \mathcal I_\pm}f(x)\frac{\partial\widetilde V}{\partial u_x}(f)=\sum_{x\in \mathcal I_\pm}f(x)\,(\partial_xV)(f) .
\]
Since $f(x)=\chi_z(x)$ on $\mathcal I_+$ and $f(x)=-\chi_z(x)$ on $\mathcal I_-$, we obtain
\[
 \sqrt N\,(\partial_{\chi_z}V)(f)=\sum_x\chi_z(x)(\partial_xV)(f)=g_+'(1)-g_-'(1).
\]

\emph{Step 3: Markov's inequality.} By \cref{lem:fr-markov}, applied in the space of operators on $A$ with the operator norm, we have $\|g_\pm'(1)\|\le q^2$. The second claim follows from \eqref{eq:fr-Gz-pointwise}, $\|\Pi_z\|\le1$, and $N\widehat f(z)^2=Np_f(z)$.
\end{proof}

The splitting into two directions is needed because the normalized character direction $\chi_z/\sqrt N$, whose coordinates are $\pm N^{-1/2}$, does not stay inside the cube when followed from $f$, whereas each of $u^\pm(t)$ does for $|t|\le1$. Apart from the score bound and the growth bound of
\cref{lem:fourier-growth}, the bound on the number of queries enters the
proof only here.

Finally, we bound $\E_f\max_zp_f(z)$, the expected probability of the most likely outcome, which is Eve's probability of guessing an ideal sample.

\begin{lemma}[Maximal Fourier weight]\label{lem:fr-max-weight}
For a uniformly random $f$,
\[
 \E_f\max_zp_f(z)\le\bigl(2\ln(2N)+2\bigr)/N\le2(n+2)/N .
\]
Moreover, for every $u\ge0$,
\[
 \Pr\bigl[\max_zNp_f(z)\ge u\bigr]\le2Ne^{-u/2} .
\]
\end{lemma}

\begin{proof}
Fix $z$. The signs $r_x=\chi_z(x)f(x)$ are independent and uniform, and
\[
 Np_f(z)=N^{-1}\Bigl(\sum_xr_x\Bigr)^2 .
\]
For $\vartheta\ge0$,
\[
 \E e^{\vartheta\sum_xr_x}=(\cosh\vartheta)^N\le e^{N\vartheta^2/2},
\]
because $\cosh\vartheta\le e^{\vartheta^2/2}$, since $(2k)!\ge2^kk!$. By Markov's inequality with $\vartheta=s/N$,
\begin{align*}
 \Pr\Bigl[\sum_xr_x\ge s\Bigr]&\le e^{-\vartheta s}(\cosh\vartheta)^N\le e^{-s^2/(2N)},\\
 \Pr\Bigl[\Bigl|\sum_xr_x\Bigr|\ge s\Bigr]&\le2e^{-s^2/(2N)}\qquad\text{by symmetry}.
\end{align*}

With $s=\sqrt{uN}$, where $u\ge0$, this reads $\Pr[Np_f(z)\ge u]\le2e^{-u/2}$, and the union bound over $z$ gives the second assertion.

The first assertion follows by integrating this tail bound. For every $u_0\ge0$,
\begin{align*}
 \E\max_zNp_f(z)&=\int_0^\infty\Pr\bigl[\max_zNp_f(z)\ge u\bigr]\,du\\
 &\le u_0+\int_{u_0}^\infty2Ne^{-u/2}\,du=u_0+4Ne^{-u_0/2},
\end{align*}
and $u_0=2\ln(2N)$ gives
\[
 2\ln(2N)+2=2(n+1)\ln2+2\le2(n+2) . \qedhere
\]
\end{proof}

\begin{proof}[Proof of \cref{thm:fourier-randomness}]
\emph{Step 1: the hypotheses and the entropy bound.} By \cref{lem:fr-comparison}, the hypotheses \eqref{eq:assembly-hypotheses} of \cref{prop:randomness-assembly} hold with $G_z(f)$ given by \eqref{eq:fr-Gz-pointwise}, with $\kappa=4q^4$ by \cref{lem:fr-derivative}, and with the stated $\eta_0$. Here $\E_f\sum_z\|\Psi_z(f)-G_z(f)\|^2=\sum_z\|\Psi_z-G_z\|^2$, with the norm of $\mathcal F$ on the right.

The proposition, with its divisor $\xi\ge(1-\eta_0)^2$, and \cref{lem:fr-max-weight} give \eqref{eq:fr-main-bound}, since
\[
 \log(4q^4)=2+4\log q,\qquad \log\bigl(N\E_f\max_zp_f(z)\bigr)\le1+\log(n+2) .
\]

\emph{Step 2: the parameters.} The only term of the sum defining $\bar\lambda_1$ is $\sqrt0=0$, and, by \cref{lem:fourier-growth}, $\bar\lambda_q\le(4/3)q^{3/2}N^{-1/2}$. Using $\sqrt{a+b+c}\le\sqrt a+\sqrt b+\sqrt c$ and $1/(1-1/N)\le2$, we obtain
\begin{align*}
 \eta_0&\le\sqrt\delta+\sqrt{2\bar\lambda_q}+2\bar\lambda_q\\
 &\le\sqrt\delta+\sqrt{\frac83}\,q^{3/4}N^{-1/4}+\frac83q^{3/2}N^{-1/2},
\end{align*}
and $\sqrt{8/3}\le2$ and $8/3\le3$ give \eqref{eq:fr-eta0-simple}.
\end{proof}

For $q=1$ the polynomial $\widetilde V$ is affine, so the operator
\[
 (\partial_{\chi_z}V)(f)=N^{-1/2}U_1\Bigl(\sum_x\chi_z(x)P^{\mathrm c}_x\Bigr)U_0
\]
has norm at most $N^{-1/2}$ because the $P^{\mathrm c}_x$ are mutually orthogonal projections, and the factor $4q^4$ may be replaced by $1$. If moreover $\delta=0$, then $\eta_0=0$, the branches equal the auxiliary vectors, and the conditional states satisfy
\[
 \operatorname{Tr}(M\rho^f_z)\le p_f(z)\operatorname{Tr}(M\rho_E)\qquad\text{for all positive semidefinite $M$}.
\]
Since $\sum_z\rho^f_z=\rho_E=\sum_zp_f(z)\rho_E$, all these inequalities are equalities. This recovers the exact privacy statement of \cref{prop:fourier-one-query} for arbitrary one-query devices, not only for phase queries.

For one phase query, a near-ideal score forces the actual family of conditional states to be close to the \emph{ideal family} $\rho^\circ$ with $(\rho^\circ)^f_z=p_f(z)\rho_E$, in which the output is an exact sample from $p_f$ and is independent of Eve's register. This is the robust form of \cref{prop:fourier-one-query}, and it does not pass through a comparison family.

\begin{proposition}[Closeness to the ideal family]\label{prop:fr-one-phase}
Let the device make a single phase query, as defined in \cref{sec:fourier-one-query}, with Eve's register arbitrary, and let $\delta=3-2/N-B(\mathcal A)$ be its deficit, so that $0\le\delta\le2$ by \cref{prop:fourier-one-query}. Write $\eta_\delta=\sqrt{\delta-\delta^2/4}$. Then
\[
 P(\rho,\rho^\circ)\le\eta_\delta,\qquad\text{and consequently}\qquad \Hmin^{\eta_\delta}(Z\mid RE)_\rho\ge n-\log(2(n+2)) .
\]
\end{proposition}

\begin{proof}
\emph{Step 1: the score.} By \eqref{eq:fourier-one-query-state}, $\Psi(f)=\sum_w\widehat f(w)\Psi(\chi_w)$, and by \cref{prop:fourier-one-query},
\[
 B(\mathcal A)=1-2/N+2\sigma\qquad\text{with}\qquad\sigma=\frac1N\sum_w\|\Pi_w\Psi(\chi_w)\|^2,
\]
so $\delta=2(1-\sigma)$. Put $g_w=\Pi_w\Psi(\chi_w)$ and $G^\circ_z(f)=\widehat f(z)\,g_z$.

\emph{Step 2: the fidelity with the ideal family.} The vector $\widehat f(z)\Psi(\chi_z)$ has reduced state $p_f(z)\rho_E=(\rho^\circ)^f_z$ on $E$, because $\Psi(\chi_z)=(V_{\chi_z}\otimes I_E)\psi_0$ with $V_{\chi_z}$ unitary on $A$. Hence \eqref{eq:prelim-uhlmann} gives, for every $f$ and $z$,
\[
 \bigl\|\sqrt{\rho^f_z}\sqrt{(\rho^\circ)^f_z}\bigr\|_1\ge\bigl|\langle\Psi_z(f),\widehat f(z)\Psi(\chi_z)\rangle\bigr|=\bigl|\langle\Psi_z(f),G^\circ_z(f)\rangle\bigr| ,
\]
where the equality holds because $\Pi_z$ is an orthogonal projection and $\Psi_z(f)=\Pi_z\Psi(f)$. Since
\[
 \Psi_z(f)=\sum_w\widehat f(w)\Pi_z\Psi(\chi_w)\qquad\text{and}\qquad\E_f\widehat f(w)\widehat f(z)=\mathbf 1[w=z]/N,
\]
we have
\begin{align*}
 \E_f\sum_z\langle\Psi_z(f),G^\circ_z(f)\rangle&=\frac1N\sum_z\|g_z\|^2=\sigma,\\
 F(\rho,\rho^\circ)&\ge\E_f\sum_z|\langle\Psi_z(f),G^\circ_z(f)\rangle|\\
 &\ge\sigma=1-\delta/2,\\
 P(\rho,\rho^\circ)&\le\sqrt{1-(1-\delta/2)^2}=\eta_\delta ,
\end{align*}
where the second line follows by summing the first display of this step over $z$ and averaging over $f$, the third uses the first line, the triangle inequality, and Step~1, and the last uses the definition of $P$.

\emph{Step 3: the min-entropy of the ideal family.} In the family $\rho^\circ$, Eve's state does not depend on $z$, so for every measurement
\[
 \sum_z\operatorname{Tr}(M_zp_f(z)\rho_E)\le\max_zp_f(z),
\]
with equality for the measurement that always outputs a heaviest $z$. Hence, by \cref{lem:fr-max-weight},
\[
 p_{\mathrm{guess}}(Z\mid RE)_{\rho^\circ}=\E_f\max_zp_f(z)\le2(n+2)/N ,
\]
and $\Hmin(Z\mid RE)_{\rho^\circ}\ge n-\log(2(n+2))$. The claim follows from \cref{def:entropies}, since $\rho^\circ$ is a normalized family.
\end{proof}

The case $q=1$ of \cref{thm:fourier-randomness}, specialized as above, gives $\sqrt{\delta/(2(1-1/N))}$ as the smoothing parameter. For large $N$ and small positive deficit $\delta$, this is smaller than the ideal-family parameter by a factor approaching $\sqrt2$. The exact ratio is $[2(1-1/N)(1-\delta/4)]^{-1/2}$. It certifies closeness to a comparison family rather than to the ideal family itself.

\subsection{Haar-random states}\label{sec:randomness-haar}

We use the notation of \cref{sec:haar}: vector-valued polynomials, patterns, the level operator $C$, the progress measure $\lambda$, the operator $X$ of the score, and the bound $\lambda(\Psi)\le\Lambda_q$ of \cref{thm:haar-lower}, where $\Lambda_q:=\Lambda_{0,q}$. The oracle is $O_\phi$, the device may make controlled queries, and Eve learns $\phi$.

\begin{theorem}[Certified randomness from Haar-random states]\label{thm:haar-randomness}
Let $q\ge1$ and $\delta\ge0$, and consider an experiment with the canonical oracle in which the device makes at most $q$ queries and has score $B(\mathcal A)\ge\beta_N-\delta$. Let $\eta_0$ and $\Lambda_q$ be given by
\begin{equation}\label{eq:hr-eta0}
\begin{aligned}
 \eta_0&=\bigl(3(\delta+2\sqrt2\,\Lambda_q+72\,\Lambda_q^2)+\Lambda_q^2\bigr)^{1/2},\\
 \Lambda_q&=\frac{c_0}{\sqrt N}\sum_{t=0}^{q-1}\sqrt{2t+1}\le\sqrt3\,c_0\,\frac{q^{3/2}}{\sqrt N} .
\end{aligned}
\end{equation}
If $\eta_0<1$, then for every $\eta\in[\eta_0,1]$,
\begin{equation}\label{eq:hr-main-bound}
 \Hmin^\eta(Z\mid RE)_\rho\ge n-2\log q-\log H_N-2-\log\frac1{1-\eta_0^2},
\end{equation}
where $H_N\le\ln N+1$ is the harmonic number of \cref{lem:heaviest}. Moreover, for every $q$,
\begin{equation}\label{eq:hr-eta0-simple}
 \eta_0\le\sqrt{3\delta}+4\sqrt{c_0}\,q^{3/4}N^{-1/4}+26\,c_0\,q^{3/2}N^{-1/2} .
\end{equation}
\end{theorem}

The three steps are as in the Fourier case, with the phase $\theta_z$ of the amplitude $\alpha_z$ in the role of the Fourier coefficient. For $z\in[N]$ and real $t$, let $\Theta_{z,t}$ be the map on the sphere that multiplies the $z$-th coordinate by $e^{\mathrm it}$, and for a vector-valued polynomial $F$ define the \emph{phase average}
\[
 (Q_zF)(\phi)=\frac1{2\pi}\int_0^{2\pi}F(\Theta_{z,t}\phi)\,dt .
\]

\begin{lemma}[Phase averaging]\label{lem:hr-phase-average}
The operator $Q_z$ is the orthogonal projection onto the sum of the blocks whose pattern satisfies $\gamma_z=0$. It commutes with every $P_k$, with $C$, with every $\phi$-independent operator on $\mathcal V$, and with $Q_{z'}$ for every $z'$.
\end{lemma}

\begin{proof}
Since $\Theta_{z,t}$ replaces $\alpha_z$ by $e^{\mathrm it}\alpha_z$ and leaves $p$ unchanged, we have $u_\gamma(\Theta_{z,t}\phi)=e^{\mathrm i\gamma_zt}u_\gamma(\phi)$, so
\[
 Q_z(u_\gamma h\otimes v)=\mathbf 1[\gamma_z=0]\,u_\gamma h\otimes v\qquad\text{for polynomials $h(p)$ and $v\in\mathcal V$}.
\]
By \cref{lem:haar-blocks}, $Q_z$ acts on each block as $0$ or $1$, which proves the first claim, and the operators in the second claim all map each block to itself.
\end{proof}

\begin{lemma}[The phase-independent part of the branches]\label{lem:hr-phase-mass}
Let $F\in\mathcal P_{\le d}(\mathcal H)$ be a unit vector with $\langle F,XF\rangle\ge\beta_N-\delta$, and let $F_z=\Pi_zF$. Then
\begin{align*}
 \sum_z\|Q_zF_z\|^2&\le\frac{\delta+2\sqrt2\lambda(F)+72\lambda(F)^2}{\beta_N-1}+\lambda(F)^2\\
 &\le3\bigl(\delta+2\sqrt2\lambda(F)+72\lambda(F)^2\bigr)+\lambda(F)^2 .
\end{align*}
\end{lemma}

\begin{proof}
\emph{Step 1: the positive-level part.} Write $F=F_0+F_+$ with $F_0=P_0F$. Since $Q_z\Pi_z$ commutes with $P_0$ (\cref{lem:hr-phase-average}), we have $Q_zF_z=Q_z\Pi_zF_0+Q_z\Pi_zF_+$ with orthogonal summands, and, by \cref{lem:haar-cross},
\begin{align*}
 \sum_z\|Q_z\Pi_zF_+\|^2&\le\sum_z\|\Pi_zF_+\|^2=\|F_+\|^2\\
 &\le\lambda(F)^2 .
\end{align*}

\emph{Step 2: the level-zero part.} Let $w_0=\sum_z\|Q_z\Pi_zF_0\|^2$ be the weight of the level-zero terms whose pattern has $\gamma_z=0$ at their own output $z$. By the exact ratio computed in \cref{prop:haar-level-zero}, a level-zero term with $\gamma_z=0$ has ratio $N/(N+|\gamma|_1)\le1$ on its own branch, and a term with $|\gamma_z|=1$ has ratio at most $\beta_N$, so
\begin{align*}
 \langle F_0,XF_0\rangle&\le w_0+\beta_N\bigl(\|F_0\|^2-w_0\bigr)\\
 &=\beta_N\|F_0\|^2-(\beta_N-1)w_0 .
\end{align*}

\emph{Step 3: solving for $w_0$.} By \cref{lem:haar-cross,lem:haar-global}, using $\|F_0\|\le1$,
\begin{align*}
 \beta_N-\delta&\le\langle F,XF\rangle\\
 &\le\beta_N-(\beta_N-1)w_0+2\sqrt2\lambda(F)+72\lambda(F)^2 .
\end{align*}
Solving for $w_0$ and adding the positive-level part gives the first inequality, and the second follows from $\beta_N-1=(N-1)/(N+1)\ge1/3$ for $N\ge2$.
\end{proof}

Thus a level-zero term that does not depend on the phase of $\alpha_z$ scores at most $1$ on the branch $z$, which is less than the ideal score by a constant, so a near-ideal score forces the branches to be almost entirely phase-dependent.

\begin{lemma}[Changing one phase of the oracle]\label{lem:hr-perturbation}
For every $\phi$, $z$, and $t$,
\[
 \|O_\phi-O_{\Theta_{z,t}\phi}\|=2\lvert\sin(t/2)\rvert\sqrt{p_z(2-p_z)}\le2\sqrt2\,\lvert\sin(t/2)\rvert\sqrt{p_z},
\]
and the same holds for the controlled queries.
\end{lemma}

\begin{proof}
For unit vectors $u,v$, the operator $\ket u\bra u-\ket v\bra v$ vanishes on the orthogonal complement of their span, and on the span it has trace $0$ and determinant $-(1-|\langle u,v\rangle|^2)$, so its norm is $\sqrt{1-|\langle u,v\rangle|^2}$.

\emph{Step 1: the overlap of the two oracle vectors.} Let $\phi'=\Theta_{z,t}\phi$, and write $u=\eta_\phi$ and $v=\eta_{\phi'}$ for the vectors of \eqref{eq:prelim-canonical-oracle}, so that $O_\phi-O_{\phi'}=-2\bigl(\ket u\bra u-\ket v\bra v\bigr)$. Since $\phi'$ differs from $\phi$ only in the phase of the $z$-th coordinate, we have, with $w=(p_z/2)(1-e^{\mathrm it})$,
\begin{align*}
 \langle\phi,\phi'\rangle&=1-p_z+p_ze^{\mathrm it},\\
 \langle u,v\rangle&=\frac12\bigl(1+\langle\phi,\phi'\rangle\bigr)=1-w,
\end{align*}
and therefore
\begin{align*}
 1-|\langle u,v\rangle|^2&=2\operatorname{Re}w-|w|^2=p_z(1-\cos t)-\frac{p_z^2}2(1-\cos t)\\
 &=\sin^2(t/2)\,p_z(2-p_z) .
\end{align*}

\emph{Step 2: the norm and the controlled queries.} Taking the square root and multiplying by $2$ gives the norm of $O_\phi-O_{\phi'}$. For the controlled queries,
\[
 O^{\mathrm c}_\phi-O^{\mathrm c}_{\phi'}=\ket1\bra1\otimes(O_\phi-O_{\phi'}) . \qedhere
\]
\end{proof}

\begin{proof}[Proof of \cref{thm:haar-randomness}]
\emph{Step 1: the auxiliary vectors.} By \cref{lem:haar-degree}, the final state $\Psi$ is a unit vector in $\mathcal P_{\le2q}(\mathcal H)$, its score is $B(\mathcal A)=\langle\Psi,X\Psi\rangle$ by \eqref{eq:haar-score-form}, and $\lambda(\Psi)\le\Lambda_q$ by \cref{thm:haar-lower} with $s=0$. As auxiliary vectors we take $G_z=(I-Q_z)\Psi_z$, where $\Psi_z=\Pi_z\Psi$. Since the right side of \cref{lem:hr-phase-mass} increases with $\lambda(F)$,
\begin{align*}
 \sum_z\|\Psi_z-G_z\|^2&=\sum_z\|Q_z\Psi_z\|^2\\
 &\le3(\delta+2\sqrt2\Lambda_q+72\Lambda_q^2)+\Lambda_q^2=\eta_0^2 .
\end{align*}

\emph{Step 2: the operators $K_z$ and their norms.} Pointwise, since $\Pi_z$ does not depend on $\phi$,
\begin{align*}
 G_z(\phi)&=\Pi_z\Psi(\phi)-\frac1{2\pi}\int_0^{2\pi}\Pi_z\Psi(\Theta_{z,t}\phi)\,dt=(K_z(\phi)\otimes I_E)\psi_0,\\
 K_z(\phi)&=\frac{\Pi_z}{2\pi}\int_0^{2\pi}\bigl(V_\phi-V_{\Theta_{z,t}\phi}\bigr)\,dt .
\end{align*}

For products of unitaries,
\[
 W_q\cdots W_1-W'_q\cdots W'_1=\sum_jW_q\cdots W_{j+1}(W_j-W'_j)W'_{j-1}\cdots W'_1 ,
\]
and the unitaries $U_j$ do not depend on $\phi$, so \cref{lem:hr-perturbation} gives
\[
 \|V_\phi-V_{\Theta_{z,t}\phi}\|\le2q\lvert\sin(t/2)\rvert\sqrt{p_z(2-p_z)} .
\]
Averaging over $t$, with $(2\pi)^{-1}\int_0^{2\pi}2\lvert\sin(t/2)\rvert\,dt=4/\pi$, we obtain
\[
 \|K_z(\phi)\|\le\frac4\pi\,q\sqrt{p_z(2-p_z)}\le\frac{4\sqrt2}\pi\,q\sqrt{p_z},
\]
and hence $\|K_z(\phi)\|^2\le(32/\pi^2)\,q^2p_z$. This hybrid argument plays the role that Markov's inequality for polynomials played in \cref{lem:fr-derivative}, as explained in \cref{sec:overview-haar-randomness}.

\emph{Step 3: the entropy bound.} The hypotheses \eqref{eq:assembly-hypotheses} therefore hold with $\kappa=32q^2/\pi^2\le4q^2$, and $K_z$ depends continuously on $\phi$. By \cref{lem:hr-phase-average}, $I-Q_z$ is an orthogonal projection of the space of vector-valued polynomials, so $\E_\phi\langle\Psi_z(\phi),G_z(\phi)\rangle=\E_\phi\|G_z(\phi)\|^2$ for every $z$. This identity holds on average over $\phi$ and not pointwise, because $Q_z$ averages over the phase of $\alpha_z$, but the average is what part~(iv) of \cref{prop:randomness-assembly} requires, so that part applies with $\xi\ge1-\eta_0^2$, and together with \cref{lem:heaviest} it gives \eqref{eq:hr-main-bound}.

\emph{Step 4: the parameters.} Since $x\mapsto\sqrt{2x+1}$ is increasing,
\begin{align*}
 \sum_{t=0}^{q-1}\sqrt{2t+1}&\le\int_0^q\sqrt{2x+1}\,dx\le\frac13(2q+1)^{3/2}\\
 &\le\sqrt3q^{3/2}
\end{align*}
for $q\ge1$, which gives the bound on $\Lambda_q$. Hence $\eta_0^2=3\delta+6\sqrt2\Lambda_q+217\Lambda_q^2$ with
\begin{align*}
 6\sqrt2\Lambda_q&\le6\sqrt6\,c_0q^{3/2}N^{-1/2}\le16c_0q^{3/2}N^{-1/2},\\
 217\Lambda_q^2&\le651c_0^2q^3/N\le(26c_0)^2q^3/N,
\end{align*}
and $\sqrt{x+y+z}\le\sqrt x+\sqrt y+\sqrt z$ gives \eqref{eq:hr-eta0-simple}.
\end{proof}

\subsection{Tightness of the parameters}\label{sec:randomness-smoothing}

The first device shows that the term $\sqrt{\delta/(2(1-1/N))}$ in \eqref{eq:fr-eta0} is of the right order, including its constant, whenever many bits are to be certified. The second shows that the losses $4\log q$ in \cref{thm:fourier-randomness} and $2\log q$ in \cref{thm:haar-randomness} cannot be brought below $2\log q-O(1)$: it uses amplitude amplification to raise the probability that one ideal sample lies in a fixed set of about $N/q^2$ strings, which costs $q$ queries and only $O(q^{-2}+q^4N^{-1/2})$ in score.

\begin{proposition}[Mixing a fixed output with honest sampling]\label{prop:randomness-smoothing}
Let $0<\mu<1$ and $z_0\in\{0,1\}^n$, and consider the one-query device in the Fourier ensemble that, using a private coin, outputs $z_0$ with probability $\mu$ and otherwise runs the honest sampler, with Eve's register trivial. Then the following hold.
\begin{enumerate}
\item[(i)] The score of the device is $B(\mathcal A)=3-2/N-2(1-1/N)\mu$, so \cref{thm:fourier-randomness} applies with $q=1$, $\delta=2(1-1/N)\mu$, and $\eta_0=\sqrt\mu$.
\item[(ii)] The conditional min-entropy satisfies $\Hmin(Z\mid R)_\rho\le\log(1/\mu)$.
\item[(iii)] Every normalized family $\tau$ with $\Hmin(Z\mid R)_\tau\ge h$ satisfies
\[
 P(\rho,\tau)\ge\Bigl(1-\min\bigl\{1,\sqrt{1-\mu}+2^{-h/2}\bigr\}^2\Bigr)^{1/2},
\]
which tends to $\sqrt\mu=\eta_0$ as $h\to\infty$.
\end{enumerate}
\end{proposition}

\begin{proof}
\emph{Step 1: the score and $\eta_0$.} For a fixed $z_0$, we have $\E_fNp_f(z_0)=1$, and the honest sampler scores $3-2/N$, so $B(\mathcal A)=(1-\mu)(3-2/N)+\mu$, which is the stated value. With $q=1$ we have $\bar\lambda_1=0$, and \eqref{eq:fr-eta0} gives
\[
 \eta_0=\sqrt{\delta/(2(1-1/N))}=\sqrt\mu .
\]

\emph{Step 2: the min-entropy of $\rho$.} Since $E$ is trivial, a family of conditional states is a function $\tau(f,z)\ge0$ with $\E_f\sum_z\tau(f,z)=1$, the guessing probability is $\E_f\max_z\tau(f,z)$, and
\[
 F(\rho,\tau)=\E_f\sum_z\sqrt{\rho(f,z)\tau(f,z)} .
\]
The actual family is $\rho(f,z)=(1-\mu)p_f(z)+\mu\mathbf 1[z=z_0]$, whose guessing probability is at least $\mu$; hence $\Hmin(Z\mid R)_\rho\le\log(1/\mu)$.

\emph{Step 3: the bound on $P(\rho,\tau)$.} Now let $\E_f\max_z\tau(f,z)\le2^{-h}$. We split the sum over $z$ according to whether $z=z_0$ and apply the Cauchy--Schwarz inequality to each part, with $\E_f\rho(f,z_0)\le1$ and $\E_f\sum_z\tau(f,z)=1$. This gives
\begin{align*}
 \E_f\sqrt{\rho(f,z_0)\tau(f,z_0)}&\le\bigl(\E_f\tau(f,z_0)\bigr)^{1/2}\le\bigl(\E_f\max_z\tau(f,z)\bigr)^{1/2}\\
 &\le2^{-h/2}
\end{align*}
and
\[
 \E_f\sum_{z\ne z_0}\sqrt{\rho(f,z)\tau(f,z)}\le\bigl(1-\E_f\rho(f,z_0)\bigr)^{1/2} .
\]

Finally, $\E_f\rho(f,z_0)=\mu+(1-\mu)/N\ge\mu$, so
\[
 F(\rho,\tau)\le\min\bigl\{1,\sqrt{1-\mu}+2^{-h/2}\bigr\},
\]
which gives the stated bound on $P(\rho,\tau)$.
\end{proof}

The second device uses amplitude amplification to raise the probability that one ideal sample lies in a fixed set of outcomes, as the squaring algorithm does for its stored set. Let $\mathcal T$ be the set of strings whose first $k$ bits are zero, so that $|\mathcal T|=N'$ with $N'=N/2^k$, and let $\Pi_{\mathcal T}=\sum_{z\in\mathcal T}\ket z\bra z$. For an oracle $\omega$, write $p_\omega(\mathcal T)=\sum_{z\in\mathcal T}p_\omega(z)$ and $P_\star=N'/N=2^{-k}$, where $P_\star$ is the mean of $p_\omega(\mathcal T)$ in both ensembles.

\begin{proposition}[Amplitude amplification onto a fixed set of outputs]\label{prop:randomness-amplification}
Let $k\ge4$ be an integer with $2^{k+10}\le N$, and consider the device that prepares the state $\ket{\psi_\omega}$ of \cref{sec:squaring-access} with one query, applies $\ell$ times the operator $(2\ket{\psi_\omega}\bra{\psi_\omega}-I)(I-2\Pi_{\mathcal T})$, where $\ell$ is chosen as in \cref{def:squaring} for the value $P_\star=2^{-k}$, and measures; Eve's register is trivial. Then the following hold.
\begin{enumerate}
\item[(i)] The device makes $q=L=2\ell+1$ queries, where $2^{k/2}\le q\le2^{k/2+1}$.
\item[(ii)] Its output lies in $\mathcal T$ with probability at least $3/4$, so that
\[
 \Hmin^\eta(Z\mid R)_\rho\le n-k+3\le n-2\log q+5\qquad\text{for every $0\le\eta\le1/2$} .
\]
\item[(iii)] Its expected score is at least the ideal score minus $C_{\mathrm{amp}}(q^{-2}+q^4N^{-1/2})$, in both ensembles, for an absolute constant $C_{\mathrm{amp}}$.
\item[(iv)] In particular, for $q\to\infty$ with $q\le N^{1/9}$ and $N$ large, the device scores within $o(1)$ of the ideal score and its output has at most $n-2\log q+5$ bits of smooth min-entropy for every smoothing parameter $0\le\eta\le1/2$.
\end{enumerate}
\end{proposition}

\begin{proof}
\emph{Step 1: the number of queries and the output distribution.} Fix $\omega$, write $p=p_\omega$ and $P=p(\mathcal T)$, and put $Y=P/P_\star$, so that $\E_\omega Y=1$. By \cref{def:squaring}, $L=2\ell+1$ satisfies $|L\vartheta_0-\pi/2|\le\vartheta_0$, where $\vartheta_0=\arcsin\sqrt{P_\star}=\arcsin2^{-k/2}$, and, as in the proof of \cref{lem:squaring-amplification-error}, $2^{-k/2}\le\vartheta_0\le(\pi/3)2^{-k/2}$. Hence
\[
 \frac32\,2^{k/2}-1\le q\le1+\frac\pi2\,2^{k/2},
\]
and both stated bounds on $q$ follow for $k\ge4$.

If $P=0$, the marked projection annihilates $\ket{\psi_\omega}$, every amplification step fixes this state, and $g(0)=0$. For $P>0$, by \cref{lem:amplification} with $\Pi=\Pi_{\mathcal T}$, as in the proof of \cref{lem:squaring-fixed}, the output has the distribution
\[
 g(P)\,\mu_{\mathcal T}+(1-g(P))\,\mu_{\mathcal T^{\mathrm c}}\qquad\text{with}\qquad g(P)=\sin^2(L\arcsin\sqrt P),
\]
where the first term is absent if $P=0$ and the second if $P=1$. In particular, $\Pr[Z\in\mathcal T\mid\omega]=g(P)$, and \cref{lem:squaring-amplification-error}, which applies because $P_\star\le1/4$, gives
\begin{equation}\label{eq:amp-failure}
 1-g(P)\le18\,(Y-1)^2+3P_\star ,
\end{equation}
also when $P=0$, where the left side is $1$.

\emph{Step 2: the moments of $Y$.} In the Haar ensemble, $P$ has the Beta distribution with parameters $(N',N-N')$ by \cref{lem:dirichlet}(iii), so
\[
 \Var(P)=N'(N-N')/(N^2(N+1))\le N'/N^2
\]
by \cref{lem:dirichlet}(i) applied to the two-component vector $(P,1-P)$, and $\Var(Y)\le1/N'=2^k/N$.

In the Fourier ensemble, let $\mathcal T^\perp$ be the set of the $2^k$ strings whose last $n-k$ bits are zero, and let $\mathcal C$ run over the $N'$ cosets of $\mathcal T^\perp$ in $\{0,1\}^n$. With $X_z=\sqrt N\widehat f(z)$ as in the proof of \cref{lem:moments},
\begin{equation}\label{eq:amp-cosets}
\begin{aligned}
 Y&=\frac1{N'}\sum_{z\in\mathcal T}X_z^2=\frac1N\sum_{\mathcal C}\Bigl(\sum_{x\in\mathcal C}f(x)\Bigr)^2\\
 &=\frac1{N'}\sum_{\mathcal C}Y_{\mathcal C},\qquad Y_{\mathcal C}=\frac1{2^k}\Bigl(\sum_{x\in\mathcal C}f(x)\Bigr)^2 ,
\end{aligned}
\end{equation}
where the second equality follows from
\begin{align*}
 X_z^2&=\frac1N\sum_{x,y}\chi_z(x\oplus y)f(x)f(y),\\
 \sum_{z\in\mathcal T}\chi_z(x\oplus y)&=N'\mathbf 1[x\oplus y\in\mathcal T^\perp],
\end{align*}
since $x\oplus y\in\mathcal T^\perp$ exactly when $x$ and $y$ lie in the same coset.

The $Y_{\mathcal C}$ are independent, because distinct cosets use disjoint sets of values of $f$, and each is the normalized square of a sum of $2^k$ independent uniform signs, so $\E Y_{\mathcal C}=1$, $\E Y_{\mathcal C}^2=3-2/2^k\le3$, and $\E Y_{\mathcal C}^4\le105$ by the pairing bound in the proof of \cref{lem:moments}(ii). Hence
\begin{align*}
 \Var(Y)&=\frac1{N'^2}\sum_{\mathcal C}\Var(Y_{\mathcal C})\le2/N'\\
 &=2\cdot2^k/N ,
\end{align*}
and, since the summands $Y_{\mathcal C}-1$ are independent with mean zero,
\begin{align*}
 \E(Y-1)^4&=\frac1{N'^4}\Bigl(\sum_{\mathcal C}\E(Y_{\mathcal C}-1)^4+3\sum_{\mathcal C\ne\mathcal C'}\Var(Y_{\mathcal C})\Var(Y_{\mathcal C'})\Bigr)\\
 &\le\frac{8\cdot106}{N'^3}+\frac{12}{N'^2}\le\frac{900}{N'^2}\\
 &=\frac{900\cdot2^{2k}}{N^2},
\end{align*}
using $(y-1)^4\le8(y^4+1)$. In both ensembles, therefore, $\Var(Y)\le2\cdot2^k/N$.

\emph{Step 3: the min-entropy.} By \eqref{eq:amp-failure} and the variance bound,
\begin{align*}
 \varepsilon_{\mathcal T}:=\Pr[Z\notin\mathcal T]&=\E_\omega[1-g(P)]\\
 &\le18\Var(Y)+3P_\star\\
 &\le\frac{36\cdot2^k}N+\frac3{2^k}\le\frac14 ,
\end{align*}
using $2^{k+10}\le N$ and $k\ge4$.

Let $\tau$ be a normalized family within purified distance $\eta$ of $\rho$, where $\eta\le1/2$, and write $m_{\mathcal T}$ for the weight $\E_\omega\sum_{z\in\mathcal T}\tau(\omega,z)$ of $\tau$ on $\mathcal T$. Since $E$ is trivial,
\[
 F(\rho,\tau)=\E_\omega\sum_z\sqrt{\rho(\omega,z)\tau(\omega,z)},
\]
and the Cauchy--Schwarz inequality, applied separately to the pairs $(\omega,z)$ with $z\in\mathcal T$ and with $z\notin\mathcal T$, gives
\begin{align*}
 \sqrt{1-\eta^2}&\le F(\rho,\tau)\le\sqrt{m_{\mathcal T}}+\sqrt{\varepsilon_{\mathcal T}}\\
 &\le\sqrt{m_{\mathcal T}}+\frac12 ,
\end{align*}
so $m_{\mathcal T}\ge(\sqrt3/2-1/2)^2\ge1/8$. Since
\[
 \max_z\tau(\omega,z)\ge\frac1{N'}\sum_{z\in\mathcal T}\tau(\omega,z)\qquad\text{for every }\omega,
\]
the guessing probability of $\tau$ is at least $m_{\mathcal T}/N'\ge1/(8N')$, so $\Hmin(Z\mid R)_\tau\le n-k+3$, and $n-k+3\le n-2\log q+5$ because $k\ge2\log q-2$.

\emph{Step 4: the score in the Haar ensemble.} By \cref{lem:dirichlet}(iii), the vector $\tilde p=(p_z/P)_{z\in\mathcal T}$ is independent of $P$ and uniform on the simplex of dimension $N'-1$, so $\E\sum_{z\in\mathcal T}\tilde p_z^2=2/(N'+1)$ by \cref{lem:dirichlet}(i). The score of the device is at least the contribution of the outputs in $\mathcal T$, which is
\begin{align*}
 \E_\phi\Bigl[g(P)\,N\sum_{z\in\mathcal T}\frac{p_z^2}P\Bigr]&=N\,\E_\phi\bigl[g(P)P\bigr]\,\E\Bigl[\sum_{z\in\mathcal T}\tilde p_z^2\Bigr]\\
 &=\frac{2N}{N'+1}\,\E_\phi\bigl[g(P)P\bigr]\\
 &=\frac{2N'}{N'+1}\Bigl(1-\E_\phi\bigl[(1-g(P))Y\bigr]\Bigr),
\end{align*}
where we used $P=P_\star Y$, $NP_\star=N'$, and $\E Y=1$.

By \eqref{eq:amp-failure}, the bound $Y\le1/P_\star=2^k$, and $\Var(Y)\le2^k/N$, we obtain
\begin{align*}
 \E_\phi\bigl[(1-g(P))Y\bigr]&\le18\,\E_\phi\bigl[(Y-1)^2Y\bigr]+3P_\star\,\E_\phi Y\\
 &\le18\cdot2^k\Var(Y)+\frac3{2^k}\le\frac{18\cdot2^{2k}}N+\frac3{2^k} .
\end{align*}

Since $2N'/(N'+1)\le2$ and
\[
 2N'/(N'+1)=\beta_N-2/(N'+1)+2/(N+1)\ge\beta_N-2\cdot2^k/N,
\]
the score is at least
\[
 \beta_N-2\cdot2^k/N-36\cdot2^{2k}/N-6\cdot2^{-k}\ge\beta_N-38\,q^4/N-24\,q^{-2},
\]
using $2^k\le q^2$ and $2^{-k}\le4q^{-2}$.

\emph{Step 5: the score in the Fourier ensemble.} Write
\[
 \bar X_4=\frac1{N'}\sum_{z\in\mathcal T}X_z^4,\qquad\text{so that}\qquad N\sum_{z\in\mathcal T}p_z^2/P=\bar X_4/Y,
\]
with the value $0$ when $Y=0$. Since $1/y\ge2-y$ for $y>0$, we have $\bar X_4/Y\ge\bar X_4(2-Y)$, also when $Y=0$, where both sides vanish. The score is therefore at least
\[
 \E_f\bigl[g(P)\,\bar X_4(2-Y)\bigr]\ge\E_f\bigl[\bar X_4(2-Y)\bigr]-2\,\E_f\bigl[(1-g(P))\bar X_4\bigr],
\]
because $(1-g(P))\bar X_4(2-Y)\le2(1-g(P))\bar X_4$.

Here $\E_f\bar X_4=3-2/N$, and
\[
 \E_f\bar X_4^2\le\frac1{N'}\sum_{z\in\mathcal T}\E_fX_z^8\le105
\]
by \cref{lem:moments}(ii), so, using $\E_fY=1$,
\begin{align*}
 \E_f\bigl[\bar X_4(2-Y)\bigr]&=\E_f\bar X_4-\operatorname{Cov}(\bar X_4,Y)\\
 &\ge3-\frac2N-\sqrt{105\Var(Y)}\ge3-\frac2N-15\,\frac{2^{k/2}}{\sqrt N},
\end{align*}
and by \eqref{eq:amp-failure} and the Cauchy--Schwarz inequality,
\begin{align*}
 \E_f\bigl[(1-g(P))\bar X_4\bigr]&\le18\sqrt{\E_f(Y-1)^4\,\E_f\bar X_4^2}+3P_\star\,\E_f\bar X_4\\
 &\le18\cdot30\sqrt{105}\,\frac{2^k}N+\frac9{2^k}\le\frac{5600\cdot2^k}N+\frac9{2^k} .
\end{align*}

Hence the score is at least
\[
 3-2/N-15\cdot2^{k/2}N^{-1/2}-11200\cdot2^kN^{-1}-18\cdot2^{-k},
\]
which, with $2^{k/2}\le q$, $2^k\le q^2$, and $2^{-k}\le4q^{-2}$, is at least
\[
 3-2/N-15qN^{-1/2}-11200q^2N^{-1}-72q^{-2} .
\]

\emph{Step 6: the deficit in both ensembles, and part (iv).} In both ensembles the deficit is at most $C_{\mathrm{amp}}(q^{-2}+q^4N^{-1/2})$ for an absolute constant $C_{\mathrm{amp}}$, because
\[
 q^2N^{-1}\le qN^{-1/2}\le q^4N^{-1/2}\qquad\text{for }q\le\sqrt N ,
\]
where $q\le\sqrt N$ holds because $q\le2^{k/2+1}\le\sqrt N/16$.
For $q\to\infty$ with $q\le N^{1/9}$, this deficit tends to zero. The upper score bounds of \cref{thm:fourier-lower,thm:haar-lower}, with $s=0$ in the Haar case, also give an excess above the ideal score of at most $O(q^{3/2}N^{-1/2}+q^3/N)=o(1)$. Thus the score is within $o(1)$ of the ideal score on both sides. Part (iv) follows, since $2^{k+10}\le N$ holds for $q\le N^{1/9}$ and $N$ large, as $2^k\le q^2$.
\end{proof}

\begin{remark}[Standard smooth min-entropy]
\Cref{prop:randomness-smoothing,prop:randomness-amplification} also bound the standard smooth min-entropy of \cref{rem:entropy-conventions}, in the Haar ensemble for every finite discretization of $\phi$, since the arguments do not use the value of $R$.

The pinching map that measures $R$ and $Z$ in their bases leaves $\rho$ fixed and does not increase the purified distance. It also does not decrease $\Hmin(Z\mid R)$, because it is unital on $Z$ and a channel on $R$: if $\tau_{ZR}\le I_Z\otimes\sigma_R$, then the pinched state satisfies the same inequality with the pinched $\sigma_R$, which has the same trace. Hence the standard smoothing may be restricted to families that are classical on $RZ$ and have total mass at most $1$.

Since $\rho$ is normalized, the generalized fidelity of $\rho$ and such a family $\tau$ equals $F(\rho,\tau)$. The proofs of both propositions use only $F(\rho,\tau)$ together with the bound $\E_\omega\sum_z\tau(\omega,z)\le1$, so they apply to these families without change.
\end{remark}

\subsection{Proof of \texorpdfstring{\cref{thm:intro-randomness}}{Theorem 1.2}}\label{sec:randomness-proof}

\begin{proof}[Proof of \cref{thm:intro-randomness}]
Let $\eta_0$ be as in \cref{thm:fourier-randomness} for Fourier sampling and as in \cref{thm:haar-randomness} for Haar-random states, and put $\eta=\min\{1,\eta_0\}$ if $q^3\le N$ and $\eta=1$ otherwise.

\emph{Step 1: the bound on $\eta$.} If $q^3\le N$, then $q^{3/2}N^{-1/2}\le q^{3/4}N^{-1/4}$, so \eqref{eq:fr-eta0-simple} gives
\[
 \eta_0\le\sqrt\delta+5q^{3/4}N^{-1/4}
\]
and \eqref{eq:hr-eta0-simple} gives
\[
 \eta_0\le\sqrt{3\delta}+(4\sqrt{c_0}+26c_0)q^{3/4}N^{-1/4},
\]
while if $q^3>N$, then $q^{3/4}N^{-1/4}>1=\eta$. In both cases $\eta=O(\sqrt\delta+q^{3/4}N^{-1/4})$, as stated.

\emph{Step 2: the two entropy bounds.} The hypothesis $\eta\le1/2$ forces $q^3\le N$ and $\eta=\eta_0\le1/2$, so that $2\log(1/(1-\eta_0))\le2$ and $\log(1/(1-\eta_0^2))\le1$. Then \eqref{eq:fr-main-bound} gives
\[
 \Hmin^{\eta_0}(Z\mid RE)_\rho\ge n-4\log q-\log(n+2)-5\qquad\text{for Fourier sampling},
\]
and \eqref{eq:hr-main-bound} with $\log H_N\le\log(n+1)$ gives
\[
 \Hmin^{\eta_0}(Z\mid RE)_\rho\ge n-2\log q-\log(n+1)-3\qquad\text{for Haar-random states}.
\]
Finally, if $q\le n^{O(1)}$ and $\delta=o(1)$, then $\eta_0=o(1)$ in both cases, and $\eta_0$ is negligible when $\delta$ is.
\end{proof}

\begin{remark}[Scope of the theorems]\label{rem:randomness-scope}
The hypotheses of \cref{thm:fourier-randomness,thm:haar-randomness} concern the expected score of the experiment; estimating this score from finitely many runs, and conditioning on the acceptance of such a test, are not addressed here.

The theorems do not require $\delta\to0$: they give $n-O(\log n)$ bits whenever $q\le n^{O(1)}$ and $1-\eta_0\ge n^{-O(1)}$, and $(1-o(1))n$ bits for $q=2^{o(n)}$ when $1-\eta_0\ge2^{-o(n)}$. They are vacuous when $\eta_0\ge1$, in particular when $q\ge N^{1/3}+1$ in the Fourier case, because then
\[
 \bar\lambda_q\ge\frac43(q-1)^{3/2}N^{-1/2}\ge\frac43
\]
and $\eta_0\ge2\bar\lambda_q$; accordingly, the proof of \cref{thm:intro-randomness} sets $\eta=1$ when $q^3>N$.

The losses $4\log q$ and $2\log q$ come only from \cref{lem:fr-derivative} and from the hybrid argument in the proof of \cref{thm:haar-randomness}; by \cref{prop:randomness-amplification} neither loss can be brought below $2\log q-O(1)$, so the coefficient $2$ in the Haar case is optimal and the correct coefficient for Fourier sampling lies between $2$ and $4$. By \cref{prop:randomness-smoothing} the leading term of $\eta_0$ in the Fourier case cannot be improved by a constant factor if $\omega(1)$ bits are to be certified; we do not claim that the terms involving $\bar\lambda_q$ or $\Lambda_q$ are optimal.

Finally, for Haar-random states the result is for the canonical oracle; the obstacle for an Eve who learns a Haar-random unitary is explained at the end of \cref{sec:overview-haar-randomness}.
\end{remark}

\section{Quantum Hadamard products}\label{sec:hadamard}

The squaring algorithm measures its stored samples. Keeping them coherent and retaining the match index in the reversible computation gives a related algorithm that preserves the phases needed to prepare a quantum Hadamard product. For two unit vectors $\ket\alpha=\sum_z\alpha_z\ket z$ and $\ket\beta=\sum_z\beta_z\ket z$ in $\C^N$ whose coordinatewise product is not zero, the \emph{quantum Hadamard product} is the normalized state
\[
 \ket{\alpha\circ\beta}=\frac{\sum_z\alpha_z\beta_z\ket z}{\bigl(\sum_y|\alpha_y\beta_y|^2\bigr)^{1/2}} ,
\]
and preparing it is a task studied in~\cite{HolmesEtAl,HuangHadamard}.

We assume oracle access for state preparation and reflection, as in \cref{sec:squaring-access}: one query prepares each input state, and two queries reflect about it. This is the case for the canonical oracles $O_\alpha$ and $O_\beta$, and for unitaries that prepare the two states from $\ket{0^n}$, when
their inverses may also be applied.

The direct approach prepares $\ket\alpha\ket\beta$ and amplifies the component in which the two registers agree. When all amplitudes have magnitude $N^{-1/2}$, that component has weight $\sum_z|\alpha_z\beta_z|^2=1/N$, so the amplification takes about $\sqrt N$ reflections, each of which costs queries to both oracles. Our algorithm instead prepares $m$ coherent copies of $\ket\alpha$ and amplifies a copy of $\ket\beta$ onto their stored labels. It uses $m+1+O(\sqrt{N/m})$ queries, and reversible deletion retains the product phases. Taking $m=\lceil N^{1/3}\rceil$ gives $O(N^{1/3})$ queries and trace-distance error $O(N^{-1/6})$, as proved below.

\subsection[Hadamard products of states with uniform magnitudes]{Hadamard products of states with uniform magnitudes}\label{sec:hadamard-uniform}

\begin{theorem}[Hadamard products of states with uniform magnitudes]\label{thm:hadamard}
Let $\ket\alpha,\ket\beta\in\C^N$ satisfy $|\alpha_z|=|\beta_z|=N^{-1/2}$ for all $z$, so that $\ket{\alpha\circ\beta}=\sqrt N\sum_z\alpha_z\beta_z\ket z$, and let $1\le m\le N/4$. There is an algorithm with oracle access for state preparation and reflection to $\ket\alpha$ and $\ket\beta$ that makes $m$ preparations of $\ket\alpha$, one preparation of $\ket\beta$, and at most $\pi\sqrt{N/m}/2$ further queries, and outputs a state $\rho$ on $\C^N$ with
\[
 \frac12\bigl\|\rho-\ket{\alpha\circ\beta}\bra{\alpha\circ\beta}\,\bigr\|_1\le\frac\pi3\sqrt{m/N}+\sqrt2\,\frac m{\sqrt N} .
\]
For $N\ge8$ and $m=\lceil N^{1/3}\rceil$, the algorithm makes $O(N^{1/3})$ queries and the error is $O(N^{-1/6})$.
\end{theorem}

\begin{proof}
\emph{The algorithm.} Let $\mathsf S_1,\dots,\mathsf S_m$ and $\mathsf Y$ be registers of the type on which the oracles act, and let $\mathsf J$ be a register with basis $\ket0,\ket1,\dots,\ket m$, initially in $\ket0$. Write $P_\star=m/N$ and $\vartheta_0=\arcsin\sqrt{P_\star}$, let $\ell$ be a nonnegative integer nearest to $\pi/(4\vartheta_0)-1/2$ and $L=2\ell+1$, as in \cref{def:squaring}, so that
\begin{align*}
 |L\vartheta_0-\pi/2|&\le\vartheta_0,\\
 2\ell&\le\pi/(2\vartheta_0)\le\pi\sqrt{N/m}/2 ,
\end{align*}
where the second line follows from the first and from $\vartheta_0\ge\sqrt{P_\star}$. The algorithm consists of the following four steps.
\begin{enumerate}
\item Prepare $\ket\alpha$ in each of $\mathsf S_1,\dots,\mathsf S_m$ and $\ket\beta$ in $\mathsf Y$.
\item Let $\Pi$ be the projection onto the span of the basis states $\ket{x_1,\dots,x_m}\ket y$ with $y\in\{x_1,\dots,x_m\}$, which is implemented by a reversible comparison of $\mathsf Y$ with each $\mathsf S_i$ and uses no query. Apply $\ell$ times the operator $(2\ket\beta\bra\beta_{\mathsf Y}-I)(I-2\Pi)$.
\item Compute into $\mathsf J$ the smallest index $j$ with $x_j=y$, or $0$ if there is none. For $j\ge1$, swap $\mathsf S_j$ with $\mathsf S_m$, controlled by $\mathsf J$; for $j=0$, perform no swap. Then add the contents of $\mathsf Y$ bitwise into $\mathsf S_m$, so that $\ket y_{\mathsf S_m}\ket y_{\mathsf Y}$ becomes $\ket{0^n}_{\mathsf S_m}\ket y_{\mathsf Y}$. All of this is a fixed reversible computation on the basis states and uses no query.
\item Output the register $\mathsf Y$.
\end{enumerate}
The query count is $m+1+2\ell$, which is as claimed.

\emph{The state after amplification.} Write $x=(x_1,\dots,x_m)$, $\alpha_x=\alpha_{x_1}\cdots\alpha_{x_m}$, and let $S(x)$ be the set $\{x_1,\dots,x_m\}$ of entries of $x$. After step~1, the state is $\sum_x\alpha_x\ket x\otimes\ket\beta$. The projection $\Pi$ and the reflection about $\ket\beta$ are both block diagonal with respect to the basis of $\mathsf S_1\cdots\mathsf S_m$. Indeed,
\[
 \Pi=\sum_x\ket x\bra x\otimes\Pi_{S(x)}\qquad\text{with}\qquad\Pi_S=\sum_{z\in S}\ket z\bra z ,
\]
and the reflection acts on $\mathsf Y$ alone. Hence, after step~2, the state is
\[
 \sum_x\alpha_x\ket x\otimes G_x^\ell\ket\beta\qquad\text{with}\qquad G_x=(2\ket\beta\bra\beta-I)(I-2\Pi_{S(x)}) .
\]

Call $x$ \emph{distinct} if its entries are distinct. For distinct $x$,
\begin{align*}
 \|\Pi_{S(x)}\beta\|^2&=\sum_{z\in S(x)}|\beta_z|^2\\
 &=m/N=P_\star
\end{align*}
holds exactly, because $S(x)$ has $m$ elements and $|\beta_z|^2=1/N$ for every $z$. Therefore \cref{lem:amplification} gives
\begin{align*}
 G_x^\ell\ket\beta&=\sin(L\vartheta_0)\ket{u_x}+\cos(L\vartheta_0)\ket{u_x^\perp},\\
 \ket{u_x}&=\frac{\Pi_{S(x)}\ket\beta}{\sqrt{P_\star}}=\sqrt{\frac Nm}\sum_{j=1}^m\beta_{x_j}\ket{x_j},
\end{align*}
where $\ket{u_x^\perp}$ is a unit vector supported outside $S(x)$.

The terms with $x$ not distinct have total squared norm
\[
 p_{\mathrm{rep}}=\smashoperator{\sum_{x\text{ not distinct}}}|\alpha_x|^2 ,
\]
which is the probability that $m$ independent uniformly random outcomes are not all distinct. Hence, by the union bound,
\[
 p_{\mathrm{rep}}\le m(m-1)/(2N)\le m^2/(2N) .
\]
Let
\[
 \Phi_{\mathrm{good}}=\smashoperator{\sum_{x\text{ distinct}}}\alpha_x\ket x\otimes\ket{u_x},\qquad\|\Phi_{\mathrm{good}}\|^2=1-p_{\mathrm{rep}} .
\]
Put $a=\sin(L\vartheta_0)$, $b=\cos(L\vartheta_0)$, and let
\[
 \Phi_{\mathrm{out}}=\smashoperator{\sum_{x\text{ distinct}}}\alpha_x\ket x\otimes\ket{u_x^\perp}.
\]
Let $\Phi_{\mathrm{rep}}$ be the terms of the amplified state $\Phi$ with repeated table entries. The vectors $\Phi_{\mathrm{good}}$, $\Phi_{\mathrm{out}}$, and $\Phi_{\mathrm{rep}}$ are mutually orthogonal: the first two have disjoint marked and unmarked support for each distinct $x$, and the last has disjoint table support. Their squared norms are respectively $1-p_{\mathrm{rep}}$, $1-p_{\mathrm{rep}}$, and $p_{\mathrm{rep}}$. Thus
\[
 \|\Phi-\Phi_{\mathrm{good}}\|^2
 =\bigl((a-1)^2+b^2\bigr)(1-p_{\mathrm{rep}})+p_{\mathrm{rep}}
 =2(1-a)(1-p_{\mathrm{rep}})+p_{\mathrm{rep}}.
\]
Writing $\Delta=L\vartheta_0-\pi/2$, we have $a=\cos\Delta$ and $|\Delta|\le\vartheta_0$, so $2(1-a)=4\sin^2(\Delta/2)\le\vartheta_0^2$. Consequently,
\[
 \|\Phi-\Phi_{\mathrm{good}}\|\le\sqrt{\vartheta_0^2(1-p_{\mathrm{rep}})+p_{\mathrm{rep}}}
 \le\vartheta_0+\sqrt{p_{\mathrm{rep}}}
 \le\frac\pi3\sqrt{m/N}+\sqrt{p_{\mathrm{rep}}}.
\]
The last inequality follows from $m/N\le1/4$ and the convexity bound for $\arcsin$ used in \cref{lem:squaring-amplification-error}.

\emph{Deletion.} The operation of step~3 is a unitary $U_{\mathrm{del}}$. For $m=2$ and distinct $x_1\ne x_2$, the vector $\Phi_{\mathrm{good}}$ contains the terms
\[
 \sqrt{N/2}\,\alpha_{x_1}\alpha_{x_2}\ket{x_1,x_2}\bigl(\beta_{x_1}\ket{x_1}+\beta_{x_2}\ket{x_2}\bigr)\ket0_{\mathsf J} .
\]
The first term has index $j=1$ in $\mathsf J$, so the swap moves $x_2$ into $\mathsf S_1$ and the bitwise addition clears $\mathsf S_2$. The second term has index $j=2$. Writing $w$ for the remaining entry and $z$ for the output, both terms become
\[
 \sqrt{N/2}\,\alpha_w\alpha_z\beta_z\ket w\ket{0^n}\ket z\ket j_{\mathsf J}\qquad\text{with }j\in\{1,2\},
\]
and summing over $w\ne z$ gives the projection of
\[
 \ket\alpha\otimes\ket{0^n}\otimes\ket{\alpha\circ\beta}\otimes\frac1{\sqrt2}(\ket1+\ket2)
\]
onto the basis states with $w\ne z$. The general case is the same computation: on a basis state $\ket x\ket{x_j}\ket0_{\mathsf J}$ with $x$ distinct, the index computed into $\mathsf J$ is $j$, and $U_{\mathrm{del}}$ maps the state to
\[
 \ket{x^{(j)}}\ket{0^n}_{\mathsf S_m}\ket{x_j}_{\mathsf Y}\ket j_{\mathsf J},
\]
where $x^{(j)}\in[N]^{m-1}$ lists the entries of $x$ other than $x_j$ in the order in which they appear in $\mathsf S_1,\dots,\mathsf S_{m-1}$ after the swap.

For fixed $j$, the map $[N]^m\to[N]^{m-1}\times[N]$, $x\mapsto(x^{(j)},x_j)$, is a bijection, and $\alpha_x=\alpha_{x^{(j)}}\alpha_{x_j}$, with $\alpha_{x^{(j)}}$ the product of the amplitudes of the entries of $x^{(j)}$. Therefore
\begin{align*}
 U_{\mathrm{del}}\Phi_{\mathrm{good}}
 &=\sqrt{\frac Nm}\sum_{j=1}^m\ \smashoperator{\sum_{x\text{ distinct}}}\alpha_{x^{(j)}}\,\alpha_{x_j}\beta_{x_j}\,\ket{x^{(j)}}\ket{0^n}\ket{x_j}\ket j\\
 &=\Pi_{\mathrm{dist}}\,\Phi_\circ ,
\end{align*}
where
\[
 \Phi_\circ=\ket\alpha^{\otimes(m-1)}\otimes\ket{0^n}_{\mathsf S_m}\otimes\ket{\alpha\circ\beta}_{\mathsf Y}\otimes\frac1{\sqrt m}\sum_{j=1}^m\ket j_{\mathsf J}
\]
and $\Pi_{\mathrm{dist}}$ is the projection onto the basis states $\ket w\ket{0^n}\ket y\ket j$ in which the $m$ entries of $(w,y)\in[N]^{m-1}\times[N]$ are distinct. Indeed, expanding $\Phi_\circ$ in the basis gives the amplitude $m^{-1/2}\alpha_w\sqrt N\alpha_y\beta_y$ at $\ket w\ket{0^n}\ket y\ket j$, and the terms of $U_{\mathrm{del}}\Phi_{\mathrm{good}}$ are exactly those in which $(w,y)$ comes from a distinct $x$, with the same amplitude.

The squared norm of $\Phi_\circ-\Pi_{\mathrm{dist}}\Phi_\circ$ is the total weight of the basis states in which the entries of $(w,y)$ are not distinct. Since $|\alpha_w|^2|\sqrt N\alpha_y\beta_y|^2/m=N^{-m}/m$ for every basis state, this weight is again the probability that $m$ uniformly random outcomes are not distinct, which is $p_{\mathrm{rep}}$. Hence
\[
\begin{aligned}
 \|U_{\mathrm{del}}\Phi-\Phi_\circ\|&\le\|\Phi-\Phi_{\mathrm{good}}\|+\sqrt{p_{\mathrm{rep}}}\\
 &\le\frac\pi3\sqrt{m/N}+2\sqrt{p_{\mathrm{rep}}} ,
\end{aligned}
\]
where the first step uses the triangle inequality and the unitarity of $U_{\mathrm{del}}$, and the second uses the improved orthogonality bound on $\|\Phi-\Phi_{\mathrm{good}}\|$ proved above.

\emph{The output.} For unit vectors $u,v$, the trace distance satisfies
\[
 \frac12\bigl\|\,\ket u\bra u-\ket v\bra v\,\bigr\|_1=\sqrt{1-|\langle u,v\rangle|^2}\le\|u-v\| ,
\]
where the inequality holds because
\begin{align*}
 1-|\langle u,v\rangle|^2&\le(1-\operatorname{Re}\langle u,v\rangle)(1+\operatorname{Re}\langle u,v\rangle)\\
 &\le2(1-\operatorname{Re}\langle u,v\rangle)=\|u-v\|^2 .
\end{align*}
The partial trace does not increase the trace distance, and the reduced state of $\mathsf Y$ in $\Phi_\circ$ is $\ket{\alpha\circ\beta}\bra{\alpha\circ\beta}$. Therefore the output $\rho$ satisfies
\[
 \frac12\bigl\|\rho-\ket{\alpha\circ\beta}\bra{\alpha\circ\beta}\,\bigr\|_1\le\frac\pi3\sqrt{m/N}+2\sqrt{p_{\mathrm{rep}}},
\]
and $2\sqrt{p_{\mathrm{rep}}}\le\sqrt2\,m/\sqrt N$ because $p_{\mathrm{rep}}\le m^2/(2N)$.
\end{proof}

The theorem holds for every fixed pair of input states with uniform magnitudes, not only on average over an ensemble, and the algorithm does not need to know the phases. With the canonical oracles, the registers are copies of $\mathsf Q=\C\ket\bot\oplus\C^N$, the prepared states lie in $\C^N$, and the reflection about $\ket\beta$ maps $\C^N$ to itself, so the component along $\ket\bot$ is never populated and the computation above applies without change.

\subsection[A lower bound for Hadamard squares of Haar-random states]{A lower bound for Hadamard squares of Haar-random states}\label{sec:hadamard-lower}

When both inputs are the same Haar-random state, we call their normalized coordinatewise product its Hadamard square. We have only a lower bound for this case, which follows from \cref{cor:collision-formal} by measuring the output.

\begin{proposition}[Hadamard squares of Haar-random states]\label{prop:self-product}
Fix $0<\nu<1/12$. Let $\omega$ be either a Haar-random state $\phi$, accessed through $O_\phi$, or a Haar-random unitary $U$, accessed through $U$ and $U^\dagger$. Write $\phi_\omega=\phi$ in the first case and $\phi_\omega=U\ket{0^n}$ in the second. Suppose an algorithm with no supplied copies of $\ket{\phi_\omega}$ outputs a state $\rho_\omega$ on $\C^N$ with
\[
 \E_\omega\frac12\bigl\|\rho_\omega-\ket{\phi_\omega\circ\phi_\omega}\bra{\phi_\omega\circ\phi_\omega}\,\bigr\|_1\le\nu ,
\]
where, writing $\ket{\phi_\omega}=\sum_z\alpha_z\ket z$ and $p_z=|\alpha_z|^2$,
\[
 \ket{\phi_\omega\circ\phi_\omega}=\frac{\sum_z\alpha_z^2\ket z}{\bigl(\sum_yp_y^2\bigr)^{1/2}} .
\]
Then the algorithm makes $\Omega(N^{1/3})$ queries, with a constant that depends on $\nu$.
\end{proposition}

\begin{proof}
For either oracle model, measuring $\ket{\phi_\omega\circ\phi_\omega}$ in the computational basis gives $z$ with probability
\[
 \frac{|\alpha_z|^4}{\sum_yp_y^2}=r_{p_\omega}(z),
\]
where $r_{p_\omega}$ is the collision distribution. Measuring $\rho_\omega$ instead gives a distribution $a_\omega$ with
\[
 \TV(a_\omega,r_{p_\omega})\le\frac12\bigl\|\rho_\omega-\ket{\phi_\omega\circ\phi_\omega}\bra{\phi_\omega\circ\phi_\omega}\,\bigr\|_1 ,
\]
because a measurement does not increase the trace distance. Hence $\E_\omega\TV(a_\omega,r_{p_\omega})\le\nu$, and \cref{cor:collision-formal} gives the claim in the corresponding oracle model.
\end{proof}

\bibliographystyle{alpha}
\bibliography{references_v2}

\clearpage
\appendix
\section{Deferred proofs}\label{app:deferred-proofs}

\subsection[Proofs for the Haar lower bound]{Proofs for the Haar lower bound}\label{app:haar-deferred}

We give the proofs deferred from \cref{sec:haar}. We use its notation throughout. In the Dirichlet calculations, $a_i\ge1$, $|a|=\sum_i a_i$, and the derivatives act on polynomial expressions in $p_1,\dots,p_N$ before restriction to the simplex. The tangent derivatives $T_i$ are defined in \eqref{eq:haar-tangent}.

\subsubsection[The level of the supplied copies]{The level of the supplied copies}\label{app:haar-copies}

\begin{proof}[Proof of \cref{lem:haar-copies}]
\emph{Step 1: the level of one string.} Expand
\[
 \ket\phi^{\otimes s}=\sum_{x\in[N]^s}\alpha_{x_1}\cdots\alpha_{x_s}\ket{x_1,\dots,x_s} .
\]
The monomial of the string $x$ contains no conjugates, so its pattern is $(n_j)_j$, where $n_j$ is the number of occurrences of $j$ in $x$, and its polynomial factor is constant. Its level is
\[
 \sum_j\max(n_j-1,0)=s-\operatorname{dist}(x),
\]
where $\operatorname{dist}(x)$ is the number of distinct entries of $x$.

\emph{Step 2: the expected number of distinct samples.} Distinct strings are orthogonal in the register basis and the squared amplitude of $x$ is $p_{x_1}\cdots p_{x_s}$, so
\begin{align*}
 \lambda(\ket\phi^{\otimes s})^2&=\E_\phi\sum_xp_{x_1}\cdots p_{x_s}\bigl(s-\operatorname{dist}(x)\bigr)\\
 &=s-\E[\operatorname{dist}],
\end{align*}
where in the last expression $p$ is uniform on the simplex and, conditional on $p$, the entries of $x$ are independent samples from $p$.

By linearity,
\begin{align*}
 \E[\operatorname{dist}]&=\sum_i\Pr[i\text{ appears}]\\
 &=N\bigl(1-\E_\phi(1-p_i)^s\bigr).
\end{align*}
By \cref{lem:dirichlet}(iii), $1-p_i$ has the Beta distribution with parameters $(N-1,1)$, so
\begin{align*}
 \E(1-p_i)^s&=(N-1)\int_0^1t^{N-2+s}\,dt\\
 &=\frac{N-1}{N+s-1} .
\end{align*}
Hence
\[
 \E[\operatorname{dist}]=\frac{Ns}{N+s-1}\qquad\text{and}\qquad s-\E[\operatorname{dist}]=\frac{s(s-1)}{N+s-1} . \qedhere
\]
\end{proof}

\subsubsection[Integration by parts]{Integration by parts}\label{app:haar-ibp}

\begin{proof}[Proof of \cref{lem:haar-ibp}]
\emph{Step 1: reduction to monomials.} Both sides of \eqref{eq:haar-ibp} are bilinear in $(g,v)$, so it suffices to take $g=p^b$ and $v=p^c$ for vectors $b,c$ of nonnegative integers, where $p^b=\prod_jp_j^{b_j}$. Put
\[
 \nu=\E_a[p^{b+c}],\qquad \nu_i=\E_a[p^{b+c+e_i}] .
\]

\emph{Step 2: the two sides.} The partial derivatives of $p^b$ are $\partial_ip^b=b_ip^{b-e_i}$, read as $0$ when $b_i=0$. Since moreover $\sum_jp_j\partial_jp^b=|b|p^b$ by Euler's identity, and in the same way for $p^c$, we have
\[
 p_iT_ip^b=b_ip^b-|b|p^{b+e_i}\qquad\text{and}\qquad p_iT_ip^c=c_ip^c-|c|p^{c+e_i},
\]
so the left side of \eqref{eq:haar-ibp} equals $b_i\nu-|b|\nu_i$, and the right side equals $(|a|+|c|)\nu_i-(a_i+c_i)\nu$.

\emph{Step 3: comparison of the two sides.} After rearrangement, the two agree because
\[
 (|a|+|b|+|c|)\,\nu_i=(a_i+b_i+c_i)\,\nu,
\]
which is \cref{lem:dirichlet}(i). Indeed, the formulas for the two moments differ only by the factors
\[
 \frac{(a_i)_{b_i+c_i+1}}{(a_i)_{b_i+c_i}}=a_i+b_i+c_i\qquad\text{and}\qquad\frac{(|a|)_{|b|+|c|+1}}{(|a|)_{|b|+|c|}}=|a|+|b|+|c|
\]
in the numerator and in the denominator, respectively.
\end{proof}

\subsubsection[The expression for the operator]{The expression for the operator}\label{app:haar-La}

\begin{proof}[Verification of \eqref{eq:haar-La-tangent}]
Indeed, the terms involving $a$ agree, since $\sum_ia_iT_ig=\sum_i(a_i-|a|p_i)\partial_ig$. For the others, write $\mathcal E=\sum_jp_j\partial_j$, so that $T_i=\partial_i-\mathcal E$ by \eqref{eq:haar-tangent}. Then
\begin{align*}
 T_i^2g&=\partial_i^2g-\partial_i(\mathcal Eg)-\mathcal E(\partial_ig)+\mathcal E(\mathcal Eg),\\
 \partial_i(\mathcal Eg)&=\partial_ig+\mathcal E(\partial_ig),\\
 \mathcal E(\mathcal Eg)&=\mathcal Eg+\sum_{j,k}p_jp_k\partial_j\partial_kg,
\end{align*}
and therefore
\begin{align*}
 \sum_ip_iT_i^2g&=\sum_ip_i\partial_i^2g-2\sum_{i,j}p_ip_j\partial_i\partial_jg-\mathcal Eg\\
 &\quad+\Bigl(\sum_ip_i\Bigr)\Bigl(\mathcal Eg+\sum_{j,k}p_jp_k\partial_j\partial_kg\Bigr)\\
 &=\sum_{i,j}p_i(\delta_{ij}-p_j)\partial_i\partial_jg\qquad\text{on the simplex}.
\end{align*}
\end{proof}

\subsubsection[Symmetry and eigenvalues]{Symmetry and eigenvalues}\label{app:haar-energy}

\begin{proof}[Proof of \cref{lem:haar-energy}]
\emph{Part (i).} Apply \cref{lem:haar-ibp} with $\bar f$ in place of $g$ and $T_ig$ in place of $v$. Since $T_i$ has real coefficients, $T_i\bar f=\overline{T_if}$, and the identity reads
\[
 \E_a\bigl[p_i\,\overline{T_if}\,T_ig\bigr]=\E_a\bigl[\bar f\bigl((|a|p_i-a_i)T_ig-p_iT_i^2g\bigr)\bigr].
\]
Summing over $i$, using the identity $\sum_ip_iT_ig=0$ on the simplex, and comparing with \eqref{eq:haar-La-tangent} gives
\begin{align*}
 \E_a\Gamma(f,g)&=-\E_a\Bigl[\bar f\sum_i(a_iT_ig+p_iT_i^2g)\Bigr]\\
 &=-\langle f,L_ag\rangle_a .
\end{align*}
Since $\Gamma(f,g)=\overline{\Gamma(g,f)}$, the operator $L_a$ is symmetric, and
\[
 \langle g,L_ag\rangle_a=-\E_a\Gamma(g,g)\le0 .
\]

\emph{Part (ii).} Each term of $L_a$ maps polynomials of degree at most $m$ to such polynomials. Let $g\in \mathcal R^a_r$, choose a representative of degree at most $r$, and let $g_r$ be its homogeneous part of degree $r$. The terms of $L_a$ that do not lower the degree are $-\sum_{i,j}p_ip_j\partial_i\partial_j-|a|\sum_ip_i\partial_i$, and two applications of Euler's identity evaluate them on $g_r$ to $-r(r-1)g_r-|a|rg_r$. Hence $\Delta=L_ag+r(|a|+r-1)g$ has degree at most $r-1$.

For every $v$ of degree at most $r-1$, we have $\langle v,g\rangle_a=0$ and, by the symmetry in (i),
\[
 \langle v,L_ag\rangle_a=\langle L_av,g\rangle_a=0
\]
because $L_av$ has degree at most $r-1$. Hence $\langle v,\Delta\rangle_a=0$, and taking $v=\Delta$ gives $\Delta=0$.

\emph{Part (iii).} Decompose $h=\sum_{j\le r}h_j$ with $h_j\in \mathcal R^a_j$. By (i) and (ii),
\[
 \E_a\Gamma(h,h)=-\langle h,L_ah\rangle_a=\sum_jj(|a|+j-1)\|h_j\|_a^2,
\]
because the $h_j$ are orthogonal and $L_a$ maps each $\mathcal R^a_j$ to itself, and the coefficient $j(|a|+j-1)$ is at most $r(|a|+r-1)$.
\end{proof}

\subsubsection[The mass of a set of outcomes]{The mass of a set of outcomes}\label{app:haar-subset}

\begin{proof}[Proof of \cref{lem:haar-subset}]
\emph{Step 1: reductions.} Since $|h|^2=(\operatorname{Re}h)^2+(\operatorname{Im}h)^2$ and both parts have degree at most $r$, it suffices to treat real $h$. Write $y^2=\E_a[p(S)h^2]$ and $b=a(S)$. If $r>|a|$, then the claim holds because $y^2\le\|h\|_a^2$ and $(b+r)/|a|>1$. We may therefore assume $r\le |a|$, whence $r(|a|+r-1)\le2r|a|$.

\emph{Step 2: an identity for $y^2$.} We combine three facts:
\begin{enumerate}
\item the formula $L_ap(S)=b-|a|p(S)$;
\item the identity $\E_a\Gamma(f,g)=-\langle f,L_ag\rangle_a$ of \cref{lem:haar-energy}(i), with $f=h^2$ and $g=p(S)$;
\item the product rule $\Gamma(h^2,p(S))=2h\,\Gamma(h,p(S))$, which holds because the form $\Gamma$ is built from first derivatives.
\end{enumerate}
Together they give
\begin{align*}
 |a|y^2-b\|h\|_a^2&=-\E_a[h^2L_ap(S)]\\
 &=\E_a\Gamma(h^2,p(S))=2\,\E_a[h\,\Gamma(h,p(S))].
\end{align*}

\emph{Step 3: a bound on the right side of Step~2.} At every point of the simplex, $\Gamma$ is a positive semidefinite form, so
\begin{align*}
 |\Gamma(h,p(S))|&\le\sqrt{\Gamma(h,h)}\sqrt{\Gamma(p(S),p(S))},\\
 \Gamma(p(S),p(S))&=\sum_{i,j\in S}p_i(\delta_{ij}-p_j)\\
 &=p(S)\bigl(1-p(S)\bigr)\le p(S).
\end{align*}

Taking expectations, applying the Cauchy--Schwarz inequality, and using \cref{lem:haar-energy}(iii) together with the bound $r(|a|+r-1)\le2r|a|$ of Step~1, we obtain
\begin{align*}
 2\,\E_a[h\,\Gamma(h,p(S))]&\le2\,\E_a\bigl[\sqrt{p(S)}|h|\sqrt{\Gamma(h,h)}\bigr]\\
 &\le2y\sqrt{\E_a\Gamma(h,h)}\le2y\sqrt{2r|a|}\,\|h\|_a .
\end{align*}

Finally, by Steps~2 and~3 for the first inequality and by the arithmetic--geometric mean inequality for the second,
\begin{align*}
 |a|y^2-b\|h\|_a^2&\le2y\sqrt{2r|a|}\,\|h\|_a\\
 &\le\frac{|a|}2y^2+4r\|h\|_a^2 ,
\end{align*}
that is,
\[
 y^2\le\frac{(2b+8r)\|h\|_a^2}{|a|}\le\frac{8(b+r)\|h\|_a^2}{|a|} . \qedhere
\]
\end{proof}

\subsubsection[Raising the level by one]{Raising the level by one}\label{app:haar-up}

\begin{proof}[Proof of \cref{lem:haar-up}]
Write $G$ for the input of $L_+$ and $F=F_\bot\otimes\ket\bot+\sum_iF_i\otimes\ket i$ for the input of $L_-$. Norms without subscripts are norms of vector-valued polynomials; \eqref{eq:haar-norm-conversion} and \eqref{eq:haar-density-conversion} convert them to Dirichlet norms. For $L_+$ we consider pairs consisting of a block of $G$ and an outcome $i$. For $L_-$ we consider pairs consisting of an outcome $i$ and a block of $F_i$. We first identify the output blocks and then bound their contributions.

\emph{Step 1: the output blocks.} By the table, the part of each pair that raises the level lies in at most one output block.

For $L_+$, the output records $i$ in the appended register, and the output pattern is the input pattern plus $e_i$. The upward cases are the following two, which are rows one and three of the table with the signs of the entries reversed:
\begin{enumerate}
\item the case $\gamma_i\ge1$, in which the degree under the Dirichlet measure is unchanged;
\item the case $\gamma_i=-1$, in which the degree under the Dirichlet measure rises by one.
\end{enumerate}
The two are distinguished by the value of the output pattern at $i$. Hence the output block and the outcome determine the input block, so the contributions of distinct pairs for $L_+$ are pairwise orthogonal.

\emph{Step 2: the map $L_+$.} Let $u_\gamma h$, with $h\in \mathcal R^a_r$, $a=a(\gamma)$, and $|\gamma|_1+2r\le d$, be a block of $G$. Only the outcomes $i$ in the occupied set $S=\{i:\gamma_i\ne0\}$ receive a contribution one level up from it, and projecting onto that level does not increase norms, so its contribution is at most
\begin{align*}
 \sum_{i\in S}\|\alpha_iu_\gamma h\|^2&=M_\gamma\,\E_a[p(S)|h|^2]\\
 &\le8\,\frac{a(S)+r}{|a|}\,M_\gamma\|h\|_a^2\le16\,\frac dN\,\|u_\gamma h\|^2 ,
\end{align*}
by \eqref{eq:haar-norm-conversion} and \cref{lem:haar-subset}, where the last step uses
\begin{align*}
 a(S)&=|S|+|\gamma|_1\le2|\gamma|_1,\\
 \text{hence}\quad a(S)+r&\le2(|\gamma|_1+2r)\le2d,
\end{align*}
and $|a|\ge N$. Summing over the blocks of $G$, whose contributions are orthogonal, gives the bound $16d/N\le36(d+1)/N$ for $L_+$.

For $L_-$, fix an output block $(\gamma,r)$. The contributions to it from one level below come from pairs of two types, which can be read off the output pattern.
\begin{enumerate}
\item In \emph{type I}, $\gamma_i\le-2$ and the input block is $(\gamma+e_i,r)$, which is the first row of the table.
\item In \emph{type II}, $\gamma_i=0$ and the input block is $(\gamma+e_i,r-1)$, which is the third row.
\end{enumerate}
Each pair for $L_-$ belongs to at most one type and contributes to at most one output block.

For output block $(\gamma,r)=(0,1)$, the inputs $\alpha_iv_i\otimes\ket i$, with $v_i$ constant, contribute $(p_i-1/N)v_i$. We bound their sum in each output block, then add the squared norms of the orthogonal output blocks. Each pair contributes to at most one such block.

\emph{Step 3: repeated unpaired factors in $L_-$.} Fix an output block $(\gamma,r)$ with $a=a(\gamma)$, and let $S=\{i:\gamma_i\le-2\}$. For $i\in S$, the relevant block of $F_i$ is $u_{\gamma+e_i}h_i$, with $h_i\in\mathcal R^{a-e_i}_r$. Multiplication by $\overline{\alpha_i}$ gives $u_\gamma h_i$. Before projection onto $\mathcal R_r^a$, these contributions add to $u_\gamma\sum_{i\in S}h_i$. The projection does not increase norms.

By the Cauchy--Schwarz inequality, pointwise,
\[
 \Bigl|\sum_{i\in S}h_i\Bigr|^2\le p(S)\sum_{i\in S}\frac{|h_i|^2}{p_i} ,
\]
so by \eqref{eq:haar-density-conversion} in the direction of decreasing magnitude,
\begin{align*}
 \Bigl\|u_\gamma\sum_{i\in S}h_i\Bigr\|^2&\le M_\gamma\,\E_a\Bigl[p(S)\sum_{i\in S}\frac{|h_i|^2}{p_i}\Bigr]\\
 &=\sum_{i\in S}M_{\gamma+e_i}\,\E_{a-e_i}\bigl[p(S)|h_i|^2\bigr].
\end{align*}

Each input block lies in $\mathcal P_{\le d}$, and therefore $|\gamma|_1-1+2r\le d$. In the measure $a-e_i$, whose total is $|a|-1\ge N$, we have
\begin{align*}
 (a-e_i)(S)+r&\le a(S)+r\\
 &\le2(|\gamma|_1+2r)\le2(d+1).
\end{align*}
\Cref{lem:haar-subset} therefore gives
\[
 \E_{a-e_i}\bigl[p(S)|h_i|^2\bigr]\le16\,\frac{d+1}N\,\|h_i\|_{a-e_i}^2 ,
\]
and the type I contribution has squared norm at most $(16(d+1)/N)\sum_{i\in S}\|u_{\gamma+e_i}h_i\|^2$.

\emph{Step 4: paired conjugate factors in $L_-$.} For $r=0$ there is no contribution of this type, so assume $r\ge1$. For each $i$ with $\gamma_i=0$, the relevant block of $F_i$ is $u_{\gamma+e_i}h_i$, with $h_i\in\mathcal R^{a+e_i}_{r-1}$. Set $h_i=0$ for the other outcomes. Multiplication by $\overline{\alpha_i}$ gives $u_\gamma p_ih_i$. After summation and projection onto $\mathcal R_r^a$, the contribution is $u_\gamma\Pi^a_r\sum_ip_ih_i$.

This is the situation of \cref{lem:haar-coherent} with $r-1$ in place of $r$, which gives
\begin{align*}
 \Bigl\|\Pi^a_r\sum_ip_ih_i\Bigr\|_a^2&\le\frac r{|a|+2r-1}\sum_i\E_a[p_i|h_i|^2]\\
 &\le\frac{d+1}{2N}\sum_i\E_a[p_i|h_i|^2],
\end{align*}
using $|a|\ge N$ and $r\le(d+1)/2$, the latter from the degree constraint $|\gamma|_1+1+2(r-1)\le d$ of the input blocks.

By \eqref{eq:haar-density-conversion} in the direction of increasing magnitude, $M_\gamma\E_a[p_i|h_i|^2]=\|u_{\gamma+e_i}h_i\|^2$, so multiplying by $M_\gamma$ shows that the type II contribution has squared norm at most $((d+1)/(2N))\sum_i\|u_{\gamma+e_i}h_i\|^2$.

\emph{Step 5: adding the contributions.} Within one output block, add the contributions of the two types and apply $\|x+y\|^2\le2\|x\|^2+2\|y\|^2$. Its squared norm is at most $2(16+1/2)(d+1)/N\le36(d+1)/N$ times the sum of the squared norms of its input components. We then sum over the orthogonal output blocks. Each pair contributes to at most one block and one type, so each input component is counted at most once. Their total squared norm is at most $\sum_i\|F_i\|^2\le\|F\|^2$, proving the bound for $L_-$.
\end{proof}

\subsubsection[Lowering the level by one]{Lowering the level by one}\label{app:haar-down}

\begin{proof}[Proof of \cref{lem:haar-down}]
\emph{Step 1: reformulation.} Let
\[
 L_{\downarrow}=\sum_{k\ge0}P_kLP_{k+1}
\]
be the part of $L$ that lowers the level. Since its terms lie at distinct levels, the claim is
\[
 \|L_{\downarrow}F\|^2\le36\,\frac{d+2}N\,\|F\|^2 .
\]

\emph{Step 2: duality.} The state $L_{\downarrow}F$ lies in $\mathcal P_{\le d+1}$, because $L$ raises the degree by at most one, so
\[
 \|L_{\downarrow}F\|=\sup_{G\in\mathcal P_{\le d+1},\ \|G\|=1}|\langle G,L_{\downarrow}F\rangle| .
\]
Since the $P_k$ are self-adjoint,
\begin{align*}
 \langle G,L_{\downarrow}F\rangle&=\sum_k\langle P_{k+1}L^\dagger P_kG,F\rangle=\langle(L^\dagger)_{\uparrow}G,F\rangle,\\
 (L^\dagger)_{\uparrow}&=\sum_kP_{k+1}L^\dagger P_k .
\end{align*}

\emph{Step 3: the bound.} Since $L_+^\dagger=L_-$ and $L_-^\dagger=L_+$, we can apply \cref{lem:haar-up} to $L^\dagger$ at degree $d+1$. This gives $\|(L^\dagger)_{\uparrow}G\|^2\le36(d+2)/N$. The Cauchy--Schwarz inequality bounds the pairing:
\[
 |\langle(L^\dagger)_{\uparrow}G,F\rangle|\le\|(L^\dagger)_{\uparrow}G\|\,\|F\| . \qedhere
\]
\end{proof}

\subsubsection[The heaviest outcome]{The heaviest outcome}\label{app:haar-heaviest}

\begin{proof}[Proof of \cref{lem:heaviest}]
For $p$ uniform on the simplex and $0<t<1$, the joint tail satisfies
\[
 \Pr[p_{i_1}>t,\dots,p_{i_k}>t]=(1-kt)_+^{N-1},
\]
since the set in question is a copy of the simplex scaled by $1-kt$.
Inclusion--exclusion therefore gives
\[
 \Pr[\max_zp_z>t]=\sum_{k\ge1}(-1)^{k+1}\binom Nk(1-kt)_+^{N-1} .
\]
Integrating over $t\in(0,1)$ gives
\begin{align*}
 \E\max_zp_z&=\sum_{k=1}^N(-1)^{k+1}\binom Nk\frac1{kN}\\
 &=\frac1N\int_0^1\frac{1-(1-x)^N}x\,dx\\
 &=\frac1N\int_0^1\sum_{j=0}^{N-1}(1-x)^j\,dx=\frac{H_N}N ,
\end{align*}
where the second equality expands $(1-(1-x)^N)/x$ by the binomial theorem and integrates term by term.
\end{proof}

\subsection[Moments of the ideal distribution]{Moments of the ideal distribution}\label{app:squaring-moments}

\begin{proof}[Proof of \cref{lem:moments}]
Steps 1--3 prove part~(i), and Steps 4--7 prove part~(ii).

\emph{Step 1: the mean and the second moment in the Haar ensemble.} By \cref{lem:dirichlet}(i) with all parameters equal to $1$,
\[
 \E p_x^k=\frac{k!}{(N)_k},\qquad \E p_x^kp_y^l=\frac{k!\,l!}{(N)_{k+l}}\quad\text{for }x\ne y .
\]
Hence $\E M_j=N^j\,j!/(N)_j$, which is the stated value, and it is at most $j!$ because $(N)_j\ge N^j$.

For the second moment, the two expectations above give
\begin{align*}
 \E M_j^2&=N^{2j-2}\Bigl(N\,\E p_x^{2j}+N(N-1)\,\E p_x^jp_y^j\Bigr)\\
 &=N^{2j-1}\,\frac{(2j)!+(N-1)(j!)^2}{(N)_{2j}}\\
 &\le(j!)^2+\frac{(2j)!}N ,
\end{align*}
where the last step holds again because $(N)_{2j}\ge N^{2j}$.

\emph{Step 2: the variance of $M_j$ in the Haar ensemble.} We have
\[
 \E M_j\ge j!\,\Bigl(1-\frac{j(j-1)}{2N}\Bigr),
\qquad\text{because}\qquad
 \prod_{i<j}\Bigl(1+\frac iN\Bigr)^{-1}\ge1-\sum_{i<j}\frac iN ,
\]
so that $(\E M_j)^2\ge(j!)^2(1-j^2/N)$, which holds trivially when $j^2\ge N$. With Step 1, the bound on $\Var M_j=\E M_j^2-(\E M_j)^2$ follows.

\emph{Step 3: the value of $\E(M_2-2)^2$.} We have $\E(M_2-2)^2=\E M_2^2-4\E M_2+4$, and its value follows from
\[
 \E M_2^2=\frac{N^2(4N+20)}{(N+1)(N+2)(N+3)},
\]
which is the case $j=2$ of Step 1 before the last inequality. The bound $4/N$ follows from $N(4N^2-4N+24)\le4(N+1)(N+2)(N+3)$.

\emph{Step 4: the pairing bound in the Fourier ensemble.} Let
\[
 X_z=\sqrt N\widehat f(z)=N^{-1/2}\sum_x\chi_z(x)f(x),\qquad\text{so that}\qquad M_j=\frac1N\sum_zX_z^{2j},
\]
and write $r_x=\chi_z(x)f(x)$. The $r_x$ are independent uniform signs. In the expansion
\[
 X_z^{2j}=N^{-j}\sum_{x_1,\dots,x_{2j}}r_{x_1}\cdots r_{x_{2j}},
\]
the expectation of a term is $1$ if every address occurs an even number of times in the tuple $(x_1,\dots,x_{2j})$ and $0$ otherwise. Every tuple of the first kind is constant on the pairs of at least one partition of its $2j$ positions into $j$ pairs, obtained by pairing up the positions of each address. There are $(2j-1)!!$ partitions of the $2j$ positions into $j$ pairs, and each is compatible with exactly $N^j$ tuples, namely those constant on each of its pairs, so there are at most $(2j-1)!!\,N^j$ tuples of the first kind, and $\E X_z^{2j}\le(2j-1)!!=\kappa_j$. Hence $\E M_j\le\kappa_j$. By the Cauchy--Schwarz inequality,
\[
 \E[X_z^{2j}X_w^{2j}]\le(\E X_z^{4j}\,\E X_w^{4j})^{1/2}\le\kappa_{2j},
\]
hence $\E M_j^2\le\kappa_{2j}$.

For $N\ge j$, the tuples with exactly $j$ distinct addresses, each occurring twice, number $\kappa_jN(N-1)\cdots(N-j+1)$: their pairs of equal positions determine a unique pairing. Thus
\[
 \kappa_j\frac{N(N-1)\cdots(N-j+1)}{N^j}\le\E M_j\le\kappa_j,
\]
which proves $\E M_j\to\kappa_j$ as $N\to\infty$ for every fixed $j$.

\emph{Step 5: the exact fourth and sixth moments.} The exact values are obtained by counting the tuples of the first kind exactly. For $2j=4$, they are the $3N(N-1)$ tuples that consist of two distinct addresses, each twice, and the $N$ constant tuples, so
\[
 \E X_z^4=\frac{3N^2-2N}{N^2}=3-\frac2N .
\]
For $2j=6$, the tuples of the first kind are the $15N(N-1)(N-2)$ tuples with three distinct addresses, each twice; the $15N(N-1)$ tuples with one address four times and another twice; and the $N$ constant tuples. Hence
\[
 \E X_z^6=\frac{15N^2-30N+16}{N^2} .
\]

\emph{Step 6: the variance of $M_2$ in the Fourier ensemble.} Summing over $z$ first and using \eqref{eq:prelim-walsh}, we obtain
\begin{align*}
 M_2&=\frac1{N^3}\sum_z\sum_{x,y,u,v}\chi_z(x\oplus y\oplus u\oplus v)f(x)f(y)f(u)f(v)\\
 &=\frac1{N^2}\sum_{x,y,u}f(x)f(y)f(u)f(x\oplus y\oplus u) .
\end{align*}
The triples $(x,y,u)$ in the last sum are of two kinds.
\begin{enumerate}
\item If two of $x,y,u$ coincide, the summand equals $1$, and there are $3N^2-2N$ such triples.
\item Otherwise $x,y,u$ and $v=x\oplus y\oplus u$ are four distinct points, and the summand is a product of four independent signs with mean zero. There are $N(N-1)(N-2)$ such triples.
\end{enumerate}
The product of two summands of the second kind has mean $1$ if their four-element sets coincide and mean $0$ otherwise, since otherwise some sign occurs an odd number of times. Each four-element set $\{a,b,c,d\}$ with $a\oplus b\oplus c\oplus d=0$ arises from exactly $4!=24$ ordered triples $(x,y,u)$, one for each ordering of its elements. Hence each of the $N(N-1)(N-2)$ triples of the second kind has a product of mean $1$ with exactly $24$ triples, itself included, and
\[
 \Var M_2=N^{-4}\cdot24\cdot N(N-1)(N-2)\le\frac{24}N .
\]
Since $\E M_2=3-2/N$, this gives $\E(M_2-3)^2=\Var M_2+(2/N)^2\le28/N$.

\emph{Step 7: the variance of $M_3$ in the Fourier ensemble.} The variance of $M_3$ is bounded in the same way. Summing over $z$ first gives
\[
 M_3=N^{-3}\sum f(x_1)\cdots f(x_6),
\]
the sum running over the $6$-tuples with $x_1\oplus\cdots\oplus x_6=0$. For such a tuple, let $\mathcal Y$ be the set of addresses that occur an odd number of times in it. The product then equals $\prod_{y\in\mathcal Y}f(y)$, and the set $\mathcal Y$ satisfies $\bigoplus_{y\in\mathcal Y}y=0$. Moreover, $|\mathcal Y|\in\{0,4,6\}$, because $|\mathcal Y|$ has the parity of the total multiplicity $6$, and the sum of two distinct addresses is not $0$. We use three facts:
\begin{enumerate}
\item the tuples with $\mathcal Y=\varnothing$ contribute $\E M_3$;
\item a product with $\mathcal Y\ne\varnothing$ has mean zero;
\item two products with distinct sets $\mathcal Y$ are uncorrelated.
\end{enumerate}
Hence
\[
 \Var M_3=N^{-6}\sum_{\mathcal Y\ne\varnothing}c_{\mathcal Y}^2,
\]
where $c_{\mathcal Y}$ is the number of tuples with odd set $\mathcal Y$. We count these tuples for $|\mathcal Y|=6$ and $|\mathcal Y|=4$.
\begin{enumerate}
\item A set of size $6$ arises from its $720$ orderings, so $c_{\mathcal Y}=720$, and there are at most $N^5/720$ such sets, because an ordering of such a set is determined by its first five entries.
\item A set of size $4$ arises from the $480$ tuples in which one of its elements occurs three times and from the $360(N-4)$ tuples that contain one further address twice, so $c_{\mathcal Y}\le400N$, and there are at most $N^3/24$ such sets, because an ordering is determined by its first three entries.
\end{enumerate}
Hence
\[
 \Var M_3\le N^{-6}\Bigl(720N^5+\frac{N^3}{24}(400N)^2\Bigr)\le\frac{10^4}N . \qedhere
\]
\end{proof}

\subsection[Heavy outputs]{Heavy outputs}\label{app:heavy-outputs}

We use the Haar $L^2$ inner product, the patterns, and the projections $P_k$ of \cref{sec:haar-level}. All auxiliary initial registers are independent of the oracle, and any mixed state or measurement can be purified as in \cref{lem:normal-form}. For a measurable $w:[0,1]\to[0,1]$, define the multiplication operator
\[
 (X_wF)(\phi)=\sum_z w(p_z)\Pi_zF(\phi).
\]
It is a positive contraction on the space of all square-integrable vector-valued functions. Its expectation in a final state is the expected payoff $w(p_Z)$.

\begin{lemma}[Bounded output rewards]\label{lem:heavy-bounded}
Let $N\ge2$ and define
\begin{equation}\label{eq:heavy-kappa}
 \kappa_w=\max\left\{
 \begin{aligned}
  &\max_{N\le A\le2N-1}\E\,w\bigl(\operatorname{Beta}(1,A-1)\bigr),\\
  &\max_{N+1\le A\le2N}\E\,w\bigl(\operatorname{Beta}(2,A-2)\bigr)
 \end{aligned}
 \right\},
\end{equation}
where $A$ ranges over integers. This is the largest value of $\langle G,X_wG\rangle$ over unit vectors $G$ of level zero. Every unit vector-valued polynomial $F$ satisfies
\begin{equation}\label{eq:heavy-extraction}
 \langle F,X_wF\rangle\le\kappa_w+\lambda(F).
\end{equation}
For $w(x)=\mathbf1_{\{x>\theta\}}$, $0<\theta<1$,
\begin{equation}\label{eq:heavy-calibration}
 \kappa_w=(1-\theta)^{N-1}\bigl(1+(N-1)\theta\bigr).
\end{equation}
\end{lemma}

\begin{proof}
\emph{Step 1: level zero.} A level-zero block has a square-free pattern $\gamma$ and a constant polynomial factor. By \eqref{eq:haar-norm-conversion}, its normalized inner product uses $\operatorname{Dir}(a(\gamma))$, where $a_i=1+|\gamma_i|$ and $A=N+|\gamma|_1$. On output branch $z$, $p_z$ has law $\operatorname{Beta}(1,A-1)$ if $\gamma_z=0$, and $\operatorname{Beta}(2,A-2)$ if $|\gamma_z|=1$. In the first case $0\le|\gamma|_1\le N-1$; in the second $1\le|\gamma|_1\le N$. Multiplication by $w(p_z)$ preserves the pattern and the output branch, so cross terms between distinct patterns or branches vanish. Thus $\langle G,X_wG\rangle\le\kappa_w\|G\|^2$ for every level-zero $G$. Every integer $A$ in either range is realized by a square-free pattern and a constant vector supported on the specified output branch. After normalization, this vector attains the corresponding expectation. Hence the finite maximum in \eqref{eq:heavy-kappa} is exact.

\emph{Step 2: the other levels.} Write $F_0=P_0F$ and $F_+=F-F_0$. If $F_0\ne0$, put $G=F_0/\|F_0\|$. Orthogonality and $\|F\|=1$ give
\[
 \frac12\bigl\|\,|F\rangle\langle F|-|G\rangle\langle G|\,\bigr\|_1
 =\sqrt{1-|\langle F,G\rangle|^2}
 =\|F_+\|\le\lambda(F).
\]
These are rank-one operators on the Haar $L^2$ Hilbert space. In particular, $G$ need not be normalized for each $\phi$ or be an algorithm state. Since $0\le X_w\le I$, its expectations differ by at most this trace distance. Together with Step~1, this proves \eqref{eq:heavy-extraction}, including all cross terms. If $F_0=0$, then $\lambda(F)\ge1$ and the bound follows from $\langle F,X_wF\rangle\le1$.

\emph{Step 3: indicator rewards.} Write $v=1-\theta$. The $\operatorname{Beta}(1,A-1)$ tail is $v^{A-1}$, whose maximum is $v^{N-1}$. For $b\ge1$, integrating the density $b(b+1)x(1-x)^{b-1}$ gives
\begin{equation}\label{eq:heavy-beta-tail}
 \Pr[\operatorname{Beta}(2,b)>\theta]
 =v^b(1+b\theta).
\end{equation}
For $S(A)=v^{A-2}(1+(A-2)\theta)$,
\[
 \frac{S(A+1)}{S(A)}
 =1-\frac{(A-1)\theta^2}{1+(A-2)\theta}<1.
\]
The second maximum in \eqref{eq:heavy-kappa} is therefore attained at $A=N+1$ and dominates the first. This proves \eqref{eq:heavy-calibration}.
\end{proof}

\begin{proof}[Proof of \cref{thm:heavy-outputs}]
\emph{Step 1: the calibration and finite bounds.} Under Haar measure, $p_z\sim\operatorname{Beta}(1,N-1)$ and $\Pr[p_z>\theta]=(1-\theta)^{N-1}$. Thus $\theta_N$ in \eqref{eq:heavy-threshold} is its median, and an output independent of $\phi$ has heavy probability $1/2$. Honest sampling weights this marginal by $Np_z$, giving $\operatorname{Beta}(2,N-1)$. Equation~\eqref{eq:heavy-beta-tail} and $(1-\theta_N)^{N-1}=1/2$ give \eqref{eq:heavy-honest}. Writing $a=\ln2$,
\[
 (N-1)\theta_N=(N-1)(1-e^{-a/(N-1)})\le a,
 \qquad (N-1)\theta_N\longrightarrow a,
\]
which proves the stated inequality and limit.

For canonical access with $s$ supplied copies, \cref{thm:haar-lower} gives $\lambda(\Psi)\le\Lambda_{s,q}$. Apply \cref{lem:heavy-bounded} to the unit vector $\Psi$ with $w(x)=\mathbf1_{\{x>\theta_N\}}$ to obtain \eqref{eq:heavy-canonical}. For unitary access, \cref{thm:kretschmer-simulation} gives a single oracle-independent simulator using $2q$ canonical queries. For each fixed $\phi$, its output law and the output law of the unitary algorithm, averaged over $U_\phi$, differ in total variation by at most $(5q+2)/\sqrt N$. A reward in $[0,1]$ changes by at most that distance. Averaging over $\phi$ and applying the canonical bound with $s=0$ proves \eqref{eq:heavy-unitary}. The same argument holds for every deterministic threshold $0<\theta<1$, with baseline \eqref{eq:heavy-calibration}.

\emph{Step 2: the lower bound at the fixed threshold.} If $B^{\mathrm{hog}}_{\theta_N}(\mathcal A)\ge h_{\mathrm{hon}}+\varepsilon$ for $\varepsilon>0$, canonical access requires $\Lambda_{s,q}\ge\varepsilon$. The budget obeys
\begin{align}
 \Lambda_{s,q}
 &\le\frac{s+c_0q\sqrt{s+2q+1}}{\sqrt N}\notag\\
 &\le\frac{(1+\sqrt2c_0)\bigl(s+q\sqrt{s+q+1}\bigr)}{\sqrt N}.
 \label{eq:heavy-tradeoff}
\end{align}
Consequently $s+q\sqrt{s+q+1}\ge\varepsilon\sqrt N/(1+\sqrt2c_0)$. With $s=0$, $q=0$ cannot achieve the target; for $q\ge1$,
\[
 \Lambda_{0,q}\le\frac{\sqrt2c_0q^{3/2}}{\sqrt N},
 \qquad
 q\ge\left(\frac{\varepsilon}{\sqrt2c_0}\right)^{2/3}N^{1/3}.
\]
For unitary access without supplied copies, we may assume $q\le N^{1/3}$, since otherwise the desired lower bound already holds. The simulation error is then at most $(5N^{1/3}+2)/\sqrt N=o(1)$. For all sufficiently large $N$ it is at most $\varepsilon/2$, so \eqref{eq:heavy-unitary} gives $\Lambda_{0,2q}\ge\varepsilon/2$. Since $\Lambda_{0,2q}\le4c_0q^{3/2}/\sqrt N$, this also forces $\Omega_\varepsilon(N^{1/3})$ queries.

\emph{Step 3: the squaring algorithm.} Let $a_\omega$ be the output distribution of $\mathsf{Sq}_m$, where $\omega$ is either form of Haar oracle, and let $r_p(z)=p_z^2/\sum_yp_y^2$. For $N\ge64$ and $m=\lceil N^{1/3}\rceil$, \cref{thm:output-law} gives at most $3N^{1/3}$ queries and $\E_\omega\TV(a_\omega,r_p)\le40N^{-1/3}$. For any threshold $\tau_\phi$, including one depending on the state, this implies
\begin{equation}\label{eq:heavy-tv}
 \left|B^{\mathrm{hog}}_{\tau_\phi}(\mathsf{Sq}_m)
 -\E_\phi\frac{\sum_zp_z^2\mathbf1_{\{p_z>\tau_\phi\}}}{\sum_zp_z^2}\right|
 \le40N^{-1/3}.
\end{equation}
No independence between the threshold and the state is needed here: for each fixed oracle, the heavy outcomes form a subset of the output space.

Use the exponential representation $p_z=G_z/S$, where the $G_z$ are independent exponential random variables of mean one and $S=\sum_zG_z$. If $S\tau_\phi\to t>0$ in probability, the ratio in \eqref{eq:heavy-tv} is
\[
 \frac{N^{-1}\sum_zG_z^2\mathbf1_{\{G_z>S\tau_\phi\}}}
      {N^{-1}\sum_zG_z^2}.
\]
The denominator tends to $2$ by the law of large numbers. For each fixed $u>0$, the numerator with threshold $u$ tends to $\int_u^\infty x^2e^{-x}\,dx$. On the event $|S\tau_\phi-t|<\delta$, the numerator at the random threshold lies between those at $t+\delta$ and $t-\delta$, for $0<\delta<t$. Letting $\delta\downarrow0$ proves convergence in probability of the ratio to
\begin{equation}\label{eq:heavy-collision-tail}
 \frac12\int_t^\infty x^2e^{-x}\,dx
 =e^{-t}\left(1+t+\frac{t^2}{2}\right).
\end{equation}
The ratio lies in $[0,1]$, so its expectation converges as well. Since $S/N\to1$ and $N\theta_N\to\ln2$, taking $\tau_\phi=\theta_N$ proves the fixed-threshold limit $h_\star$.

For the empirical median, write $G_{(j)}$ for the increasing order statistics. At every fixed $u$, the fraction of samples at most $u$ tends to $1-e^{-u}$. For $u<\ln2$ this limit is below $1/2$, and for $u>\ln2$ it is above $1/2$. It follows that both $G_{(N/2)}$ and $G_{(N/2+1)}$ tend to $\ln2$ in probability. Scaling preserves their order, and hence
\[
 S\widehat\theta_N=\frac{G_{(N/2)}+G_{(N/2+1)}}2
 \longrightarrow\ln2,\qquad
 N\widehat\theta_N\longrightarrow\ln2.
\]
Equation~\eqref{eq:heavy-collision-tail} therefore gives the same limit $h_\star$ for the empirical median. Applying the same sandwich argument with $G_z$ in place of $G_z^2$, and denominator $S/N\to1$, shows that honest sampling has limiting empirical-median heavy probability $\int_{\ln2}^\infty xe^{-x}\,dx=h_{\mathrm{hon}}$.

\emph{Step 4: the lower bound for the empirical median.} Fix a target $h>h_{\mathrm{hon}}$. Choose $0<\delta<\ln2$ small enough that
\[
 e^{-(\ln2-\delta)}(1+\ln2-\delta)<h,
\]
and set $\tau_N=(\ln2-\delta)/N$. For every algorithm,
\begin{equation}\label{eq:heavy-median-comparison}
 B^{\mathrm{hog}}_{\widehat\theta_N}(\mathcal A)
 \le B^{\mathrm{hog}}_{\tau_N}(\mathcal A)
       +\Pr[\widehat\theta_N<\tau_N].
\end{equation}
The last probability is $o(1)$ by Step~3 and is independent of the algorithm. The deterministic-threshold version of Step~1 bounds the first term with baseline
\[
 (1-\tau_N)^{N-1}(1+(N-1)\tau_N)
 \longrightarrow e^{-(\ln2-\delta)}(1+\ln2-\delta).
\]
This leaves a fixed positive gap below $h$. The no-copy query estimates of Step~2 therefore give $\Omega_h(N^{1/3})$ in both access models, uniformly over algorithms. Step~3 supplies the matching upper bound whenever $h<h_\star$. These arguments concern fixed targets and ensemble-averaged heavy probabilities; they assert no rate for a target approaching $h_{\mathrm{hon}}$ with $N$.
\end{proof}

\begin{proposition}[Heavy outputs from copies without queries]\label{prop:heavy-copies}
Let an algorithm receive $\ket\phi^{\otimes s}$, tensored with an oracle-independent joint state of its other registers, and make no oracle queries. It may apply arbitrary joint operations on its copies and device registers. Put
\begin{equation}\label{eq:heavy-copy-repeat}
 R_s=
 \begin{cases}
  1-\displaystyle\prod_{j=0}^{s-1}\frac{N-j}{N+j},&0\le s\le N,\\
  1,&s>N,
 \end{cases}
 \qquad
 R_s\le\min\left\{1,\frac{s(s-1)}{N+s-1}\right\},
\end{equation}
where the empty product is one. For every $w$ as in \cref{lem:heavy-bounded},
\begin{equation}\label{eq:heavy-copy-reward}
 \E\,w(p_Z)\le\kappa_w+(1-\kappa_w)R_s.
\end{equation}
In particular,
\[
 B^{\mathrm{hog}}_{\theta_N}(\mathcal A)
 \le\beta_N^{\mathrm{hog}}+(1-\beta_N^{\mathrm{hog}})R_s.
\]
For $s=0$, the heavy probability is exactly $1/2$. A heavy probability of at least $\beta_N^{\mathrm{hog}}+\varepsilon$, for $\varepsilon>0$, requires
\begin{equation}\label{eq:heavy-copy-lower}
 s(s-1)\ge\frac{\varepsilon(N+s-1)}{1-\beta_N^{\mathrm{hog}}},
\end{equation}
and in particular $s=\Omega(\sqrt{\varepsilon N})$.
\end{proposition}

\begin{proof}
Purify the initial state and defer measurements. With no queries, every operation is $\phi$-independent. Before these operations the state is
\[
 \Phi_s(\phi)=\sum_{\mathbf x\in[N]^s}
       \alpha_{x_1}\cdots\alpha_{x_s}\ket{\mathbf x}\otimes\psi_0.
\]
Every monomial has a nonnegative pattern $\gamma$ with $|\gamma|_1=s$ and a constant polynomial factor. Its level is zero precisely when all entries of the tuple are distinct. Fixed operations preserve the pattern and the constant polynomial factor. Consequently the final state's parts $F_0=P_0F$ and $F_+=F-F_0$ have disjoint patterns. Because $X_w$ preserves patterns,
\[
 \langle F_0,X_wF_+\rangle=0.
\]
This is an orthogonality statement in the Haar $L^2$ space and requires no measurement of, or access to, the pattern.

The copy-register basis vectors are orthogonal. For a distinct tuple, \cref{lem:dirichlet}(i) gives $\E_\phi p_{x_1}\cdots p_{x_s}=1/(N)_s$, where $(N)_s=N(N+1)\cdots(N+s-1)$. There are $N(N-1)\cdots(N-s+1)$ such tuples when $s\le N$, and none when $s>N$. Thus $\|F_+\|^2=R_s$ in \eqref{eq:heavy-copy-repeat}; a fixed isometry preserves this norm. Equivalently, $R_s$ is the Haar-averaged probability of a repeat among $s$ independent samples conditional on $p$. Every positive level is at least one, so \cref{lem:haar-copies} gives
\[
 R_s=\|F_+\|^2\le\lambda(F)^2
      =\frac{s(s-1)}{N+s-1}.
\]
Together with $R_s\le1$, this proves the displayed bound.

The level-zero calibration in \cref{lem:heavy-bounded}, the vanishing cross term, and $0\le X_w\le I$ now give
\[
 \langle F,X_wF\rangle
 \le\kappa_w\|F_0\|^2+\|F_+\|^2
 =\kappa_w+(1-\kappa_w)R_s.
\]
The indicator calibration gives the HOG bound, and rearranging it gives \eqref{eq:heavy-copy-lower}. A positive advantage forces $s\ge2$, so $N+s-1\ge N$ and $1-\beta_N^{\mathrm{hog}}\le1$ imply $s\ge\sqrt{\varepsilon N}$. For $s=0$, the final output is independent of $\phi$, giving heavy probability $1/2$ by the marginal median calculation above.
\end{proof}

\end{document}